%% file: Main.tex
\documentclass[reqno]{amsart}
\usepackage[margin=1.5in,bottom=1.25in]{geometry}		

\usepackage{amsmath}		
\usepackage{amssymb}		
\usepackage{amsfonts}		
\usepackage{amsthm}		
\usepackage[foot]{amsaddr}		

\usepackage{mathtools}		

\mathtoolsset{%
}

\usepackage[utf8]{inputenc}		
\usepackage[T1]{fontenc}		

\usepackage[
cal=cm,
]
{mathalfa}

\usepackage{dsfont}		

\usepackage[sf,mono=false]{libertine}

\usepackage{acronym}		
\renewcommand*{\aclabelfont}[1]{\acsfont{#1}}		
\newcommand{\acli}[1]{\emph{\acl{#1}}}		
\newcommand{\acdef}[1]{\emph{\acl{#1}} \textup{(\acs{#1})}\acused{#1}}		

\usepackage[labelfont={bf,small},labelsep=colon,font=small]{caption}	
\usepackage[dvipsnames,svgnames]{xcolor}		
\colorlet{MyRed}{DarkViolet}
\colorlet{MyGreen}{DarkGreen!80!Black}
\colorlet{MyBlue}{MediumBlue}

\newcommand{\afterhead}{.\;}		
\newcommand{\para}[1]{\medskip\paragraph{\textbf{#1\afterhead}}}

\usepackage{latexsym}		

\usepackage{subcaption}		
\usepackage{tikz}		
\usetikzlibrary{calc,patterns}		

\usepackage{array}		
\usepackage{booktabs}		
\usepackage{tabularx}

\usepackage[inline,shortlabels]{enumitem}		
\setenumerate{itemsep=\smallskipamount,topsep=\smallskipamount,left=\parindent}
\setitemize{itemsep=\smallskipamount,topsep=\smallskipamount,left=\parindent}

\usepackage[kerning=true]{microtype}		

\usepackage{lipsum}		
\usepackage{xspace}		

\usepackage[sort&compress]{natbib}		

\bibpunct[, ]{[}{]}{,}{n}{,}{,}

\usepackage{hyperref}
\hypersetup{
colorlinks=true,
linktocpage=true,
pdfstartview=FitH,
breaklinks=true,
pdfpagemode=UseNone,
pageanchor=true,
pdfpagemode=UseOutlines,
plainpages=false,
bookmarksnumbered,
bookmarksopen=false,
bookmarksopenlevel=1,
hypertexnames=false,
pdfhighlight=/O,
urlcolor=MyBlue,linkcolor=MyBlue,citecolor=MyBlue,	
pdftitle={},
pdfauthor={},
pdfsubject={},
pdfkeywords={},
pdfcreator={pdfLaTeX},
pdfproducer={LaTeX with hyperref}
}

\newcommand{\EMAIL}[1]{\email{\href{mailto:#1}{#1}}}

\usepackage[sort&compress,capitalize,nameinlink]{cleveref}		
\crefname{assumption}{Assumption}{Assumptions}
\crefname{model}{Model}{Models}

\usepackage{algorithm}		
\usepackage{algpseudocode}		

\usepackage{thmtools}		
\usepackage{thm-restate}		

\theoremstyle{plain}
\newtheorem{corollary}{Corollary}		
\newtheorem{lemma}{Lemma}		
\newtheorem{proposition}{Proposition}		

\newtheorem*{corollary*}{Corollary}		

\theoremstyle{definition}
\newtheorem{definition}{Definition}		
\newtheorem{model}{Model}		

\newtheorem*{definition*}{Definition}		
\newtheorem*{assumption*}{Assumptions}		
\newtheorem*{example*}{Example}		

\theoremstyle{remark}
\newtheorem{remark}{Remark}		

\newtheorem*{remark*}{Remark}		

\def\endenv{\hfill\P\smallskip}

\newcounter{proofpart}

\numberwithin{example}{section}		

\usepackage[showdeletions,suppress]{color-edits}		
\usepackage[normalem]{ulem}		

\newcommand{\debug}[1]{#1}		

\newcommand{\revise}[1]{#1}		
\newcommand{\remove}[1]{}		

\def\beginrev{\color{black}}		
\def\endedit{\color{black}}		

\newcommand{\newmacro}[2]{\newcommand{#1}{\debug{#2}}}		
\newcommand{\newop}[2]{\DeclareMathOperator{#1}{\debug{#2}}}		

\DeclarePairedDelimiter{\braces}{\{}{\}}		
\DeclarePairedDelimiter{\bracks}{[}{]}		
\DeclarePairedDelimiter{\parens}{(}{)}		

\DeclarePairedDelimiter{\abs}{\lvert}{\rvert}		

\DeclarePairedDelimiterX{\setdef}[2]{\{}{\}}{#1:#2}		
\DeclarePairedDelimiterXPP{\exclude}[1]{\mathopen{}\setminus}{\{}{\}}{}{#1}

\newcommand{\N}{\mathbb{N}}		
\newcommand{\R}{\mathbb{R}}		

\DeclareMathOperator*{\argmin}{arg\,min}		
\DeclareMathOperator*{\intersect}{\bigcap}		
\DeclareMathOperator*{\union}{\bigcup}		

\DeclareMathOperator{\bigoh}{\mathcal{O}}		
\DeclareMathOperator{\dist}{dist}		
\DeclareMathOperator{\Jac}{Jac}		
\DeclareMathOperator{\one}{\mathds{1}}		
\DeclareMathOperator{\proj}{proj}		
\DeclareMathOperator{\relint}{ri}		
\DeclareMathOperator{\supp}{supp}		
\DeclareMathOperator{\unif}{Unif}		

\newcommand{\cf}{cf.\xspace}		
\newcommand{\eg}{e.g.,\xspace}		
\newcommand{\ie}{i.e.,\xspace}		

\newcommand{\textpar}[1]{\textup(#1\textup)}		

\newcommand{\alt}[1]{#1'}		
\newcommand{\altalt}[1]{#1''}		

\newmacro{\dd}{\:d}		
\newcommand{\eps}{\varepsilon}		
\newcommand{\pd}{\partial}		

\newcommand{\insum}{\sum\nolimits}		

\newmacro{\const}{c}		
\newmacro{\coef}{\lambda}		
\newmacro{\param}{\theta}		
\newmacro{\params}{\Theta}		

\newmacro{\pexp}{p}		
\newmacro{\qexp}{q}		
\newmacro{\rexp}{r}		

\newmacro{\beforestart}{0}		
\newmacro{\start}{1}		
\newmacro{\afterstart}{2}		
\newmacro{\running}{\start,\afterstart,\dotsc}		

\newmacro{\run}{n}		
\newmacro{\runalt}{k}		
\newmacro{\runaltalt}{m}		
\newmacro{\runoff}{m}		
\newmacro{\nRuns}{T}		
\newmacro{\runs}{\mathcal{\nRuns}}		

\newmacro{\state}{\policy}		
\newmacro{\statealt}{Y}		
\newmacro{\statealtalt}{Z}		

\newcommand{\beforeinit}[1][\state]{\debug{#1}_{\beforestart}}		
\newcommand{\init}[1][\state]{\debug{#1}_{\start}}		

\newcommand{\beforeiter}[1][\state]{\debug{#1}_{\runalt-1}}		
\newcommand{\iter}[1][\state]{\debug{#1}_{\runalt}}		
\newcommand{\afteriter}[1][\state]{\debug{#1}_{\runalt+1}}		

\newcommand{\beforeprev}[1][\state]{\debug{#1}_{\run-2}}		
\newcommand{\prev}[1][\state]{\debug{#1}_{\run-1}}		
\newcommand{\curr}[1][\state]{\debug{#1}_{\run}}		
\renewcommand{\next}[1][\state]{\debug{#1}_{\run+1}}		

\newmacro{\tstart}{0}		
\newmacro{\tafterstart}{1}
\renewcommand{\time}{\debug{t}}		
\newmacro{\timealt}{s}		
\newmacro{\horizon}{T}		

\newmacro{\traj}{\tau}		
\newmacro{\trajalt}{y}		
\newmacro{\trajaltalt}{z}		

\newmacro{\flowmap}{\Theta}		
\DeclarePairedDelimiterXPP{\flowof}[2]{\flowmap_{#1}}{(}{)}{}{#2}		

\newop{\Nash}{NE}		
\newop{\CE}{CE}		
\newop{\CCE}{CCE}		
\newop{\NI}{NI}		

\newop{\brep}{br}		
\newop{\reg}{Reg}		
\newop{\preg}{\overline{Reg}}		
\newop{\val}{val}		

\newmacro{\play}{i}		
\newmacro{\playalt}{j}		
\newmacro{\playaltalt}{k}		
\newmacro{\nPlayers}{N}		
\newmacro{\players}{\mathcal{\nPlayers}}		

\newmacro{\pure}{\alpha}		
\newmacro{\purealt}{\beta}		
\newmacro{\purealtalt}{\gamma}		
\newmacro{\pureq}{\eq[\pure]}		
\newmacro{\nPures}{A}		
\newmacro{\pures}{\mathcal{\nPures}}		

\newmacro{\strat}{x}		
\newmacro{\stratalt}{\alt\strat}		
\newmacro{\strataltalt}{\altalt\strat}		
\newmacro{\strats}{\mathcal{X}}		
\newmacro{\intstrats}{\strats^{\circle}}		

\newmacro{\score}{y}		
\newmacro{\scorealt}{\alt\score}		

\newcommand{\eq}{\sol}		

\newmacro{\loss}{\ell}		
\newmacro{\pay}{u}		
\newmacro{\payv}{v}		
\newmacro{\pot}{f}		

\newmacro{\game}{\mathcal{G}}		
\newmacro{\gameall}{\game(\players,\points,\loss)}		

\newmacro{\fingame}{\Gamma}		
\newmacro{\fingameall}{\Gamma(\players,\pures,\pay)}		

\newmacro{\gmat}{g}		
\newmacro{\gdist}{\dist_{\gmat}}
\newmacro{\mfld}{M}		
\newmacro{\form}{\omega}		

\newmacro{\tvec}{z}		
\newmacro{\uvec}{u}		

\newmacro{\ball}{\mathbb{B}}		
\newmacro{\sphere}{\mathbb{S}}		

\newmacro{\vertex}{v}		
\newmacro{\vertexalt}{\alt\vertex}		
\newmacro{\vertexaltalt}{\altalt\vertex}		
\newmacro{\nVertices}{V}		
\newmacro{\vertices}{\mathcal{\nVertices}}		

\newmacro{\edge}{e}		
\newmacro{\edgealt}{\alt\edge}		
\newmacro{\edgealtalt}{\altalt\edge}		
\newmacro{\nEdges}{E}		
\newmacro{\edges}{\mathcal{\nEdges}}		

\newmacro{\graph}{\mathcal{G}}		
\newmacro{\graphall}{\graph(\vertices,\edges)}		

\newmacro{\vecspace}{\mathcal{V}}		
\newmacro{\subspace}{\mathcal{W}}		

\newmacro{\bvec}{e}		
\newmacro{\bvecs}{\mathcal{E}}		
\newmacro{\vvec}{v}		

\newmacro{\coord}{i}		
\newmacro{\coordalt}{j}		
\newmacro{\coordaltalt}{k}		
\newmacro{\nCoords}{d}		
\newmacro{\dims}{\nCoords}		
\newmacro{\vdim}{\nCoords}		

\newmacro{\pspace}{\mathcal{X}}		
\newmacro{\dspace}{\prod_{\play}\R^{\pures_{\play}\times\states}}		

\newmacro{\ppoint}{x}		
\newmacro{\ppointalt}{\alt\ppoint}		
\newmacro{\ppointaltalt}{\altalt\ppoint}		
\newmacro{\ppoints}{\mathcal{X}}		
\newmacro{\pstate}{X}		

\newmacro{\dpoint}{y}		
\newmacro{\dpointalt}{\alt\dpoint}		
\newmacro{\dpointaltalt}{\altalt\dpoint}		
\newmacro{\dpoints}{\mathcal{Y}}		
\newmacro{\dstate}{y}		

\newmacro{\pvec}{u}		
\newmacro{\dvec}{v}		

\newmacro{\mat}{M}		
\newmacro{\hmat}{H}		

\newmacro{\ones}{\mathbf{1}}		
\newmacro{\eye}{I}		
\newmacro{\zer}{\mathbf{0}}		

\DeclarePairedDelimiter{\norm}{\lVert}{\rVert}		
\DeclarePairedDelimiterXPP{\dnorm}[1]{}{\lVert}{\rVert}{}{#1}		

\DeclarePairedDelimiterXPP{\onenorm}[1]{}{\lVert}{\rVert}{_{1}}{#1}		
\DeclarePairedDelimiterXPP{\twonorm}[1]{}{\lVert}{\rVert}{_{2}}{#1}		
\DeclarePairedDelimiterXPP{\supnorm}[1]{}{\lVert}{\rVert}{_{\infty}}{#1}		

\DeclarePairedDelimiterX{\braket}[2]{\langle}{\rangle}{#1,#2}		

\DeclarePairedDelimiterX{\inner}[2]{\langle}{\rangle}{#1,#2}		

\newcommand{\defeq}{\coloneqq}		

\newcommand{\from}{\colon}		
\newcommand{\comp}[1]{#1^{\mathtt{c}}}		

\newmacro{\source}{O}		
\newmacro{\sink}{D}		

\newmacro{\pair}{i}		
\newmacro{\pairalt}{j}		
\newmacro{\pairaltalt}{k}		
\newmacro{\nPairs}{N}		
\newmacro{\pairs}{\mathcal{\nPairs}}		

\newmacro{\route}{p}		
\newmacro{\routealt}{\alt\route}		
\newmacro{\routealtalt}{\altalt\route}		
\newmacro{\nRoutes}{P}		
\newmacro{\routes}{\mathcal{\nRoutes}}		

\newmacro{\flow}{f}		
\newmacro{\flowalt}{\alt\flow}		
\newmacro{\flowaltalt}{\altalt\flow}		
\newmacro{\flows}{\mathcal{F}}		

\newmacro{\load}{x}		
\newmacro{\loadalt}{\alt\load}		
\newmacro{\loadaltalt}{\altalt\load}		
\newmacro{\loads}{\mathcal{X}}		

\newop{\Opt}{Opt}		
\newop{\Sol}{Sol}		
\newop{\gap}{Gap}		
\newop{\orcl}{Or}		

\newmacro{\obj}{f}		
\newmacro{\objalt}{g}		
\newmacro{\sobj}{F}		

\newmacro{\gvec}{g}		
\newmacro{\oper}{A}		
\newmacro{\vecfield}{v}		

\newcommand{\sol}[1][\policy]{#1^{\ast}}		

\newmacro{\vbound}{G}		
\newmacro{\lips}{L}		
\newmacro{\strong}{\mu}		
\newmacro{\smooth}{\beta}		

\newop{\tspace}{T}		
\newop{\tcone}{TC}		
\newop{\dcone}{\tcone^{\ast}}		
\newop{\ncone}{NC}		
\newop{\pcone}{PC}		
\newop{\hull}{\Delta}		

\newmacro{\cvx}{\mathcal{C}}		
\newmacro{\subd}{\partial}		

\newmacro{\minmax}{L}		

\newmacro{\minvar}{\theta}		
\newmacro{\minvaralt}{\alt\minvar}		
\newmacro{\minvars}{\Theta}		

\newmacro{\maxvar}{\phi}		
\newmacro{\maxvaralt}{\alt\maxvar}		
\newmacro{\maxvars}{\Phi}		

\newop{\Eucl}{\Pi}		
\newop{\logit}{\Lambda}		
\newop{\dkl}{KL}		

\newmacro{\hreg}{h}		
\newmacro{\breg}{D}		
\newmacro{\mprox}{P}		
\newmacro{\mirror}{Q}		
\newmacro{\fench}{F}		
\newmacro{\hstr}{K}		
\newmacro{\depth}{H}		
\newmacro{\proxdom}{\points^{\hreg}}		

\DeclarePairedDelimiterXPP{\proxof}[2]{\mprox_{#1}}{(}{)}{}{#2}		

\newmacro{\zone}{\mathbb{D}}		

\newmacro{\point}{x}		
\newmacro{\pointalt}{\alt\point}		
\newmacro{\pointaltalt}{\altalt\point}		
\newmacro{\points}{\mathcal{K}}		
\newmacro{\intpoints}{\relint\points}		

\newmacro{\base}{p}		
\newmacro{\basealt}{q}		
\newmacro{\basealtalt}{u}		

\newmacro{\open}{\mathcal{U}}		
\newmacro{\closed}{\mathcal{C}}		
\newmacro{\cpt}{\mathcal{K}}		
\newmacro{\nhd}{\mathcal{U}}		

\newop{\ex}{\mathbb{E}}		
\newop{\probMatrix}{P} 
\newop{\prob}{\mathbb{P}}		
\newop{\Var}{Var}		
\newop{\simplex}{\hull}		

\providecommand\given{}		

\DeclarePairedDelimiterXPP{\exof}[1]{\ex}{[}{]}{}{
\renewcommand\given{\nonscript\,\delimsize\vert\nonscript\,\mathopen{}} #1}

\DeclarePairedDelimiterXPP{\probof}[1]{\prob}{(}{)}{}{
\renewcommand\given{\nonscript\:\delimsize\vert\nonscript\:\mathopen{}} #1}

\DeclarePairedDelimiterXPP{\tprobof}[1]{P}{(}{)}{}{
\renewcommand\given{\nonscript\:\delimsize\vert\nonscript\:\mathopen{}} #1}

\DeclarePairedDelimiterXPP{\oneof}[1]{\one}{\{}{\}}{}{
\renewcommand\given{\nonscript\,\delimsize\vert\nonscript\,\mathopen{}} #1}

\newmacro{\sample}{\omega}		
\newmacro{\samples}{\Omega}		

\newmacro{\filter}{\mathcal{F}}		
\newmacro{\probspace}{(\samples,\filter,\prob)}		

\newcommand{\as}{\debug{\textpar{a.s.}}\xspace}		
\newmacro{\event}{\mathcal{E}}       
\newmacro{\eventalt}{H}       
\newmacro{\mean}{\mu}		
\newmacro{\sdev}{\sigma}		
\newmacro{\variance}{\sdev^{2}}		

\newcommand{\est}[1]{\hat #1}		

\newmacro{\signal}{\est\gvalue}		
\newmacro{\step}{\gamma}		
\newmacro{\learn}{\eta}		
\newmacro{\mix}{\eps}		
\newmacro{\conf}{\delta}		
\newmacro{\mindiff}{d^{\ast}}		

\newmacro{\proper}{\tau}		

\newmacro{\noise}{U}		
\newmacro{\snoise}{\xi}		
\newmacro{\noisepar}{\sdev}		
\newmacro{\noisevar}{\variance}		
\newmacro{\aggnoise}{\mathrm{\uppercase\expandafter{\romannumeral1}}}		
\newmacro{\supnoise}{\aggnoise_{\infty}}		
\newmacro{\maxnoise}{\aggnoise^{\ast}}		

\newmacro{\bias}{b}		
\newmacro{\bbound}{B}		
\newmacro{\sbias}{\chi}		
\newmacro{\aggbias}{\mathrm{\uppercase\expandafter{\romannumeral2}}}		
\newmacro{\supbias}{\aggbias_{\infty}}		
\newmacro{\maxbias}{\aggbias^{\ast}}		

\newmacro{\second}{\psi}		
\newmacro{\sbound}{M}		
\newmacro{\aggsecond}{\mathrm{\uppercase\expandafter{\romannumeral3}}}		
\newmacro{\supsecond}{\aggsecond_{\infty}}		
\newmacro{\maxsecond}{\aggsecond^{\ast}}		

\newmacro{\mart}{M}     
\newmacro{\submart}{S}      
\newmacro{\toterr}{W}      

\addauthor[Kyriakos]{KL}{Red}

\addauthor[Man]{MV}{Magenta}

\newmacro{\this}{\mathbb{A}}
\newmacro{\that}{\mathbb{B}}
\newmacro{\visitMatrix}{\mathcal{T}}
\newmacro{\transMatrix}{\mathcal{P}}
\newmacro{\pertrubation}{\mathrm{pert}}
\usepackage[skins]{tcolorbox}
\usepackage{soul}

\addauthor[Pan]{PM}{MediumBlue}

\newcommand{\explain}[1]{\hspace{2em}\tag*{\itshape\#\:#1}}     

\newcommand{\reinforce}{\textsc{\debug{Reinforce}}\xspace}
\newcommand{\greedy}{\textsc{\debug{$\eps$-Greedy Policy Gradient}}\xspace}
\newcommand{\MDP}{\textsc{\debug{MDP}}\xspace}

\newcommand{\adv}{\mathrm{adv}}		
\newcommand{\avgadv}{\cramped{\overline{\adv}}}		

\newmacro{\bexp}{{\ell_{\bias}}}		
\newmacro{\sexp}{{\ell_{\sdev}}}		

\newmacro{\dbasin}{\mathcal{W}}		
\newmacro{\basin}{\mathcal{B}}		

\newmacro{\energy}{D}		
\newmacro{\poldiam}{\norm{\policies}}       
\newmacro{\radius}{\varrho}       
\newmacro{\thres}{a}		
\newmacro{\Const}{C}        

\addauthor[Angel]{AG}{DarkOrange}

\newmacro{\nStates}{S} 
\newmacro{\states}{\mathcal{\nStates}} 
\newmacro{\rstate}{s} 
\newmacro{\rstatealt}{\alt\rstate} 

\newmacro{\rewards}{R} 
\newmacro{\reward}{r} 

\newmacro{\rvalue}{V} 
\newmacro{\gvalue}{v} 
\newmacro{\logvalue}{\hat{v}} 

\newmacro{\discount}{\gamma} 
\newmacro{\expar}{\varepsilon} 
\newmacro{\stateDist}{\rho}
\newmacro{\tprob}{P} 
\newmacro{\realReward}{r}
\newmacro{\policy}{\pi} 
\newmacro{\policyalt}{\alt\policy} 
\newmacro{\policies}{\Pi} 
\newmacro{\estpolicy}{\hat{\pi}} 
\newmacro{\np}{\pi^*} 

\newmacro{\mdp}{\mathcal{G}} 
\newmacro{\mismatch}{\mathcal{C}_\mdp}
\newmacro{\action}{\pure}
\newmacro{\actions}{\pures}

\newmacro{\stopTime}{T} 
\newmacro{\stopPr}{\zeta}

\newmacro{\logtrick}{\Lambda}

\newcommand{\visitDistr}[2]{{d}_{#1}^{#2}}
\newcommand{\visitMeas}[2]{\tilde{d}_{#1}^{#2}}
\newcommand{\normConst}[2]{{Z}_{#1}^{#2}}
\newcommand{\dpol}[1][\policy]{\tilde{d}^{#1}} 
\newmacro{\rad}{\beta} 
\newmacro{\bD}{\bar\energy}

\begin{document}


\title
[On the convergence of policy gradient methods]
{On the convergence of policy gradient methods\\
to Nash equilibria in general stochastic games}

\author
[A.~Giannou]
{Angeliki Giannou$^{\ast}$}
\address{$^{\ast}$\,%
University of Wisconsin-Madison.}
\EMAIL{giannou@wisc.edu}
\author
[K.~Lotidis]
{Kyriakos Lotidis$^{\ddag}$}
\address{$^{\ddag}$\,%
Department of Management Science \& Engineering, Stanford University.}
\EMAIL{klotidis@stanford.edu}
\author
[P.~Mertikopoulos]
{Panayotis Mertikopoulos$^{\diamond,\star}$}
\address{$^{\diamond}$\,%
Univ. Grenoble Alpes, CNRS, Inria, Grenoble INP, LIG, 38000 Grenoble, France.}
\address{$^{\star}$\,%
Criteo AI Lab.}
\EMAIL{panayotis.mertikopoulos@imag.fr}
\author
[E.~V.~Vlatakis-Gkaragkounis]
{Emmanouil V. Vlatakis-Gkaragkounis$^{\S}$}
\address{$^{\S}$\,%
University of California, Berkeley.}
\EMAIL{emvlatakis@berkeley.edu}



\newacro{LHS}{left-hand side}
\newacro{RHS}{right-hand side}
\newacro{iid}[i.i.d.]{independent and identically distributed}
\newacro{lsc}[l.s.c.]{lower semi-continuous}
\newacro{NE}{Nash equilibrium}
\newacroplural{NE}[NE]{Nash equilibria}

\newacro{whp}[w.h.p.]{with high probability}
\newacro{wp1}[w.p.1]{with probability $1$}

\newacro{MDP}{Markov decision process}
\newacroplural{MDP}[MDPs]{Markov decision processes}
\newacro{GDP}{gradient dominance property}
\newacroplural{GDP}[GDPs]{gradient dominance properties}
\newacro{FOS}{first-order stationary}
\newacro{SOS}{second-order stationary}
\newacro{PG}{policy gradient}
\newacro{LPG}{lazy policy gradient}
\newacro{VI}{variational inequality}		
\newacroplural{VI}[VIs]{variational inequalities}

\begin{abstract}
\input{Abstract.tex}
\end{abstract}

\maketitle
\allowdisplaybreaks		
\acresetall		
\addtocontents{toc}{\protect\setcounter{tocdepth}{0}}

\section{Introduction}
\label{sec:introduction}
\input{Introduction}

\section{Preliminaries}
\label{sec:preliminaries}
\input{Preliminaries}

\section{Policy gradient methods}
\label{sec:model}
\input{Model}

\section{Convergence analysis and results}
\label{sec:results}
\input{Results}

\section{Concluding remarks}
\label{sec:conclusion}
\input{Conclusion}

\numberwithin{lemma}{section}		
\numberwithin{corollary}{section}		
\numberwithin{proposition}{section}		
\numberwithin{equation}{section}		
\appendix


\section{Asymptotic convergence to stable Nash policies}
\label{app:asymptotic}
\input{App-Asymptotic}

\section{Rate of convergence to \acl{SOS} policies}
\label{app:SOS}
\input{App-SOS}

\section{Rate of convergence to strict Nash policies}
\label{app:deterministic}
\input{App-Deterministic}

\section{Structural properties of policy gradient methods}
\label{app:marl}
\input{App-MARL}

\section{Statistics of \reinforce}
\label{app:reinforce}
\input{App-Statistics}

\section{Solution concepts}
\label{app:solutions}
\input{App-Solutions}
\section*{Acknowledgments}
\begingroup
\small
\input{Thanks}
\endgroup

\bibliographystyle{icml}
\bibliography{bibtex/IEEEabrv.bib,bibtex/localbibliography.bib,bibtex/Bibliography-MV.bib,bibtex/Bibliography-PM.bib}

\end{document}

%% file: Abstract.tex
%
%
Learning in stochastic games is a notoriously difficult problem because, in addition to each other's strategic decisions, the players must also contend with the fact that the game itself evolves over time, possibly in a very complicated manner.
Because of this, the convergence properties of popular learning algorithms \textendash\ like policy gradient and its variants \textendash\ are poorly understood, except in specific classes of games (such as potential or two-player, zero-sum games).
In view of this, we examine the long-run behavior of policy gradient methods with respect to Nash equilibrium policies that are \acf{SOS} in a sense similar to the type of sufficiency conditions used in optimization.
Our first result is that \ac{SOS} policies are locally attracting with high probability, and we show that policy gradient trajectories with gradient estimates provided by the \reinforce algorithm achieve an $\mathcal{O}(1/\sqrt{\run})$ distance-squared convergence rate if the method's step-size is chosen appropriately.
Subsequently, specializing to the class of \emph{deterministic} Nash policies, we show that this rate can be improved dramatically and, in fact, policy gradient methods converge within a \emph{finite} number of iterations in that case.

%% file: Introduction.tex

Ever since they were introduced by \citet{Sha53} in the 1950's, stochastic games have been one of the staples of non-cooperative game theory, with a range of pioneering applications to
multi-agent reinforcement learning \cite{SLB08},
unmanned vehicles \cite{shalev2016safe},
general game-playing \cite{Silver17:Mastering,Vinyals19:Grandmaster,Mora17:DeepStack},
etc.
Informally, a stochastic game unfolds in discrete time as follows:
At each point in time, the players are at a given state which determines the rules of the game for that stage.
The actions of the players in this state determine not only their instantaneous payoffs (as defined by the stage game), but also the transition probabilities towards the next state of the process.
In this way, each player has to balance two distinct \textendash\ and often competing \textendash\ objectives:
optimizing the payoffs of \emph{today} versus picking a possibly suboptimal action which could yield significant benefits \emph{tomorrow} (\ie by influencing the transitions of the process towards a more favorable state for the player).

Since all players in the game are involved in a similar dilemma, the decision-making problem for each player is a very complicated affair.
In particular, in addition to their changing strategic decisions, the players of the game must also contend with the fact that the stage game itself evolves over time.
Because of this, even the existence of a Nash equilibrium policy \textendash\ viz. a stationary Markovian policy that is stable to unilateral deviations \cite{fink1964equilibrium} \textendash\ is far more difficult to prove compared to standard, stateless normal form games;
for a comprehensive survey, \cf \cite{NS03,solan2015stochastic} and references therein.

The question we seek to address in this paper is whether an ensemble of boundedly rational players can reach an equilibrium policy in a stochastic game.
Specifically, if players do not have sufficient information \textendash\ or the computational resources required \textendash\ to solve a high-dimensional Bellman equation \cite{FilVri97,SB98}, it is not at all clear if they would somehow end up playing a Nash policy in the long run.
After all, the complexity of most games increases exponentially with the number of players, so the identification of a game's equilibria quickly becomes prohibitively difficult \cite{Yujia22:The}.

\para{Our contributions in the context of related work}

This issue has sparked a vigorous literature with important ramifications for the range of applications mentioned above.
Nevertheless, these efforts must grapple with a series of strong lower bounds for computing even weaker solution concepts like coarse correlated equilibria in turn-based stochastic games \cite{daskalakis2022complexity, Yujia22:The}.
On that account, a recent line of work has focused on establishing convergence in \emph{specific} subclasses of stochastic games, such as \emph{min-max} \cite{LPX20,daskalakis2020independent,Wei21:Last,Sayin20:Fictitious,Sayin21:Decentralized,Cen21:Fast} and common interest \emph{potential} games \cite{LOPP22,ding2022independent,zhang2021gradient}.
However, despite these encouraging results, the general case remains particularly elusive.

Our paper takes a complementary approach to the above and seeks to study the convergence landscape of a class of \emph{equilibrium policies} \textendash\ not \emph{games}.
For concreteness, we focus on the general class of policy gradient methods as pioneered by \cite{sutton1999policy,williams1992simple,kakade2001natural,konda1999actor}, and we examine the methods' convergence properties in general random stopping games \textendash\ as opposed to ergodic stochastic games with an infinite horizon \cite{Per13,LPX20}.
Concretely, this means that the sequence of play evolves episode-by-episode:
within each episode, the players commit a policy and play the game, and from one episode to the next, they use an iterative gradient step to update their policy and continue playing.

Our main contributions in this general context may be summarized as follows:
\begin{enumerate}
\item
We introduce a flexible algorithmic template for the analysis of policy gradient methods which accounts for different information and update frameworks \textendash\ from perfect policy gradients to value-based estimates obtained on a per-episode basis, \eg via the \reinforce algorithm \cite{williams1992simple,sutton1999policy,baxter2001infinite}.
\item
Within this framework, we show that Nash policies that satisfy a certain strategic stability condition are locally attracting with arbitrarily high probability.
Moreover, to estimate the method's rate of convergence, we focus on Nash policies that satisfy a second-order sufficiency condition similar to the type of sufficiency conditions used in optimization, and we show that such policies enjoy an $\bigoh(1/\sqrt{\run})$ squared distance convergence rate.
\item
Finally, we also consider the method's convergence to \emph{deterministic} Nash policies \textendash\ a special case of \ac{SOS} policies \textendash\ and we show that, generically, the above rate can be improved dramatically.
In particular, by a simple tweak to the method's projection step, the induced sequence of play converges to equilibrium in a \emph{finite} number of iterations, despite all the noise and uncertainty.
\end{enumerate}

It is also worth noting that our analysis focuses squarely on the actual, episode-by-episode trajectory of play, not any ``best-iterate'' or time-averaged variant thereof.
\beginrev
In regards to the latter class of guarantees, the recent work of \citet{JLWY21} proposed an algorithm (called V-learning) which updates the policy $\policy_{\run}$ of the $\run$-th episode based on the observed rewards so far.
Thanks to the algorithm's regret guarantees, \citet{JLWY21} showed that
\begin{enumerate*}
[(\itshape a\upshape)]
\item
in min-max games, the time-averaged policy $\bar\state_{\run} = (1/\run) \sum_{\runalt=\start}^{\run} \iter$ converges to equilibrium at a rate of $\bigoh(1/\sqrt{\run})$;
whereas
\item
in \emph{general} stochastic games, the empirical frequency of play converges to the game's set of coarse correlated equilibria (a substantial relaxation of the notion of Nash equilibrium) at a rate of $\bigoh(1/\sqrt{\run})$.
\end{enumerate*}

By contrast, as we mentioned above, our paper focuses on the \emph{actual} sequence of play, \ie the policy $\curr$ employed at each episode of the game.
Moreover, the rates that we obtain all concern the convergence of the players' policies to a \emph{Nash} equilibrium \textendash\ not a correlated equilibrium or other relaxation thereof.
In this regard, the best-iterate / ergodic convergence rates are incomparable to our own as they concern a weaker type of convergence (time-averaged instead of the actual sequence), and to a weaker solution concept (correlated equilibria instead of Nash equilibria).
\endedit
This aspect of our results is especially relevant for multi-agent reinforcement learning scenarios where agents learn ``on the fly'', and it has important ramifications for many of the practical applications of stochastic games.

From a technical standpoint, our analysis is based on mapping the problem of multi-agent policy learning to the problem of equilibrium learning in a class of continuous games characterized by the fact that first-order stationary points are necessarily Nash (itself a consequence of the so-called ``gradient dominance'' property of stochastic games).
By means of this reframing, we are able to leverage a series of recent techniques for establishing local convergence in (non-monotone) continuous games and \aclp{VI} \citep{RBS16,CRMB20,MZ19,HIMM19,HIMM20,AIMM21}, which ultimately also yield convergence in our setting.
As a result, even though the unbounded variance of the \reinforce estimator is a source of considerable complications, the resulting link between stochastic and continuous games is of particular technical interest because it opens up a wide array of stochastic approximation tools and techniques that can be used for the analysis of multi-agent learning in stochastic games.

%% file: Preliminaries.tex

\subsection{Setup of the game}
Throughout this work we consider $\nPlayers$-player generic stochastic games where players repeatedly select actions in a shared \ac{MDP} with the goal of maximizing their individual value functions.
Formally, we study the tabular version with random stopping of general stochastic games, which is specified by a tuple $\game = (\states, \players , \braces{\pures_{\play}, \rewards_{\play}}_{\play \in \players} , \tprob, \stopPr,\stateDist)$ with the following primitives:
\begin{itemize}
\item
A finite set of \emph{agents} $\play \in \players = \{1,2,\dots, \nPlayers\}$ and a finite set of \emph{states} $\states = \{1,\dotsc,\nStates\}$.
\item
For each $\play\in\players$, a finite space of \emph{actions} (or \emph{pure strategies}) $\pures_{\play}$ indexed by $\pure_{\play} = 1,\dotsc,\nPures_{\play} = \abs{\pures_{\play}}$.
We will write $\pures = \prod_{\play \in \players} \pures_{\play}$ and $\pures_{-\play} = \prod_{j\neq i} \pures_{j}$ for the action space of all agents and that of all agents other than $\play$ respectively.
In a similar vein, we will also write $\pure = (\pure_{\play},\pure_{-\play})$ when we want to highlight the action $\pure_{\play}$ of player $\play$ against the action profile $\pure_{-\play}$ of $\play$'s opponents.
\item
For each $\play\in\players$, we will write $\rewards_{\play} \from \states\times \pures \to [-1,1]$ for the \emph{
reward function} of agent $\play \in \players$, \ie $\rewards_{\play}(\rstate,\pure_{\play},\pure_{-\play})$ will denote the value of the 
reward of agent $\play$ when the game is at state $\rstate\in\states$, the focal agent $\play\in\players$ plays $\pure_{\play}\in\pures_{\play}$, and all other agents take actions $\pure_{-\play} \in \pures_{-\play}$.
\item
The game transits from one state to another according to a Markov transition process, so that $\tprobof{\rstatealt \given \rstate, \pure}$ denotes the probability of transitioning from $\rstate$ to $\rstate'$ when $\pure \in \pures$ is the action profile chosen by the agents.
\item
Given an action profile $\pure$ at state $\rstate$, the process terminates with probability $\stopPr_{\rstate,\action} > 0$, i.e., $\stopPr_{\rstate,\action} = 1 - \sum_{\rstatealt \in \states}\tprobof{\rstatealt \given \rstate, \pure}$; for convenience, we will write $\stopPr \defeq \min_{\rstate,\action}\braces{\stopPr_{\rstate,\action}}$.
\item
$\stateDist\in \simplex(\states)$ is the distribution for the initial state of the game.
\end{itemize}

\para{Episodic Setting} We consider an episodic setting, where in each episode a realization of the game is completed. At every time step $\time \ge \tstart$ of each episode, 
all agents observe the common state $\rstate_{\time} \in \states$, select actions $\pure_{\time}$ and receive rewards $\{\rewards_{\play}(\rstate_\time,\pure_\time)\}_{\play \in \players}$. Then, with probability $\stopPr_{\rstate_\time,\pure_\time}$ the game terminates, and with probability $1-\stopPr_{\rstate_\time,\pure_\time}$, it moves to the state $\rstate_{\time+1}$, which is drawn according to $\tprob( \cdot | \rstate_{\time}, \pure_{\time})$. 
Denoting the realized reward of player $\play$ at time $\time$ as $\realReward_{\play,\time} \defeq \rewards_{\play}(\rstate_{\time}, \pure_{\time})$, we will write $\traj = (\rstate_{\time}, \pure_{\time}, \realReward_{\time})_{\time \leq\stopTime(\traj)}$ to denote the trajectory of the episode, where $\realReward_{\time} \defeq (\realReward_{\play,\time})_{ \play \in \players}$, and $\stopTime(\traj)$ the time the episode terminates.

\para{Policies and value functions}
We consider \emph{stationary Markovian} policies, i.e., policies that do not depend on the time-step and the history, given the current state of the game. More specifically, for each agent $\play \in \players$, a \emph{policy} $\policy_{\play} \from \states \to \simplex(\pures_{\play})$ specifies a probability distribution over the actions of agent $\play$ in state $\rstate \in \states$, i.e., $\pure_{\play} \sim \policy_{\play}\parens{\cdot \vert \rstate}$ denotes the (random) action drawn by agent $\play$ at state $\rstate \in \states$ according to $\policy_{\play}$, viewed here as an element of $\policies_{\play} \defeq \simplex(\pures_{\play})^{\states}$.
In addition, we will also write
$\policy = (\policy_{\play})_{\play \in \players} \in \policies \defeq \prod_{\play} \policies_{\play}$
and
$\policy_{-\play} = (\policy_\playalt)_{\playalt\neq\play} \in \policies_{-\play} \defeq \prod_{\playalt\neq\play} \policies_{\playalt}$ for the policy profile of all agents and all agents other than $\play$, respectively.

The expected reward of agent $\play\in\players$ if agents follow policy $\policy$, starting from initial state $\rstate\in\states$, defines the \emph{value function} of agent $\play$, denoted as $\rvalue_{\play,\rstate}\parens{\policy}$, and is equal to
\begin{equation}
\rvalue_{\play,\rstate}\parens{\policy}
	\defeq\ex_{\traj \sim \MDP} \bracks*{\insum_{\time =\tstart}^{\stopTime\parens{\traj}} \rewards_{\play}\parens{\rstate_{\time},\action_\time} \Big\vert \rstate_\tstart = \rstate}
\end{equation}
where $\traj \sim \MDP$ denotes the randomness induced by the policy profile $\policy$, and the state-transition probabilities of the MDP. Overloading the notation, we set $  \rvalue_{\play,\stateDist}\parens{\policy}\defeq\ex_{\rstate\sim\stateDist}\bracks*{   \rvalue_{\play,\rstate}\parens{\policy}}$. 
Although value functions are, in general, non-convex, they share similar smoothness properties with the payoff functions of normal form games, namely bounded and Lipschitz gradients.
For precise statements, we defer to the paper's supplement.

\para{Visitation distribution and the mismatch coefficient}

For a policy profile $\policy \in \policies$ and an arbitrary initial state distribution $\rstate_\tstart\sim\stateDist$, we define the discounted state visitation measure/distribution as
\begin{equation*}
		\visitMeas{\stateDist}{\policy}\parens{\rstate} =\ex_{\traj \sim \MDP}\bracks*{\insum_{\time =\tstart}^{\stopTime\parens{\traj}}\one\braces{\rstate_{\time}=\rstate} \Big| \rstate_\tstart \sim \stateDist},\quad	\visitDistr{\stateDist}{\policy}\parens{\rstate} \defeq \visitMeas{\stateDist}{\policy}\parens{\rstate}/\normConst{\stateDist}{\policy}
\end{equation*}
In the appendix, we prove formally that the above definition is well-posed for the random stopping episodic framework described above, \ie
$\visitMeas{\stateDist}{\policy}\parens{\rstate}<\infty$,
so
$\normConst{\stateDist}{\policy}\defeq\sum_{\rstate\in\states }\visitMeas{\stateDist}{\policy}\parens{\rstate}$
is well-defined.
In our proofs, we will leverage a standard property of visitation distributions, namely the equivalence of the expected value of state-action function and the expected cumulative value over a random trajectory.
More precisely, we have:
\begin{restatable}{lemma}{ConversionLemma}\emph{[Conversion Lemma]}
\label{lem:conversion}
For an arbitrary state-action function $f\from\states\times \pures \to \R$, a policy profile $\policy$ and an initial state distribution $\rstate_\tstart\sim\stateDist$, we have
\begin{equation}
  \ex_{\traj \sim \MDP}\bracks*{\insum_{\time =\tstart}^{\stopTime\parens{\traj}}f( \rstate_\time, \pure_\time)} 
  = \normConst{\stateDist}{\policy}
  \ex_{\rstate\sim\visitDistr{\stateDist}{\policy}}\ex_{\pure\sim\policy(\cdot \mid \rstate)}\bracks*{f(\rstate, \pure)}
\end{equation}
\end{restatable}
Finally, to quantify the difficulty of hard-to-reach states via a policy gradient method, we will follow the standard approach of \cite{zhang2021gradient,fan2020theoretical,chen2019information,munos2005error,munos2003error} and use an appropriately-defined distribution ``mismatch coefficient'', generalizing the single-agent counterpart of \citet{AgarwalKLM20}.
More precisely, for a  stochastic game $\game$, we define the \emph{mismatch coefficient} as $\mismatch\defeq\max_{\policy,\policy'\in\policies}\braces[\big]{\supnorm{\visitMeas{\stateDist}{\policy}/\visitMeas{\stateDist}{\policy'}}}$ or, more simply, as $\mismatch\defeq\max_{\policy,\in\policies}\braces[\big]{\tfrac{1}{\stopPr}\supnorm{\visitDistr{\stateDist}{\policy}/\stateDist}}$.
Similar to prior work in this direction \cite{AgarwalKLM20,bhandari2019global,daskalakis2020independent}, we will assume $\mismatch$ is finite, which, equivalently, means that $\visitDistr{\stateDist}{\policy}(\rstate)>0$ for any policy $\policy$ and state $\rstate$.

\subsection{Solution concepts}

The most widely used solution concept in game theory is that of a Nash
equilibrium \ie a strategy profile  $\eq\in\policies$ that discourages unilateral deviations.
However, in stochastic games, the definition of a Nash policy is much more involved because of the existence of multiple states and steps, \cf \cite{Sha53,takahashi1962stochastic,fink1964equilibrium,solan2015stochastic} and references therein.
Formally, we have:

\begin{definition}
[Nash policies]
\label{def:NP}
A policy $\eq = (\eq_{\play})_{\play \in \players} \in \policies$ is said to be a \emph{Nash policy} for a given distribution of initial states $\stateDist \in\simplex\parens{\states}$ if, for every player $\play\in\players$, we have
\begin{equation}
\label{eq:NP}
\tag{NE}
\rvalue_{\play,\stateDist}(\eq_{\play};\eq_{-\play}) \geq \rvalue_{\play,\stateDist}(\policy_{\play};\eq_{-\play})\
	\quad
	\text{for all $\play\in\players$ and all $\policy_{\play} \in \simplex(\pures_{\play})^{\states}$}.
\end{equation}
\end{definition}


In contrast to general non-convex continuous games, stochastic games satsify a version of the well-known Polyak-Łojasiewicz condition \cite{polyak1963}
but with linear gradient growth,
also known as a \acdef{GDP} \cite{AgarwalKLM20,bhandari2019global}.
For the multi-agent case, \citet{zhang2021gradient} and \citet{daskalakis2020independent} showed that a similar property holds even in the episodic setting:

\begin{restatable}[\Acl{GDP}]{lemma}{GradientDominance}
\label{lem:GDP}
For any policy profile $\policy=(\policy_{\play})_{\play\in\players}\in\policies$, we have that \begin{equation}
\label{eq:GDP}
\tag{GDP}
\rvalue_{\play,\stateDist}(\policyalt_{\play};\policy_{-\play}) -
	\rvalue_{\play,\stateDist}(\policy_{\play};\policy_{-\play})
	\leq \mismatch \max_{\bar\policy_{\play}\in\policies_{\play}}
		\braket{\nabla_{\play}\rvalue_{\play,\stateDist}(\policy)}{\bar{\policy}_{\play}-{\policy_{\play}}}
\end{equation}
for any unilateral deviation $\policyalt_{\play}\in\policies_{\play}$ of player $\play\in\players$.
\end{restatable}
\begin{remark*}
In the above and throughout our paper, we will write $\nabla_{\play}$ to denote the gradient of the quantity in question with respect to $\policy_{\play}$, \ie when $\policy_{-\play}$ is kept fixed and only $\policy_{\play}$ is varied. 
For concision, we will write $\gvalue_{\play}(\policy) = \nabla_{\play} \rvalue_{\play,\stateDist}(\policy)$ for the individual gradient of player $\play$'s value function, and $\gvalue(\policy) = (\gvalue_{\play}(\policy))_{\play\in\players}$ for the ensemble thereof.
\endenv
\end{remark*}

Thanks to \eqref{eq:GDP}, it is straightforward to check that \ac{FOS} points of $\rvalue$ are Nash.
Formally, as in \cite{daskalakis2020independent,LOPP22,zhang2021gradient}, we have the following characterization:

\begin{restatable}[\Acl{FOS} policies are Nash]{lemma}{FOSNE}
\label{lem:fos2ne}
A policy $\eq = (\eq_{\play})_{\play \in \players} \in \policies$ is Nash if and only if it satisfies the \acl{FOS} condition
\begin{equation}
\label{eq:FOS}
\tag{FOS}
\braket{\gvalue(\eq)}{\policy - \eq}
	\leq 0
	\quad
	\text{for all $\policy\in\policies$}.
\end{equation}
\end{restatable}

\citet{LOPP22} and \citet{zhang2021gradient} proved a relaxation of the above lemma to the effect that policies that satisfy \eqref{eq:FOS} up to $\eps$ (\ie in lieu of $0$ in the \acs{RHS}) are $\bigoh(\eps)$-Nash.
Going in the other direction, we will consider the following series of refinements of Nash policies which are particularly important from a learning standpoint \cite{LRS19,solan2015stochastic}:

\begin{definition}
\label{def:refine}
Let $\eq = (\eq_{\play})_{\play\in\players} \in \policies$ be a Nash policy.
We then say that:
\begin{itemize}
\item
$\eq$ is \emph{stable} if $\braket{\gvalue(\policy)}{\policy - \eq} < 0$ for all $\policy\neq\eq$ sufficiently close to $\eq$.
\item
$\eq$ is \acli{SOS} if it satisfies the sufficiency condition
\begin{equation}
\label{eq:SOS}
\tag{SOS}
\parens{\policy - \eq}^{\top} \Jac_{\gvalue}(\eq) \parens{\policy - \eq}
	< 0
	\quad
	\text{for al $\policy\in\policies\exclude{\eq}$},
\end{equation}
where $\Jac_{\gvalue}(\eq) = (\nabla_{\playalt} \gvalue_{\play}(\eq))_{\play,\playalt\in\players} = (\nabla_{\playalt} \nabla_{\play} \rvalue_{\play}(\eq))_{\play,\playalt\in\players}$ denotes the Jacobian of $\gvalue$ at $\eq$.
\item
$\eq$ is \emph{deterministic} if it induces a deterministic selection rule $\eq_{\play}\from\states\to\pures_{\play}$ for all $\play\in\players$.
\item
$\eq$ is \emph{strict} if it is deterministic and \eqref{eq:FOS} holds as a strict inequality whenever $\policy\neq\eq$.
\end{itemize}
\end{definition}

\beginrev
\begin{remark}
In the above and what follows, ``sufficiently close'' means that there exists a neighborhood $\nhd$ of $\eq$ in $\policies$ such that the stated inequality holds for all $\policy\in\nhd$.
Unless mentioned otherwise, we will measure distances on $\policies$ relative to the Euclidean norm, but this choice does not impact our results.
\end{remark}
\endedit

Intuitively,
the condition for equilibrium stability is a game-theoretic analogue of first-order KKT sufficiency condition,
while
the condition for second-order stationarity is the second-order version thereof.
In this regard, the distinction between first-order stationary, stable and second-order stationary points is formally analogous to the distinction between critical points, minimizers, and second-order minimum points in optimization.
As for deterministic policies, we should mention that, generically, deterministic policies are also strict, so we will use the two terms interchangeably.%
\footnote{The notion of genericity is stated here in the sense of Baire, \ie the stated property holds for all but a ``meager'' set of games (\ie a countable union of nowhere dense sets in the space of all games).}

Importantly, as we show in \cref{app:solutions}, these refinements admit the following characterizations:

\begin{restatable}{proposition}{StationaryProperties}
\label{prop:var}
Let $\eq = (\eq_{\play})_{\play\in\players} \in \policies$ be a Nash policy.
Then:
\begin{subequations}
\begin{enumerate}
[\itshape a\upshape)]
\item
If $\eq$ is \acl{SOS}, there exists some $\strong>0$ such that
\begin{alignat}{2}
\label{eq:strong}
\braket{\gvalue(\policy)}{\policy - \eq}
	&\leq - \strong \, \norm{\policy - \eq}^{2}
	&\qquad
	&\text{for all $\policy$ sufficiently close to $\eq$}.
\shortintertext{%
\item
If $\eq$ is strict, there exists some $\strong>0$ such that
} 
\label{eq:sharp}
\braket{\gvalue(\policy)}{\policy - \eq}
	&\leq - \strong \, \norm{\policy - \eq}
	&\qquad
	&\text{for all $\policy$ sufficiently close to $\eq$}.
\end{alignat}
\end{enumerate}
\end{subequations}
\end{restatable}

In view of all the above, we get the following string of implications for equilibria in generic games:
\begin{equation}
\textrm{strict}
	/ \textrm{deterministic}
	\implies \textrm{\ac{SOS}}
	\implies \textrm{stable}
	\implies \textrm{FOS}
	= \textrm{Nash}
\end{equation}
For posterity, we should clarify here that, due to the highly complicated structure of the game's value functions, it is not trivial to construct a concrete example where \eqref{eq:strong} holds but \eqref{eq:sharp} does not.
Examples of strict Nash policies abound in the literature \cite{LRS19,solan2015stochastic}, but we are not otherwise aware of an argument that could be used to close the gap between \eqref{eq:strong} and \eqref{eq:sharp}.
In view of this, our analysis will treat both cases concurrently (with the obvious anticipation that more refined solution concepts should enjoy stronger convergence guarantees).

%% file: Model.tex

We now proceed to describe our general model for episodic learning in stochastic games.
To that end, we will consider a framework where agents follow a specific policy $\curr$ within each episode, and update it from one episode to the next with the objective of increasing their individual rewards.
Formally, our approach will adhere to the following inter-episode sequence of events:
\begin{enumerate}
\item
At the beginning of each episode $\run=\running$, every agent $\play\in\players$ chooses a policy $\policy_{\play,\run} \in \policies_{\play}$.
\item
Within the $\run$-th episode, each player executes their chosen policy $\policy_{\play,\run}$, inducing in this way an intra-episode trajectory of play $\traj_{\run} = (\rstate^{(\run)}_{\time}, \pure^{(\run)}_{\time}, \reward^{(\run)}_{\time})_{\time \leq \stopTime(\traj_{\run})}$.
\item
Once the episode terminates, agents update their policies and the process repeats.
\end{enumerate}

In terms of feedback, we will treat several models, depending on what type of information is available to the agents during play.
More precisely, we will focus on the generic \ac{PG} template
\begin{equation}
\label{eq:PG}
\tag{PG}
\next
	= \proj_{\policies}(\curr + \curr[\step] \curr[\signal])
\end{equation}
where:
\begin{enumerate}
\item
$\curr = (\state_{\play,\run})_{\play\in\players} \in \policies$ denotes the player's policy profile at each episode $\run=\running$
\item
$\curr[\signal] = (\signal_{\play,\run})_{\play\in\players} \in \dspace$ is an estimate for the agents' inidividual policy gradients.
\item
$\proj_{\policies}\from \dspace \to \policies$ denotes the Euclidean projection to the agents' policy space $\policies$.
\item
$\curr[\step] > 0$ is the method's step-size,
for which we will assume throughout that $\sum_{\run} \curr[\step] = \infty$;
typically, \eqref{eq:PG} is run with a step-size of the form $\curr[\step] = \step/(\run + \runoff)^{\pexp}$ for some $\step>0$, $\runoff\geq0$ and $\pexp\geq0$.
\end{enumerate}

Regarding the gradient signal $\curr[\signal]$, we will decompose it as
\begin{equation}
\label{eq:signal}
\curr[\signal]
	= \vecfield(\curr) + \curr[\noise] + \curr[\bias]
\end{equation}
where
\begin{equation}
\label{eq:bias-noise}
\curr[\noise]
	= \curr[\signal] - \exof{\curr[\signal] \given \curr[\filter]}
	\quad
	\text{and}
	\quad
\curr[\bias]
	= \exof{\curr[\signal] \given \curr[\filter]} - \vecfield(\curr).
\end{equation}
In the above, we treat $\curr$, $\run=\running$, as a stochastic process on some complete probability space $\probspace$, and we write $\curr[\filter] \defeq \filter(\init,\dotsc,\curr) \subseteq \filter$ for the history (adapted filtration) of $\curr$ up to \textendash\ and including \textendash\ stage $\run$.
By definition, $\exof{\curr[\noise] \given \curr[\filter]} = 0$ and $\curr[\bias]$ is $\curr[\filter]$-measurable, so $\curr[\noise]$ can be intepreted as a random, zero-mean error relative to $\vecfield(\curr)$, whereas $\curr[\bias]$ captures all systematic (non-zero-mean) errors.
To make this precise, we will further assume that $\curr[\bias]$ and  $\curr[\noise]$  are bounded as
\begin{equation}
\label{eq:errorbounds}
\exof{\dnorm{\curr[\bias]} \given \curr[\filter]}
	\leq \curr[\bbound]
	\qquad
	\text{ and }
	\qquad
\exof{\dnorm{\curr[\noise]}^{2} \given \curr[\filter]}
	\leq \curr[\sdev]^{2}
\end{equation}
where the sequences $\curr[\bbound]$ and $\curr[\sdev]$ 
, $\run=\running$, are to be construed as deterministic upper bounds on the bias, fluctuations, and magnitude of the gradient signal $\curr[\signal]$.

Depending on these bounds, a gradient signal with $\curr[\bbound]=0$ will be called \emph{unbiased},
and an unbiased signal with $\curr[\sdev] = 0$  will be called \emph{perfect}.
More generally, we will assume that the above statistics are bounded as
\begin{equation}
\label{eq:schedule}
\curr[\bbound]
	= \bigoh(1/\run^{\bexp})
	\qquad
	\text{and}
	\qquad
\curr[\sdev]
	= \bigoh(\run^{\sexp})
\end{equation}
for some $\bexp,\sexp>0$ which depend on the specific model under consideration.
For concreteness, we describe below three basic models that adhere to the above template for $\curr[\signal]$ in order of decreasing information requirements:

\begin{model}
[Full gradient information]
\label{mod:full}
The first model we will consider assumes that agents observe their \emph{full policy gradients}, \ie
\begin{equation}
\label{eq:signal-full}
\curr[\signal]
	= \gvalue(\curr)
\end{equation}
implying in particular that $\curr[\noise] = \curr[\bias] = 0$.
This model is fully deterministic across episodes (though intra-episode play remains stochastic).
In particular, it tacitly assumes that agents know the game (and can observe their opponents' policies) so as to calculate the full gradients of their individual value functions $\rvalue_{\play,\stateDist}$, \cf \cite{zhang2021gradient,agarwal2021theory,LOPP22} and references therein.
\endenv
\end{model}

\begin{model}
[Learning with stochastic gradients]
\label{mod:stoch}
A relaxation of the above model which is particularly relevant for applications to deep reinforcement learning concerns the case where the player have access to stochastic policy gradients \cite{zhang2020global}, \ie unbiased gradient estimates of the form
\begin{equation}
\label{eq:signal-stoch}
\curr[\signal]
	= \gvalue(\curr)
	+ \curr[\noise]
\end{equation}
with $\exof{\curr[\noise] \given \curr[\filter]} = 0$ (so we can formally take $\bexp = \infty$ and $\sexp = 0$ in \cref{eq:schedule}
 above).
\endenv
\end{model}


\begin{figure*}[t]
\begin{minipage}{0.49\textwidth}
\begin{algorithm}[H]
\caption{\reinforce}
\label{alg:reinforce}
\begin{algorithmic}[1]
	\State 
	\textbf{Input:} $\estpolicy \in \policies, \traj = (\rstate_{\time}, \pure_{\time}, \reward_{\time})_{\time \le \stopTime(\traj)} \in\mathcal{T}$\vskip 3pt
	\For{$\play = 1,\hdots, \nPlayers$}\vskip 3pt
	\State 
		$\rewards_\play\parens{\traj} \leftarrow \sum_{\time=\tstart}^{\stopTime\parens{\traj}}\reward_{\play,\time}$\vskip 3pt
	\State 
		$\logtrick_\play\parens{\traj} \leftarrow \sum_{\time=\tstart}^{\stopTime\parens{\traj}} \nabla_{\play} \parens*{\log\estpolicy_\play\parens{\action_{\play,\time}|\rstate_\time}}$\vskip 3pt
	\State $\logvalue_\play \leftarrow \rewards_\play\parens{\traj}\cdot\logtrick_\play\parens{\traj}$\vskip 3pt
	\EndFor\vskip 3pt
	\State
	\Return $\braces{\logvalue_\play}_{\play\in\players}$
	\end{algorithmic}
\end{algorithm}
\end{minipage}
\hfill
\begin{minipage}{0.49\textwidth}
\begin{algorithm}[H]
\caption{\greedy}
\label{alg:greedy}
\begin{algorithmic}[1]
	\State
	\textbf{Input:} $\policy_{\start}, \{\step_{\run}\}_{\run \in \N}, \{\expar_{\run}\}_{\run \in \N}$\vskip 3pt
	\For{$\run = \start,\afterstart,\dots$}\vskip 3pt
	\State
	$\estpolicy_{\run}\leftarrow (1-\expar_\run) \policy_{\run} + \frac{\expar_\run}{\abs{\pures}}$\vskip 3pt
	\State
	Sample $\traj_{\run} \sim \MDP(\estpolicy_{\run}|\rstate_\tstart)$\vskip 3pt
	\State
	$\logvalue_{\run} \leftarrow \reinforce(\estpolicy_{\run},\traj_{\run})$\vskip 3pt
	\State
	$\policy_{\run+1} \leftarrow \proj_{\policies}\parens*{\policy_{\run} + \step_{\run}\logvalue_{\run}}$\vskip 3pt
\EndFor \vskip 5pt
\end{algorithmic}
\end{algorithm}
\end{minipage}
\end{figure*}


\begin{model}
[Value-based learning]
\label{mod:value}
The last model we will consider concerns the case where agents only have access to their instantaneous rewards and need to reconstruct their individual gradients based on this information.
A widely used method to achieve this is via the \reinforce subroutine, which we describe in pseudocode form in \cref{alg:reinforce}.
In words, when employing \reinforce, each agent $\play\in\play$ commits to a sampling policy $\hat\policy_{\play}\in\policies_{\play}$ and executes it in an episode of the stochastic game in play.
Then, at the end of the episode, players gather the total reward $\rewards_\play\parens{\traj} \leftarrow \sum_{\time=\tstart}^{\stopTime\parens{\traj}}\reward_{\play,\time}$ associated to the intra-episode trajectory of play $\traj$, and they estimate their policy gradients via the so-called ``log-trick''  \cite{williams1992simple} as
\begin{equation}
\label{eq:logtrick}
\logvalue_\play
	= \rewards_\play\parens{\traj}
		\cdot \sum\nolimits_{\time=\tstart}^{\stopTime\parens{\traj}} \nabla_{\play}
			\parens*{\log\estpolicy_\play\parens{\action_{\play,\time}|\rstate_\time}}.
\end{equation}
\Cref{lem:reinforce} below provides the vital statistics of the \reinforce estimator:

\begin{restatable}{lemma}{ReinforceLemma}
\label{lem:reinforce}
Suppose that each agent $\play\in\players$ follows a stationary policy $\policy_{\play} \in \policies_{\play}$.
Then:
\begin{subequations}
\label{eq:reinforce}
\begin{flalign}
\quad
\text{\itshape a\upshape)}
	\quad
	&\ex_{\traj \sim \MDP}\bracks*{\reinforce(\policy)}
		= \gvalue(\policy)
		&
	\\
\quad
\text{\itshape b\upshape)}
	\quad
	&\ex_{\traj \sim \MDP} \bracks*{\norm{\reinforce_\play(\policy)-\gvalue_\play(\policy)}^2}
		\leq \dfrac{24\nPures_{\play}}{\kappa_\play \stopPr^4}
		&
\end{flalign}
\end{subequations}
where $\kappa_\play = \min_{\rstate\in\states,\pure_{\play}\in\pures_{\play}} \policy_\play(\pure_\play \vert \rstate)$.
\end{restatable}

Therefore, if \reinforce is executed at $\hat\policy \gets \curr$ at each episode $\run=\running$, we will have
\begin{equation}
\exof{\hat\gvalue_{\play,\run}}
	= \gvalue_{\play}(\curr)
	\qquad
	\text{and}
	\qquad
\exof{\dnorm{\noise_{\play,\run}}^{2} \given \curr[\filter]}
	\leq \dfrac{24\nPures_{\play}}{\stopPr^4\min_{\rstate\in\states,\pure_{\play} \in \pures_\play}\policy_{\play,\run}(\pure_\play \vert \rstate)}.
\end{equation}
In particular, this means that we will always have $\curr[\bbound] = 0$ for the bias of the estimator, but its variance could be unbounded if $\curr$ gets close to the boundary of $\policies$.
To avoid this, \reinforce can be paired with an explicit exploration step that modifies the sampling policy of the $\run$-th episode to
\begin{equation}
\estpolicy_{\play,\run}
	= (1-\expar_{\run}) \policy_{\play,\run}
		+ \expar_{\run}\unif_{\pures_{\play}}
		\quad
		\text{\revise{for all $\rstate\in\states$}}
\end{equation}
\ie $\estpolicy_{\play,\run}$ is the mixture between $\policy_{\play,\run}$ and the uniform distribution $\unif_{\pures_{\play}}$ over $\pures_{\play}$.
The resulting algorithm is known as \greedy;
for a pseudocode representation, see \cref{alg:greedy}.

Importantly, by calling \reinforce at $\curr[\estpolicy]$ instead of $\curr[\policy]$, $\curr[\signal]$ becomes biased (because of the difference between $\curr[\estpolicy]$ and $\curr$), but its variance is bounded;
in particular, by invoking \cref{lem:reinforce}, we have
\begin{equation}
\exof{\dnorm{\bias_{\play,\run}} \given \curr[\filter]}
	\leq G\curr[\expar]
	\qquad
	\text{and}
	\qquad
\exof{\dnorm{\noise_{\play,\run}}^{2} \given \curr[\filter]}
	\leq \frac{24\nPures_{\play}^2}{\expar_{\run}\stopPr^{4}}
\end{equation}
where $G$ is a constant that depends on the smoothness of $\rvalue$ and the cardinalities of $\pures$ and $\states$.%
\footnote{Specifically, from  \cref{lem:smoothness} we know that $\norm{v_i(\hat{\pi}_n) - v_i(\pi_n)} \leq 3\sqrt{\nPures}/\zeta^3 \cdot \sum_{\playalt}\sqrt{\nPures_{\playalt}} \cdot \norm{\hat{\pi}_{j,n}-\pi_{j,n}}$. Moreover,  $|\pi_{i,n}(\pure \mid s)-\hat{\pi}_{i,n}(\pure \mid s)| \leq \varepsilon_n$ for all $\rstate\in\states, \pure \in \mathcal{A}_i$, so $\norm{\pi_{i,n} - \hat{\pi}_{i,n}} \leq \sqrt{\nStates\nPures_{\play}} \varepsilon_n$.
Combining the above, it follows that we can take $G=3\nPlayers\nPures^{3/2}\sqrt{\nStates} \big/ \zeta^3$.}
In this way, \cref{alg:greedy} can be seen as a special case of \eqref{eq:PG} with $\curr[\bbound] = \bigoh(\curr[\expar])$ and $\curr[\sdev]^{2} = \bigoh(1/\curr[\expar])$.
\endenv
\end{model}

%% file: Results.tex

We are now in a position to state and discuss our main results.
For convenience, we will present our results in order of increasing structure, starting with stable policies, and then moving on to \acl{SOS} and deterministic Nash policies.
All proofs are deferred to the appendix.

\subsection{Asymptotic convergence to stable Nash policies}

Our first convergence result concerns Nash policies that satisfy the stability requirement $\braket{\gvalue(\policy)}{\policy-\eq} < 0$ of \cref{def:refine}.
In this case, we have the following guarantee:

\begin{restatable}{theorem}{stable}
\label{thm:stable}
Let $\eq$ be a stable Nash policy, and let $\curr$ be the sequence of play generated by \eqref{eq:PG} with step-size $\curr[\step] = \step/(\run + \runoff)^{\pexp}$, $\pexp \in (1/2,1]$,
and policy gradient estimates such that $\pexp + \bexp >1$ and $\pexp - \sexp > 1/2$ as per \eqref{eq:schedule}.
Then there exists a neighborhood $\nhd$ of $\eq$ in $\policies$ such that, for any given $\conf>0$, we have
\begin{equation}
\label{eq:primal-local}
\probof{\text{$\curr$ converges to $\eq$} \given \init\in\nhd}
	\geq 1-\conf
\end{equation}
provided that $\step$ is small enough \textpar{or $\runoff$ large enough} relative to $\conf$.
\end{restatable}

\begin{restatable}{corollary}{corstable}
\label{cor:stable}
Suppose that \crefrange{mod:full}{mod:value} are run with a step-size of the form $\curr[\step] = \step/(\run+\runoff)^{\pexp}$, $\pexp>1/2$, and if applicable,
an exploration parameter $\curr[\mix] = \mix/(\run+\runoff)^{\rexp}$ such that $1-\pexp < \rexp < 2\pexp-1$.
Then:
\begin{itemize}
\item
For \cref{mod:full,mod:stoch}:
the conclusions of \cref{thm:stable} hold as stated.
\item
For \cref{mod:value}:
the conclusions of \cref{thm:stable} hold as long as $\pexp>2/3$.
\end{itemize}
\end{restatable}

We note here that \cref{thm:stable}
provides a trajectory convergence guarantee which is otherwise quite difficult to obtain even in structured stochastic games.
For example, if we zoom in on the class of stochastic potential (or min-max) games, the existing guarantees in the literature concern the ``best iterate'' of the algorithm, \cf \cite{LOPP22,zhang2021gradient} and references therein.
Because of this, said guarantees do not apply to the actual trajectory of play generated by \eqref{eq:PG};
this makes them less suitable for agent-based learning where the players involved are learning ``as they go'', as opposed to \emph{simulating} the game in order to approximately compute an equilibrium policy offline.

We should also note that the convergence guarantees of \cref{thm:stable} hold locally with arbitrarily high probability.
Without further assumptions, it is not possible to obtain global trajectory convergence guarantees that hold with probability $1$, even in single-state games \textendash\ that is, the case of learning in finite normal form games.
\beginrev
The reason for this locality is twofold:
First, equilibrium policies are not unique in general, and gradient-based dynamics may also admit non-equilibrium attractors, such as limit cycles and the like \cite{HMC21,MPP18,MLZF+19,MHC22}.
As a result, in the presence of multiple equilibria/attractors, the best one can hope for is a local equilibrium convergence result, conditioned on the basin of attraction of said equilibrium (as per \cref{thm:stable}).

The second obstruction to a global, unconditional convergence result is probabilistic in nature, and has to do with the randomness that enters the learning process (\eg in the estimation of policy gradients via the \reinforce).
In this case, no matter how close one starts to an equilibrium policy, there is always a finite, non-zero probability that an unlucky realization of the noise can drive the process away from its basin, possibly never to return.
This issue can only be overcome in games where $\policies$ is partitioned (up to a set of measure zero) into basins of attraction of equilibrium policies.
However, this can only occur in games with a sufficiently strong global structure, like potential stochastic games, two-player zero-sum games and the like;
in complete generality, locality cannot be lifted, even in single-state problems \cite{FVGL+20,GVM21}.
\endedit

\subsection{Convergence to \acl{SOS} policies}

Albeit valuable as an asymptotic convergence guarantee, \cref{thm:stable} does not provide an indication of how long it will take players to actually converge to a Nash policy.
Of course, in full generality, it is not plausible to expect to be able to derive such a convergence rate because the stability requirement provides no indication on how fast the players' policy gradients stabilize near a solution.
This kind of estimate is provided by the second-order sufficient condition \eqref{eq:SOS}, which allows us to establish sufficient control over the sequence of play as indicated by the following theorem.

\begin{restatable}{theorem}{SOS}
\label{thm:SOS}
Let $\eq$ be a \acl{SOS} policy,
let $\basin$ be a neighborhood of $\eq$ such that \eqref{eq:strong} holds on $\basin$,
and
let $\curr$ be the sequence of play generated by \eqref{eq:PG} with step-size $\curr[\step] = \step/(\run + \runoff)^{\pexp}$, $\pexp \in (1/2,1]$, and policy gradient estimates such that $\pexp + \bexp > 1$ and $\pexp - \sexp > 1/2$ as per \eqref{eq:schedule}.
Then:
\begin{enumerate}
\item
There exists a neighborhood $\nhd$ of $\eq$ in $\policies$ such that, for any confidence level $\conf>0$, the event
\begin{equation}
\label{eq:good}
\event
	= \braces{\text{$\curr \in \basin$ for all $\run=\running$}}
\end{equation}
occurs with probability
\(
\probof{\event \given \init\in\nhd} \geq 1-\conf
\)
if $\runoff$ is large enough relative to $\conf$.
\item
The sequence $\curr$ converges to $\eq$ \acl{wp1} on $\event$;
in particular, we have
\begin{equation}
\label{eq:prob-converge}
\probof{\text{$\curr$ converges to $\eq$} \given \init\in\nhd}
	\geq 1-\conf
\end{equation}
if $\runoff$ is large relative to $\conf$.
Moreover, conditioned on $\event$ and taking $\qexp = \min\{\bexp,\pexp - 2\sexp\}$, we have
\begin{equation}
\label{eq:conv-rate}
\exof{\norm{\curr - \eq}^{2} \given \event}
	=
	\begin{cases}
	\bigoh(1/\run^{2\strong\step})
		&\quad
		\text{if $\pexp=1$ \, and \, $2\strong\step<\qexp$},
		\\
	\bigoh(1/\run^{\qexp})
		&\quad
		\text{otherwise}.
	\end{cases}
\end{equation}
\end{enumerate}
\end{restatable}

\begin{restatable}{corollary}{corSOS}
\label{cor:SOS}
Suppose that \crefrange{mod:full}{mod:value} are run with a step-size of the form $\curr[\step] = \step/(\run+\runoff)^{\pexp}$, $\pexp>1/2$, and if applicable, an exploration parameter $\curr[\mix] = \mix/(\run+\runoff)^{\pexp/2}$.
Then:
\begin{itemize}
\item
For \cref{mod:full,mod:stoch}:
the conclusions of \cref{thm:SOS} hold with $\qexp=\pexp$;
in particular, \eqref{eq:conv-rate} gives an $\bigoh(1/\run)$ rate of convergence if $\pexp=1$ and $2\strong\step > \qexp$.
\item
For \cref{mod:value}:
the conclusions of \cref{thm:SOS} hold for $\pexp>2/3$ with $\qexp = \pexp/2$;
in particular, \eqref{eq:conv-rate} gives an $\bigoh(1/\sqrt{\run})$ rate of convergence if $\pexp=1$ and $2\strong\step > \qexp$.
\end{itemize}
\end{restatable}

\beginrev
\begin{remark}
Getting an explicit estimate for the constant in the $\bigoh(\cdot)$ guarantee of \cref{thm:SOS} is quite involved but, up to logarithmic and subleading factors, Chung's lemma \cite{Chu54,Pol87} can be used to show that
\begin{enumerate*}
[\itshape a\upshape)]
\item
if $2\strong\step > \qexp$, it scales as $(\Const_{\bias}+\Const_{\sigma})/[(2\strong\step - \qexp)(1-\conf)]$ where $\Const_{\bias} = \sup_{\run} \curr[\step]\curr[\bbound]$ and $\Const_{\sigma} = \sup_{\run} \curr[\step]^2\curr[\noisevar]$;
\item
if $2\strong\step = \qexp$, it scales as 
$(\Const_{\bias}+\Const_{\sigma})(1+ \max\{(2\strong\step)^2,4\strong\step\})/(1-\delta)$; 
and  
\item if $2\strong\step < \qexp$ as $(\Const_{\bias}+\Const_{\sigma})(1+ \max\{(2\strong\step)^2,4\strong\step\})/[(\qexp -2\strong\step)(1-\delta)]$.
\end{enumerate*}
\end{remark}
\endedit

Besides providing
a general framework for achieving trajectory convergence, \cref{thm:SOS} gives the rates of convergence of the sequence of play to the Nash policy in question.
In particular, with this result in hand, one can confidently argue about the distance of the iterates of \eqref{eq:PG} from equilibrium in a series of different environments.
More to the point, this convergence guarantee allows the algorithm designer to adapt the parameters of the learning process according to the complexity and limitations of the environment, a feature which further highlights the significance of this result.

We should also note the delicate interplay between the method's step-size and the achieved convergence rate.
In the case of \cref{mod:full}, \cref{cor:SOS} suggests a step-size of the form $\curr[\step] = \Theta(1/\run)$, leading to a $\bigoh(1/\run)$ convergence rate.
As we show in the appendix, this rate can be improved:
in the deterministic case with perfect gradient information, \eqref{eq:PG} with a suitably chosen constant step-size achieves a \emph{geometric} convergence rate, \ie $\norm{\curr - \eq} = \bigoh(\exp(-\rho\run))$ for some $\rho>0$ (\cf \cref{prop:geometric} in \cref{app:SOS}).
By contrast, in the case of \cref{mod:stoch}, the $\bigoh(1/\run)$ rate we provide cannot be improved, even if the quadratic minorant \eqref{eq:strong} that characterizes \ac{SOS} policies holds \emph{globally} \textendash\ and this because the learning process is running against standard lower bounds from convex optimization \cite{Nes04,Bub15}.

Perhaps the most significant guarantee from a practical point of view is the $\bigoh(1/\sqrt{\run})$ convergence rate attained in \cref{mod:value} (\cf \cref{alg:reinforce,alg:greedy}).
This guarantee amounts to a $\bigoh(1/\run^{1/4})$ convergence rate in terms of the (non-squared) distance to equilibrium which, mutatis mutandis, represents a notable improvement over the $\bigoh(1/\run^{1/6})$ guarantee of \citet{LOPP22} (expressed in norm values).
Of course, the latter guarantee is global \textendash\ because the focus of \cite{LOPP22} is stochastic \emph{potential} games \textendash\ but it also concerns the ``best iterate'' of the process (not its ``last iterate''), so the two results are not immediately comparable.
However, a useful ``best-of-both-worlds'' heuristic that can be inferred by the combination of these works is that, given a budget of training episodes, \cref{alg:greedy} can be run with a constant step-size as per \cite{LOPP22} for a sufficient fraction of this budget, and then with a $\bigoh(1/\run)$ ``cooldown'' schedule for the rest.
In this way, after an aggressive ``exploration'' phase, the algorithm's $\bigoh(1/\run^{1/4})$ rate would kick in and supply faster stabilization to an \ac{SOS} policy.

\subsection{Convergence to deterministic Nash policies}

Our last series of results concerns the rate of convergence to deterministic Nash policies in generic stochastic games.
As we discussed in \cref{sec:preliminaries},
deterministic Nash policies also satisfy  \eqref{eq:SOS}, so the rate of convergence of \eqref{eq:PG} to such policies can be harvested directly from \cref{thm:SOS}.
However, as we show below, a simple projection tweak in \eqref{eq:SOS} can improve this rate dramatically.

The tweak in question is inspired by the geometry of $\policies$ around a deterministic policy:
by definition, such policies are corner points of $\policies$, so any consistent drift towards them will cause $\curr$ to hit the boundary of $\policies$ in finite time.
Of course, under \eqref{eq:PG}, the process may rebound from the boundary and return to the interior of $\policies$ if the policy gradient estimate is not particularly good at a given iteration of the algorithm.
However, if we replace the projection step of \eqref{eq:PG} with a ``lazy projection'' in the spirit of \citet{Zin03}, the aggregation of gradient steps will eventually push the process far inside the normal cone of $\policies$ at $\eq$, so rebounds of this type can no longer occur.

Formally, we will consider the following \acdef{LPG} scheme:
\begin{equation}
\label{eq:LPG}
\tag{\ac{LPG}}
\next[\dstate]
	= \curr[\dstate] + \curr[\step]\curr[\signal]
	\qquad
\next[\policy]
	= \proj_{\policies}(\next[\dstate])
\end{equation}
where $\curr[\dstate] = (\dstate_{\play,\run})_{\play\in\players} \in \dspace$ is an auxiliary variable that maintains an aggregate of gradient steps \emph{before} projecting them back to $\policies$.
We then have the following convergence result:

\begin{restatable}{theorem}{strict}
\label{thm:strict}
Let $\curr$ be the sequence of play under \eqref{eq:LPG} with step-size and policy gradient estimates such that $\pexp + \bexp >1$ and $\pexp - \sexp > 1/2$ as per \eqref{eq:schedule}.
If $\eq$ is a deterministic Nash policy, there exists an unbounded open set $\dbasin\subseteq\dspace$ of initializations such that, for any $\conf>0$, we have
\begin{equation}
\label{eq:lazy}
\probof{\text{$\curr$ converges to $\eq$} \given \init[\dstate] \in \dbasin}
	\geq 1-\conf,
\end{equation}
provided that $\step>0$ is small enough.
Moreover, conditioned on this event, $\curr$ converges to $\eq$ at a finite number of iterations, \ie there exists some $\run_{0}$ such that $\curr = \eq$ for all $\run\geq\run_{0}$.
\end{restatable}

\begin{corollary}
\label{cor:strict}
Suppose that \crefrange{mod:full}{mod:value} are run with parameters $\curr[\step] = \step/\run^{\pexp}$, $\pexp\in(1/2,1]$, and if applicable, $\curr[\mix] = \mix/\run^{\rexp}$ with $1-\pexp < \rexp < 2\pexp-1$.
Then the conclusions of \cref{thm:strict} hold.
\end{corollary}

\beginrev
\begin{remark}
Getting an explicit bound for $\run_{0}$ is quite complicated, but the last part of the proof of \cref{thm:strict} shows that $\run_{0}$ scales in terms of the parameters of the game and the algorithm as $n_0 = \bigoh\parens[\Big]{\parens[\big]{\frac{M\nStates\nPures}{\const\step}}^{1/(1-\pexp)}}$ where $\const>0$ measures the minimum payoff difference between equilibrium and non-equilibrium strategies at $\eq$, $M$ is a measure of the initial distance from $\eq$, and $\nStates$ and $\nPures$ is the number of states and pure strategies respectively.
\end{remark}
\endedit

\Cref{thm:strict} \textendash\ and, by extension, \cref{cor:strict} \textendash\ are fairly unique because they provide a guarantee for convergence to an \emph{exact} \acl{NE} in a \emph{finite} number of iterations.
To the best of our knowledge, the only comparable result in the literature is that of \cite{zhang2021gradient}, where the authors provide a finite-time convergence guarantee to strict equilibria with \emph{perfect} policy gradients (as per \cref{mod:full}).
The result of \citet{zhang2021gradient} echoes the convergence properties of deterministic first-order algorithms around sharp minima of convex functions \cite{Pol87}, but the fact that \cref{thm:strict} applies to models with \emph{stochastic} gradient feedback of \emph{unbounded} variance (\cref{mod:stoch,mod:value} respectively) is a major difference.
As far as we are aware, this is the first guarantee of its kind in the literature on learning in stochastic games.

%% file: Conclusion.tex

A key roadblock encountered by practical applications of multi-agent reinforcement learning is the lack of universal equilibrium convergence guarantees.
While the impossibility results of \cite{HMC00,HMC03} imply that unconditional convergence is not a reasonable aspiration without further assumptions on the game, the existence of local convergence results mitigates this deficiency as it provides a range of theoretically grounded stability and runtime guarantees.
In this regard, deterministic policies acquire particular importance, as the convergence of policy gradient methods is especially rapid and robust and this case.
\beginrev
Of course, this leaves open the question of non-tabular settings and parametrically encoded policies, e.g., as in the case of deep reinforcement learning;
we defer these investigations to future work.

Another open issue of high practical relevance concerns policy gradient methods that do not rely on Euclidean projections to $\policies$.
In the single-state case (\ie learning in finite normal form games), the use of methods relying on softmax choice / exponential weights is very widely used because of its regret guarantees.
Whether the use of similar softmax techniques can lead to finer convergence guarantees in the context of general stochastic games is an important and intriguing question for future research.
\endedit

%% file: App-Asymptotic.tex

Our goal in this appendix is to prove \cref{thm:stable,cor:stable}, which we restate below for convenience:

\stable*

\corstable*

Our proof strategy will comprise the following basic steps:

\begin{enumerate}

\item
To begin with, we will show that the squared distance
\begin{equation}
\label{eq:energy}
\energy(\policy)
	= \frac{1}{2}\norm{\policy - \eq}^{2}
\end{equation}
can be seen as a ``local Lyapunov function'' for \eqref{eq:PG} in the sense that it is locally decreasing near $\eq$, up to a series of error terms \textendash\ both zero-mean and non-zero-mean.

\item
Due to these errors, the evolution of the iterates $\curr[\energy] \defeq \energy(\curr)$ of $\energy$ over time may exhibit \emph{significant} jumps:
in particular, a single ``bad'' realization of the noise could carry $\curr$ out of the basin of attraction of $\eq$, possibly never to return.
To exclude this event, our second step will be to show that the aggregation of these errors can be controlled with probability at least $1-\conf$.

\item
Conditioned on the above, we will show that, with probability at least $1-\conf$, the iterates $\curr[\energy]$ cannot grow more than a token value.
As a result, if \eqref{eq:PG} is initialized close to $\eq$, it will remain in a neighborhood thereof for all $\run$ (again, with probability at least $1-\conf$).

\item
Thanks to this ``stochastic Lyapunov stability'' result, we employ a series of martingale limit theory arguments to extract a subsequence converging to $\eq$.

\item
Finally, we show that the increments of $\curr[\energy]$ are summable;
hence, by invoking the Gladyshev's lemma \cite[p.~49]{Pol87}, we conclude that $\curr[\energy]$ converges to some (finite) random variable $\energy_{\infty}$.
Combining this fact with the existence of a convergent subsequence, we obtain the desired conclusion that $\curr$ converges to $\eq$ with probability at least $1-\conf$.
\end{enumerate}

In the sequel, we make the above precise in a series of intermediate results.

\subsection{Energy inequality}

We begin by establishing a ``quasi-Lyapunov'' inequality for the iterates $\curr[\energy] = \norm{\curr - \eq}^{2}/2$ of \eqref{eq:energy}.

\begin{lemma}
\label{lem:energy}
Let $\curr[\energy] \defeq \energy(\curr)$.
Then, for all $\run = \running$, we have
\begin{equation}
\label{eq:template}
\next[\energy]
	\leq \curr[\energy]
		+ \curr[\step] \braket{\gvalue(\curr)}{\curr - \eq}
		+ \curr[\step] \curr[\snoise]
		+ \curr[\step] \curr[\sbias]
		+ \curr[\step]^{2} \curr[\second]^{2},
\end{equation}
where the error terms $\curr[\snoise]$, $\curr[\sbias]$, and $\curr[\second]$ are given by
\begin{equation}
\label{eq:errors}
\curr[\snoise]
	= \braket{\curr[\noise]}{\curr - \eq},
	\quad
\curr[\sbias]
	= \poldiam \curr[\bbound]
	\quad
	\text{and}
	\quad
\curr[\second]^{2}
	= \tfrac{1}{2} \dnorm{\curr[\signal]}^{2}.
\end{equation}
with $\poldiam \defeq \max_{\policy,\policyalt\in\policies} \norm{\policy - \policyalt}$.
\end{lemma}

\begin{proof}
By the definition of the iterates of \eqref{eq:PG}, we have
\begin{align}
\next[\energy]
	= \frac{1}{2} \norm{\next - \eq}^{2}
	&= \frac{1}{2} \norm{\proj_{\policies}(\curr + \curr[\step]\curr[\signal]) - \proj_{\policies}(\eq)}^{2}
	\notag\\
	&\leq \frac{1}{2} \norm{\curr + \curr[\step]\curr[\signal] - \eq}^{2}
	\notag\\
	&= \frac{1}{2} \norm{\curr - \eq}^{2}
		+ \curr[\step] \braket{\curr[\signal]}{\curr - \eq}
		+ \frac{1}{2} \curr[\step]^{2} \norm{\curr[\signal]}^{2}
	\notag\\
	&= \curr[\energy]
		+ \curr[\step] \braket{\gvalue(\curr) + \curr[\noise] + \curr[\bias]}{\curr - \eq}
		+ \frac{1}{2} \curr[\step]^{2} \norm{\curr[\signal]}^{2}
	\notag\\
	&\leq \curr[\energy]
		+ \curr[\step] \braket{\gvalue(\curr)}{\curr - \eq}
		+ \curr[\step] \curr[\snoise]
		+ \curr[\step] \curr[\sbias]
		+ \curr[\step]^{2} \curr[\second]^{2}
\end{align}
where we used the Cauchy-Schwarz inequality to bound the bias term as $\braket{\curr[\bias]}{\curr -\eq} \leq \dnorm{\curr[\bias]} \cdot \norm{\curr - \eq} \leq \poldiam \curr[\bbound] = \curr[\sbias]$.
\end{proof}

\subsection{Error control and stability}

The second major step in our proof (and the most challenging one from a technical standpoint) is to establish a suitable measure of control over the error increments in \eqref{eq:energy}, with the aim of showing that the process $\curr$ never leaves a neighborhood of $\eq$.

To make this idea precise, let $\basin = \setdef{\policy\in\policies}{\norm{\policy - \eq} \leq \radius}$ be a ball of radius $\radius$ based on $\eq$ in $\policies$ so that $\braket{\gvalue(\policy)}{\policy - \eq} < 0$ for all $\policy\in\basin\exclude{\eq}$ (without loss of generality, we can assume that $\basin$ is maximal in that regard).
We will then examine the event that the aggregation of the error terms in \eqref{eq:energy} is not sufficient to drive $\curr$ to escape from $\basin$.

To that end, we will begin by aggregating the errors in \eqref{eq:energy} as
\begin{equation}
\label{eq:err-agg}
\curr[\mart]
	= \sum_{\runalt=\start}^{\run} \iter[\step] \iter[\snoise]
	\quad
	\text{and}
	\quad
\curr[\submart]
	= \sum_{\runalt=\start}^{\run} \bracks{\iter[\step] \iter[\sbias] + \iter[\step]^{2} \iter[\second]^{2}}.
\end{equation}
Since $\exof{\curr[\snoise] \given \curr[\filter]} = 0$, we have $\exof{\curr[\mart] \given \curr[\filter]} = \prev[\mart]$, so $\curr[\mart]$ is a martingale;
likewise, $\exof{\curr[\submart] \given \curr[\filter]} \geq \prev[\submart]$, so $\curr[\submart]$ is a submartingale.
Then, using a technique of \citet{HIMM19} that builds on an earlier idea by \citet{MZ19}, we will also consider the ``mean square'' error process
\begin{align}
\label{eq:noise-tot}
\curr[\toterr]
	&= \curr[\mart]^{2}
		+ \curr[\submart],
\end{align}
and the associated indicator events
\begin{subequations}
\label{eq:events}
\begin{align}
\curr[\event]
	&= \braces*{\iter\in\basin \; \text{for all $\runalt=\running,\run$}}
	\quad
	\text{and}
	\quad
\curr[\eventalt]
	= \braces*{\iter[\toterr] \leq \thres \; \text{for all $\runalt = \running,\run$}},
\end{align}
\end{subequations}
where, with a fair amount of hindsight,
the error tolerance level $\thres>0$ is such that $2\thres + \sqrt{\thres} < \radius$, and we are employing the convention $\event_{0} = \eventalt_{0} = \samples$ (since every statement is true for the elements of the empty set).
We will then assume that $\init$ is initialized in a ball of radius $\sqrt{2\thres}$ centered at $\eq$, viz.
\begin{equation}
\label{eq:nhd-init}
\nhd
	= \setdef{\policy\in\policies}{\energy(\policy) \leq \thres}
	= \setdef{\policy\in\policies}{\norm{\policy-\eq}^{2}/2 \leq \thres}.
\end{equation}
With all this in hand, the key to showing that $\curr$ remains close to $\eq$ with high probability is the following conditional estimate:

\begin{lemma}
\label{lem:events}
Let $\curr$ be the sequence of play generated by \eqref{eq:PG} initialized at $\init\in\nhd$.
We then have:
\begin{enumerate}
\item
$\next[\event] \subseteq \curr[\event]$
and
$\next[\eventalt] \subseteq \curr[\eventalt]$
for all $\run=\running$
\item
$\prev[\eventalt] \subseteq \curr[\event]$
for all $\run=\running$
\item
Consider the ``bad realization'' event
\begin{align}
\label{eq:evt-bad}
\curr[\tilde\eventalt]
	\defeq \prev[\eventalt] \setminus \curr[\eventalt]
	= \braces*{\text{$\iter[\toterr] \leq \thres$ for $\runalt=\running,\run-1$ and $\curr[\toterr] > \thres$}},
\end{align}
and
let $\curr[\tilde\toterr] = \curr[\toterr] \one_{\prev[\eventalt]}$ be the cumulative error subject to the noise being ``small''.
Then we have:
\begin{equation}
\label{eq:noise-tot-cond}
\exof{\curr[\tilde\toterr]}
	\leq \exof{\prev[\tilde\toterr]}
		+ \curr[\step] \poldiam \curr[\bbound]
		+ \curr[\step]^{2} \poldiam^{2} \curr[\sdev]^{2}
		+ \tfrac{3}{2}  \curr[\step]^{2} (\vbound^{2} + \curr[\bbound]^{2} + \curr[\sdev]^{2})
		- \thres \probof{\prev[\tilde\eventalt]},
\end{equation}
where, by convention,
$\tilde\eventalt_{0} = \varnothing$ and $\beforeinit[\tilde\toterr] = 0$.
\end{enumerate}
\end{lemma}

\begin{remark*}
In the above (and what follows), the notation $\one_{A}$ is used to indicate the logical indicator of an event $A\subseteq\samples$, \ie $\one_{A}(\sample) = 1$ if $\sample\in A$ and $\one_{A}(\sample) = 0$ otherwise.
\end{remark*}

The proof of \cref{lem:events} is quite technical, so we first proceed to derive an important stability result based on this estimate.

\begin{proposition}
\label{prop:contain}
Fix some confidence threshold $\conf>0$ and let $\curr$ be the sequence of play generated by \eqref{eq:PG} with step-size and policy gradient estimates as per \cref{thm:stable}.
We then have:
\begin{equation}
\label{eq:contain}
\probof{\curr[\eventalt] \given \init\in\nhd}
	\geq 1-\conf
	\quad
	\text{for all $\run=\running$}
\end{equation}
provided that $\step$ is small enough \textpar{or $\runoff$ large enough} relative to $\conf$.
\end{proposition}

\begin{proof}
We begin by bounding the probability of the ``bad realization'' event $\curr[\tilde\eventalt] = \prev[\eventalt] \setminus \curr[\eventalt]$.
Indeed, if $\init\in\nhd$, we have:
\begin{align}
\label{eq:prob-large1}
\probof{\curr[\tilde\eventalt]}
	&= \probof{\prev[\eventalt] \setminus \curr[\eventalt]}
	= \exof{\one_{\prev[\eventalt]} \times \oneof{\curr[\toterr] > \thres}}
	\leq \exof{\one_{\prev[\eventalt]} \times (\curr[\toterr] / \thres)}
	= \exof{\curr[\tilde\toterr]} / \thres
\end{align}
where, in the penultimate step, we used the fact that $\curr[\toterr] \geq 0$ (so $\oneof{\curr[\toterr]>\thres} \leq \curr[\toterr]/\thres$).
Telescoping \eqref{eq:noise-tot-cond} then yields
\begin{equation}
\label{eq:prob-large2}
\exof{\curr[\tilde\toterr]}
	\leq \exof{\beforeinit[\tilde\toterr]}
		+ \poldiam \sum_{\runalt=\start}^{\run} \iter[\step] \iter[\bbound]
		+ \sum_{\runalt=\start}^{\run} \iter[\step]^{2} \iter[\varrho]^{2}
		- \thres \sum_{\runalt=\start}^{\run} \probof{\beforeiter[\tilde\eventalt]}
\end{equation}
where we set
\begin{equation}
\curr[\varrho]^{2}
	= \poldiam^{2} \curr[\sdev]^{2} + \tfrac{3}{2}  (\vbound^{2} + \curr[\bbound]^{2} + \curr[\sdev]^{2}).
\end{equation}
Hence, combining \eqref{eq:prob-large1} and \eqref{eq:prob-large2} and invoking our stated assumptions for $\curr[\step]$, $\curr[\bbound]$ and $\curr[\sdev]$, we get
\begin{equation}
\sum_{\runalt=\start}^{\run} \probof{\iter[\tilde\eventalt]}
	\leq \frac{1}{\thres}
		\sum_{\runalt=\start}^{\run}
			\bracks{ \iter[\step] \iter[\bbound] \poldiam
				+ \iter[\step]^{2} \iter[\varrho]^{2}}
	\leq \frac{\Const}{\thres}
\end{equation}
for some $\Const \equiv \Const(\step,\runoff) > 0$ with $\lim_{\step\to0^{+}} \Const(\step,\runoff) = \lim_{\runoff\to\infty} \Const(\step,\runoff) = 0$ \revise{(since $\curr[\step] = \step / (\run+\runoff)^{\pexp}$ and $\pexp>0$)}.

Now, by choosing $\step$ sufficiently small (or $\runoff$ sufficiently large), we can ensure that $\Const/\thres < \conf$;
thus, given that the events $\iter[\tilde\eventalt]$ are disjoint for all $\runalt=\running$, we get
\(
\probof[\big]{\union_{\runalt=\start}^{\run} \iter[\tilde\eventalt]}
	= \sum_{\runalt=\start}^{\run} \probof{\iter[\tilde\eventalt]}
	\leq \conf.
\)
In turn, this implies that
\(
\probof{\curr[\eventalt]}
	= \probof[\big]{\comp{\init[\tilde\eventalt]} \cap \dotsi \cap \comp{\curr[\tilde\eventalt]}}
	\geq 1 - \conf,
\)
and our assertion follows.
\end{proof}

We conclude this appendix with the proof of our technical result on the events $\curr[\event]$ and $\curr[\eventalt]$:

\begin{proof}[Proof of \cref{lem:events}]
The first claim of the lemma is obvious.
For the second, we proceed inductively:

\begin{enumerate}[leftmargin=2.5em]
\item
For the base case $\run=\start$, we have $\init[\event] = \{\init \in \basin \} \supseteq \{\init \in \nhd \} = \samples$ (recall that $\init$ is initialized in $\nhd \subseteq \basin$).
Since $\eventalt_{0} = \samples$, our claim follows.

\item
Inductively, assume that $\prev[\eventalt] \subseteq \curr[\event]$ for some $\run\geq\start$.
To show that $\curr[\eventalt] \subseteq \next[\event]$,
suppose that $\iter[\toterr] \leq \thres$ for all $\runalt=\running,\run$.
Since $\curr[\eventalt] \subseteq \prev[\eventalt]$, this implies that $\curr[\event]$ also occurs, \ie $\iter\in\basin$ for all $\runalt=\running,\run$;
as such, it suffices to show that $\next\in\basin$.
To do so, given that $\iter\in\nhd\subseteq\basin$ for all $\runalt=\running\run$,
we readily obtain
\begin{equation}
\afteriter[\energy]
	\leq \iter[\energy]
		+ \iter[\step] \iter[\snoise]
		+ \iter[\step] \iter[\sbias]
		+ \iter[\step]^{2} \iter[\second]^{2},
	\quad
	\text{for all $\runalt = \running\run$},
\end{equation}
and hence, after telescoping over $\runalt = \running,\run$, we get
\begin{equation}
\next[\energy]
	\leq \init[\energy]
		+ \curr[\mart]
		+ \curr[\submart]
	\leq \init[\energy]
		+ \sqrt{\curr[\toterr]}
		+ \curr[\toterr]
	\leq \thres
		+ \sqrt{\thres}
		+ \thres
	= 2\thres + \sqrt{\thres}.
\end{equation}
We conclude that $\energy(\next) \leq 2\thres + \sqrt{\thres}$, \ie $\next\in\basin$, as required for the induction.
\end{enumerate}

For our third claim, note first that
\begin{align}
\label{eq:noise-tot-upd}
\curr[\toterr]
	&= (\prev[\mart] + \curr[\step]\curr[\snoise])^{2}
		+ \prev[\submart]
		+ \curr[\step] \curr[\sbias]
		+ \curr[\step]^{2} \curr[\second]^{2}
	\notag\\
	&= \prev[\toterr]
		+ 2 \curr[\step] \curr[\snoise] \prev[\mart]
		+ \curr[\step]^{2} \curr[\snoise]^{2}
		+ \curr[\step] \curr[\sbias]
		+ \curr[\step]^{2} \curr[\second]^{2},
\end{align}
so, after taking expectations, we get
\begin{align}
\exof{\curr[\toterr] \given \curr[\filter]}
	= \prev[\toterr]
		+ 2 \prev[\mart] \curr[\step] \exof{\curr[\snoise] \given \curr[\filter]}
		+ \exof{
			\curr[\step]^{2} \curr[\snoise]^{2}
			+ \curr[\step] \curr[\sbias]
			+ \curr[\step]^{2} \curr[\second]^{2}
			\given \curr[\filter] }
	\geq \prev[\toterr],
\end{align}
\ie $\curr[\toterr]$ is a submartingale.
To proceed, let $\curr[\tilde\toterr] = \curr[\toterr] \one_{\prev[\eventalt]}
$ so
\begin{align}
\label{eq:noise-tot-cond1}
\curr[\tilde\toterr]
	= \curr[\toterr] \one_{\prev[\eventalt]}
	&= \prev[\toterr] \one_{\prev[\eventalt]}
		+ (\curr[\toterr] - \prev[\toterr]) \one_{\prev[\eventalt]}
	\notag\\
	&= \prev[\toterr] \one_{\beforeprev[\eventalt]}
		- \prev[\toterr] \one_{\prev[\tilde\eventalt]}
		+ (\curr[\toterr] - \prev[\toterr]) \one_{\prev[\eventalt]},
	\notag\\
	&= \prev[\tilde\toterr]
		+ (\curr[\toterr] - \prev[\toterr]) \one_{\prev[\eventalt]}
		- \prev[\toterr] \one_{\prev[\tilde\eventalt]},
\end{align}
where we used the fact that $\prev[\eventalt] = \beforeprev[\eventalt] \setminus \prev[\tilde\eventalt]$ so $\one_{\prev[\eventalt]} = \one_{\beforeprev[\eventalt]} - \one_{\prev[\tilde\eventalt]}$ (since $\prev[\eventalt] \subseteq \beforeprev[\eventalt]$).
Then, \eqref{eq:noise-tot-upd} yields
\begin{align}
\curr[\toterr] - \prev[\toterr]
		= 2 \prev[\mart] \curr[\step] \curr[\snoise]
		+ \curr[\step]^{2} \curr[\snoise]^{2}
		+ \curr[\step] \curr[\sbias]
		+ \curr[\step]^{2} \curr[\second]^{2}
\end{align}
and hence, given that $\prev[\eventalt]$ is $\curr[\filter]$-measurable, we get:
\begin{subequations}
\begin{align}
\exof{(\curr[\toterr] - \prev[\toterr]) \one_{\prev[\eventalt]}}
	&\label{eq:noise-tot-zero}
		= 2 \exof{\curr[\step] \prev[\mart]\curr[\snoise] \one_{\prev[\eventalt]}}
	\\
	&\label{eq:noise-tot-noise}
		+ \exof{\curr[\step]^{2} \curr[\snoise]^{2} \one_{\prev[\eventalt]}}
	\\
	&\label{eq:noise-tot-signal}
		+ \exof{(\curr[\step]\curr[\sbias] + \curr[\step]^{2} \curr[\second]^{2}) \one_{\prev[\eventalt]}}.
\end{align}
\end{subequations}
However, since $\prev[\eventalt]$ and $\prev[\mart]$ are both $\curr[\filter]$-measurable, we have the following estimates:

\begin{enumerate}
\item
For the noise term in \eqref{eq:noise-tot-zero}, we have:
\begin{equation}
\exof{\prev[\mart] \curr[\snoise] \one_{\prev[\eventalt]}}
	= \exof{\prev[\mart] \one_{\prev[\eventalt]} \exof{\curr[\snoise] \given \curr[\filter]}}
	= 0.
\end{equation}

\item
The term \eqref{eq:noise-tot-noise} is where the reduction to $\prev[\eventalt]$ kicks in;
indeed, we have:
\begin{align}
\exof{\curr[\snoise]^{2} \one_{\prev[\eventalt]}}
	&= \exof{\one_{\prev[\eventalt]} \exof{ \abs{\braket{\curr - \eq}{\curr[\noise]}}^{2}
		\given \curr[\filter]} }
	\notag\\
	&\leq \exof{\one_{\prev[\eventalt]} \norm{\curr - \eq}^{2} \exof{\dnorm{\curr[\noise]}^{2}
		\given \curr[\filter]}}
	\explain{by Cauchy\textendash Schwarz}
	\\
	&\leq \poldiam^{2} \curr[\sdev]^{2}.
\end{align}

\item
Finally, for the term \eqref{eq:noise-tot-signal}, we have:
\begin{align}
\label{eq:noise-tot-term1}
\exof{\curr[\second]^{2} \one_{\prev[\eventalt]}}
	\leq \tfrac{3}{2}  \bracks{ \vbound^{2} + \curr[\bbound]^{2} + \curr[\sdev]^{2} }
\end{align}
where we used the bound $\dnorm{\gvalue(\policy)} \leq \vbound$.
Likewise, $\curr[\sbias] \one_{\prev[\eventalt]} \leq \poldiam \curr[\bbound]$, so
\begin{align}
\label{eq:noise-tot-term2}
\eqref{eq:noise-tot-signal}
	&\leq \curr[\step] \poldiam \curr[\bbound]
		+ \tfrac{3}{2}  \curr[\step]^{2} (\vbound^{2} + \curr[\bbound]^{2} + \curr[\sdev]^{2})
\end{align}
\end{enumerate}
Thus, putting together all of the above, we obtain:
\begin{equation}
\exof{(\curr[\toterr] - \prev[\toterr]) \one_{\prev[\eventalt]}}
	\leq \curr[\step] \poldiam \curr[\bbound]
		+ \curr[\step]^{2} \poldiam^{2} \curr[\sdev]^{2}
		+ \tfrac{3}{2}  \curr[\step]^{2} (\vbound^{2} + \curr[\bbound]^{2} + \curr[\sdev]^{2})
\end{equation}
Going back to \eqref{eq:noise-tot-cond1}, we have $\prev[\toterr] > \thres$ if $\prev[\tilde\eventalt]$ occurs, so the last term becomes
\begin{equation}
\label{eq:noise-tot-term3}
\exof{\prev[\toterr] \one_{\prev[\tilde\eventalt]}}
	\geq \thres \exof{\one_{\prev[\tilde\eventalt]}}
	= \thres \probof{\prev[\tilde\eventalt]}.
\end{equation}
Our claim then follows by combining \cref{eq:noise-tot-cond1,eq:noise-tot-term1,,eq:noise-tot-term2,eq:noise-tot-term3}.
\end{proof}

\subsection{Extraction of a convergent subsequence}

Our next step is to show that any realization $\curr$ of \eqref{eq:PG} that is contained in $\basin$ admits a subsequence $\state_{\run_{\runalt}}$ converging to $\eq$.

\begin{proposition}
\label{prop:subsequence}
Let $\eq$ be a stable Nash policy, and let $\curr$ be the sequence of play generated by \eqref{eq:PG} with step-size and policy gradient estimates such that $\pexp + \bexp >1$ and $\pexp - \sexp > 1/2$ as per \eqref{eq:schedule}.
Then $\curr$ admits a subsequence $\state_{\run_{\runalt}}$ that converges to $\eq$ \acl{wp1} on the event $\event = \intersect_{\run} \curr[\event] = \{\curr\in\basin \text{ for all } \run=\running\}$.
\end{proposition}

\begin{proof}
Let $\mathcal{Q} = \{ \text{$\curr\in\basin$ for all $\run$} \} \cap \{ \liminf_{\run} \norm{\curr - \eq} > 0 \}$ denote the event that $\curr$ is contained in $\basin$ but the sequence $\curr$ does not admit a subsequence converging to $\eq$.
We will show that $\probof{\mathcal{Q}} = 0$.

Indeed, assume ad absurdum that $\probof{\mathcal{Q}} > 0$.
Hence, \acl{wp1} on $\mathcal{Q}$, there exists some positive constant $\const>0$ (again, possibly random) such that $\braket{\gvalue(\curr)}{\curr - \eq} \leq -\const < 0$ for all $\run$.
Thus, going back to \eqref{eq:energy}, we get
\begin{equation}
\next[\energy]
	\leq \curr[\energy]
		- \curr[\step] \const
		+ \curr[\step] \curr[\snoise]
		+ \curr[\step] \curr[\sbias]
		+ \curr[\step]^{2} \curr[\second]^{2},
\end{equation}
so if we let $\curr[\tau] = \sum_{\runalt=\start}^{\run} \iter[\step]$ and telescope the above, we obtain the bound
\begin{equation}
\label{eq:energy-bound1}
\next[\energy]
	\leq \init[\energy]
		- \curr[\tau] \bracks*{
			\const
			- \frac{\curr[\mart]}{\curr[\tau]}
			- \frac{\curr[\submart]}{\curr[\tau]}
		}	
\end{equation}
with $\curr[\snoise]$,
$\curr[\sbias]$
and
$\curr[\second]$ given by \eqref{eq:errors},
and
$\curr[\mart] = \sum_{\runalt=\start}^{\run} \iter[\step] \iter[\snoise]$,
$\curr[\submart] = \sum_{\runalt=\start}^{\run} \bracks{\iter[\step] \iter[\sbias] + \iter[\step]^{2} \iter[\second]^{2}}$
defined as in \eqref{eq:err-agg}.
Also, \eqref{eq:errorbounds} readily gives
\begin{equation}
\sum_{\run=\start}^{\infty} \exof{\curr[\step]^{2} \curr[\snoise]^{2} \given \curr[\filter]}
	\leq \sum_{\run=\start}^{\infty} \curr[\step]^{2} \exof{ \norm{\curr - \eq}^{2} \dnorm{\curr[\noise]}^{2} \given \curr[\filter]}
	\leq \poldiam^{2} \sum_{\run=\start}^{\infty} \curr[\step]^{2} \curr[\sdev]^{2}
	< \infty
\end{equation}
so, by the strong law of large numbers for martingale difference sequences \citep[Theorem 2.18]{HH80}, we conclude that $\curr[\mart] / \curr[\tau]$ converges to $0$ \acl{wp1}.
In a similar vein, for the submartingale $\curr[\submart]$ we have
\begin{align}
\exof{\curr[\submart]}
	&= \sum_{\runalt=\start}^{\run}
		\iter[\step] \iter[\sbias]
		+ \sum_{\runalt=\start}^{\run} \iter[\step]^{2} \exof{\iter[\second]^{2}}
	\leq \poldiam \sum_{\runalt=\start}^{\run} \iter[\step] \iter[\bbound]
		+ \frac{3}{2} \sum_{\runalt=\start}^{\run} \iter[\step]^{2}
			\bracks{\vbound^{2} + \iter[\bbound]^{2} + \iter[\sdev]^{2}},
\end{align}
so, by \eqref{eq:errorbounds} and the stated conditions for the method's step-size and bias/noise parameters, it follows that $\curr[\submart]$ is bounded in $L^{1}$.
Therefore, by Doob's submartingale convergence theorem \citep[Theorem~2.5]{HH80}, we further deduce that $\curr[\submart]$ converges \acl{wp1} to some (finite) random variable $\submart_{\infty}$.

Going back to \eqref{eq:energy-bound1} and letting $\run\to\infty$, the above shows that $\curr[\energy] \to -\infty$ \acl{wp1} on $\mathcal{Q}$.
Since $\energy$ is nonnegative by construction and $\probof{\mathcal{Q}} > 0$ by assumption, we obtain a contradiction and our proof is complete.
\end{proof}

\subsection{Convergence of the energy values}

Our last auxiliary result concerns the convergence of the values of the dual energy function $\energy$.
We encode this as follows.

\begin{proposition}
\label{prop:energy}
If \eqref{eq:PG} is run with assumptions as in \cref{prop:contain}, there exists a finite random variable $\energy_{\infty}$ such that
\begin{equation}
\label{eq:energy-conv}
\probof*{\text{$\curr[\energy] \to \energy_{\infty}$ as $\run\to\infty$} \given \text{$\curr\in\basin$ for all $\run$}}
	= 1.
\end{equation}
\end{proposition}

\begin{proof}
Let $\curr[\event] = \braces*{\iter\in\basin \; \text{for all $\runalt=\running,\run$}}$ be defined as in \eqref{eq:events}, and let $\curr[\tilde\energy] = \one_{\curr[\event]} \curr[\energy]$.
Then, by the energy inequality \eqref{eq:template} and the fact that $\next[\event] \subseteq \curr[\event]$, we get
\begin{align}
\next[\tilde\energy]
	= \one_{\next[\event]} \next[\energy]
	&\leq \one_{\curr[\event]} \next[\energy]
	\notag\\
	&\leq \one_{\curr[\event]} \curr[\energy]
		+ \one_{\curr[\event]} \curr[\step] \braket{\gvalue(\curr)}{\curr - \eq}
		+ \parens[\big]{
			\curr[\step] \curr[\snoise]
			+ \curr[\step] \curr[\sbias]
			+ \curr[\step]^{2} \curr[\second]^{2}
			} 
			\one_{\curr[\event]}
	\notag\\
	&\leq \curr[\tilde\energy]
		+ \curr[\step] \one_{\curr[\event]} \curr[\snoise]
		+ \parens[\big]{
			\curr[\step] \curr[\sbias]
			+ \curr[\step]^{2} \curr[\second]^{2}
			} 
			\one_{\curr[\event]},
\end{align}
where we used the fact that that $\braket{\gvalue(\iter)}{\iter - \eq} \leq 0$ for all $\runalt = \running,\run$ if $\curr[\event]$ occurs.
Since $\curr[\event]$ is $\curr[\filter]$-measurable, conditioning on $\curr[\filter]$ and taking expectations yields
\begin{align}
\exof{\next[\tilde\energy] \given \curr[\filter]}
	&\leq \curr[\tilde\energy]
		+ \curr[\step] \one_{\curr[\event]} \exof{\curr[\snoise] \given \curr[\filter]}
		+ \one_{\curr[\event]} \curr[\step] \curr[\sbias]
		+ \one_{\curr[\event]} \exof{\curr[\step]^{2} \curr[\second]^{2} \given \curr[\filter]}
	\notag\\
	&\leq \curr[\tilde\energy]
		+ \curr[\step] \poldiam \curr[\bbound]
		+ \curr[\step] \curr[\sbias]
		+ \exof{\curr[\step]^{2} \curr[\second]^{2} \given \curr[\filter]}
	\notag\\
	&\leq \curr[\tilde\energy]
		+ \curr[\step] \poldiam \curr[\bbound]
		+ \frac{3}{2} \bracks[\big]{ \vbound^{2} + \curr[\bbound]^{2} + \curr[\sdev]^{2}}.
\end{align}
By our step-size assumptions, we have $\sum_{\run} \curr[\step]^{2} (1 + \curr[\bbound]^{2} + \curr[\sdev]^{2}) < \infty$ and $\sum_{\run} \curr[\step] \curr[\bbound] < \infty$, which means that $\curr[\tilde\energy]$ is an almost supermartingale with almost surely summable increments, \ie
\begin{equation}
\sum_{\run=\start}^{\infty} \bracks*{\exof{\next[\tilde\energy] \given \curr[\filter]} - \curr[\tilde\energy]}
	< \infty
	\quad
	\text{\acl{wp1}}
\end{equation}
Therefore, by Gladyshev's lemma \citep[p.~49]{Pol87}, we conclude that $\curr[\tilde\energy]$ converges almost surely to some (finite) random variable $\energy_{\infty}$.
Since $\one_{\curr[\event]} = 1$ for all $\run$ if and only if $\curr\in\basin$ for all $\run$, we conclude that $\probof*{\text{$\curr[\energy]$ converges} \given \text{$\curr\in\basin$ for all $\run$}} = \probof{\text{$\curr[\tilde\energy]$ converges}}= 1$, and our claim follows.
\end{proof}

\subsection{Putting everything together}
We are now in a position to prove \cref{thm:stable,cor:stable}.

\begin{proof}[Proof of \cref{thm:stable}]
Let $\event = \intersect_{\run} \curr[\event] = \{\text{$\curr\in\basin$ for all $\run$} \}$ denote the event that $\curr$ lies in $\basin$ for all $\run$.
By \cref{prop:contain}, if $\init$ is initialized within the neighborhood $\nhd$ defined in \eqref{eq:nhd-init}, we have $\probof{\event \given \init\in\nhd} \geq 1-\thres$, noting also that the neighborhood $\nhd$ is independent of the required confidence level $\thres$.
Then, by \cref{prop:subsequence,prop:energy}, it follows that
\begin{enumerate*}
[\itshape a\upshape)]
\item
$\liminf_{\run} \norm{\curr - \eq} = 0$;
and
\item
$\curr[\energy]$ converges,
\end{enumerate*}
both events occurring \acl{wp1} on the set $\event \cap \{ \init\in\nhd \}$.
We thus conclude that $\lim_{\run\to\infty} \curr[\energy] = 0$ and hence
\begin{align*}
\probof{ \curr\to\eq \given \init \in \nhd}
	&\geq \probof{ \event \cap \{ \curr\to\eq \} \given \init \in \nhd}
	\notag\\
	&= \probof{ \curr\to\eq \given \init \in \nhd, \event}
		\times \probof{\event \given \init \in \nhd}
	\geq 1 - \conf,
\end{align*}
and our proof is complete.
\end{proof}

\begin{proof}[Proof of \cref{cor:stable}]
For \cref{mod:full,mod:stoch}, taking $\bexp = \infty, \sexp = 0$, we obtain $\pexp > 1/2$. Since we have that $\sum_{\run=1}^{\infty}\step_\run = \infty$, we get that $\pexp \leq 1$, i.e., $\pexp \in (1/2,1]$.

For \cref{mod:value}, we have that $\curr[\bbound] = \bigoh(\curr[\expar])$ and $\curr[\sdev] = \bigoh(1/\sqrt{\curr[\expar]})$, i.e., $\bexp = \rexp$ and $\sexp = \rexp/2$. Now, since $\pexp \leq 1$, $\pexp + \bexp >1$ and $\pexp - \sexp > 1/2$, we obtain that $\pexp \in (2/3,1]$ and $(1-p)/2 < r/2 < p - 1/2$.
\end{proof}

%% file: App-SOS.tex

We now proceed with the proof of \cref{thm:SOS}, which we again restate below for convenience:

\SOS*

\begin{proof}
We will follow an approach similar to \cref{thm:stable} for the first part of the theorem.
More precisely, let $\basin = \setdef{\policy\in\policies}{\norm{\policy - \eq} \leq \radius}$ be a ball of radius $\radius$ centered at $\eq$ in $\policies$ such that \eqref{eq:SOS} holds for all $\policy\in\basin$.
Then, for all $\policy\in\basin\exclude{\eq}$, we have $\braket{\gvalue(\policy)}{\policy - \eq} \leq - \strong \norm{\policy - \eq} < 0$ by \cref{prop:var}.
Hence, defining the events $\curr[\event]$ and $\curr[\eventalt]$ as in \cref{eq:events}, and assuming that $\init$ is initialized in a ball of radius $\sqrt{2\thres}$ centered at $\eq$, viz.
\begin{equation}
\label{eq:nhd-init-1}
\nhd
	= \setdef{\policy\in\policies}{\energy(\policy) \leq \thres}
	= \setdef{\policy\in\policies}{\norm{\policy-\eq}^{2}/2 \leq \thres}.
\end{equation}
then, by \cref{lem:events} and \cref{prop:contain}, we readily obtain that
\begin{equation}
\label{eq:contain}
\probof{\curr[\eventalt] \given \init\in\nhd}
	\geq 1-\conf
	\quad
	\text{for all $\run=\running$}
\end{equation}
which implies that 
\begin{equation}
	\probof{\event \given \init\in\nhd} \geq 1-\conf
\end{equation}
if $\runoff$ is large enough relative to $\conf$.

For the second part, constraining \cref{eq:template} on the event $\event_{\run}$, we get:
\begin{align}\label{eq:basic energy}
	\next[\energy]\one_{\event_{\run}} &\leq \curr[\energy]\one_{\event_{\run}} + \step_{\run}\inner{\gvalue(\policy_{
	\run})}{\curr - \np}\one_{\event_{\run}} + \one_{\event_{\run}}\parens*{\step_{\run}\snoise_{\run} + \step_{\run}\sbias_{\run} + \step_{\run}^{2} \curr[\second]^{2}}
	\notag\\
	&\leq (1-2\mu\step_{\run})\curr[\energy]\one_{\event_{\run}} + \one_{\event_{\run}}\parens*{\step_{\run}\snoise_{\run} + \step_{\run}\sbias_{\run} + \step_{\run}^{2} \curr[\second]^{2}}
\end{align}
where the last inequality comes from \eqref{eq:SOS}.
Therefore, taking expectations, we obtain:
\begin{align}
\exof*{\next[\energy]\one_{\event_{\run}}}
	&\leq (1-2\mu\step_{\run}) \exof*{\curr[\energy]\one_{\event_{\run}}}
		+ \exof*{\one_{\event_{\run}}\parens*{\step_{\run}\snoise_{\run} + \step_{\run}\sbias_{\run} + \step_{\run}^{2} \curr[\second]^{2}}}
	\notag\\
	&\leq (1-2\mu\step_{\run}) \exof*{\curr[\energy]\one_{\event_{\run}}}
		+ \step_{\run} \exof*{\one_{\event_{\run}}\snoise_{\run}}
		+ \step_{\run} \exof{\one_{\event_{\run}}\sbias_{\run}}
		+ \step_{\run}^{2} \exof{\one_{\event_{\run}}\curr[\second]^{2}}
	\notag\\
	&= (1-2\mu\step_{\run}) \exof*{\curr[\energy]\one_{\event_{\run}}}
		+ \step_{\run}\exof{\one_{\event_{\run}}\sbias_{\run}}
		+ \step_{\run}^{2}\exof{\one_{\event_{\run}}\curr[\second]^{2}}
	\notag\\
	&\leq (1-2\mu\step_{\run}) \exof*{\curr[\energy]\one_{\event_{\run}}}
		+ \poldiam \probof{\event_{\run}} \step_{\run} \bbound_{\run}
		+ \probof{\event_{\run}}
			\parens*{\vbound\step_{\run}^{2} + 3\step_{\run}^{2} \sdev_{\run}^{2} + 3\step_{\run}^{2} \bbound_{\run}^{2}}
\end{align}
where the equality in the third line comes from the fact that
\begin{align}
\exof{\one_{\event_{\run}}\snoise_{\run}}
	= \exof*{\exof{\snoise_{\run} \one_{\event_{\run}} \given \filter_{\run}}}
	= \exof*{\one_{\event_{\run}} \exof{\snoise_{\run} \given \filter_{\run}}} = 0.
\end{align}
Now, since $\one_{\event_{\run+1}} \leq \one_{\event_{\run}}$, we further have
\begin{equation}
\exof*{\next[\energy]\one_{\event_{\run+1}}}
	\leq \exof*{\next[\energy]\one_{\event_{\run}}}
\end{equation}
and hence, setting $\bD_{\run} \defeq \exof*{\curr[\energy]\one_{\event_{\run}}}$, we get
\begin{align}
\next[\bD]
	&\leq (1-2\mu\step_{\run}) \curr[\bD]
		+ \poldiam \probof{\event_{\run}} \step_{\run} \bbound_{\run}
		+ \probof{\event_{\run}}
			\parens*{\vbound \step_{\run}^{2} + 3\step_{\run}^{2} \sdev_{\run}^{2} + 3 \step_{\run}^{2}\bbound_{\run}^{2}}
	\notag\\
	&\leq (1-2\mu\step_{\run}) \curr[\bD]
		+ \poldiam\step_{\run}\bbound_{\run}
		+ \vbound \step_{\run}^{2}
		+ 3\step_{\run}^{2} \sdev_{\run}^{2}
		+ 3 \step_{\run}^{2}\bbound_{\run}^{2}.
\end{align}
Therefore, taking $\step_{\run}, \bbound_{\run},\sdev_{\run}$ as per the statement of the theorem and noting that the terms $\step_{\run}^{2}$ and $\step_{\run}^{2} \bbound_{\run}^{2}$ are respectively dominated by the terms $\step_{\run}^{2}\sdev_{\run}^{2}$ and $\step_{\run} \bbound_{\run}$, we obtain
\begin{align}
	\next[\bD]
	&\leq \parens*{1-\frac{2\mu\step}{(\run + \runoff)^\pexp}} \curr[\bD]
		+ \frac{\Const_{1}}{(\run+\runoff)^{\pexp+\bexp}}
		+ \frac{\Const_{2}}{(\run+\runoff)^{2\pexp-2\sexp}}
	\notag\\
	&\leq \parens*{1-\frac{2\mu\step}{(\run + \runoff)^\pexp}}\curr[\bD] + \frac{\Const_{1} + \Const_{2}}{(\run+\runoff)^{\pexp+\qexp}}
\end{align}
for some $\Const_{1}, \Const_{2} >0$, where $\qexp = \min\{\bexp,\pexp - 2\sexp\}$, as per the theorem's statement.
Therefore, by a straightforward modification of Chung's lemma \cite[Lemmas 2\&3]{Chu54}, \cite[p.~45]{Pol87}, we get
\begin{equation}
	\bD_{\run}
	=
	\begin{cases}
	\bigoh(1/\run^{2\strong\step})
		&\quad
		\text{if $\pexp=1$ \, and \, $2\strong\step<\qexp$},
		\\
	\bigoh(1/\run^{\qexp})
		&\quad
		\text{otherwise}.
	\end{cases}
\end{equation}
Accordingly, letting $\run \to \infty$ and recalling that $\exof{\curr[\energy]\one_{\event}} \leq \exof{\curr[\energy]\one_{\event_{\run}}} = \curr[\bD]$ 
\begin{align}
	\lim_{\run \to \infty} \exof{\curr[\energy]\one_{\event}}
	= 0.
\end{align}
Then, by Fatou's lemma \citep{Fol99}, we obtain
\begin{align}
0
	\leq \exof{\liminf_{\run \to \infty} \curr[\energy]\one_{\event}}
	\leq \liminf_{\run \to \infty} \exof{\curr[\energy]\one_{\event}}
	= 0,
\end{align}
which readily shows that $\exof{\liminf_{\run \to \infty}\curr[\energy]\one_{\event}} = 0$.
Finally, since $\liminf_{\run \to \infty}\curr[\energy]\one_{\event} \geq 0$ \as and $\exof{\liminf_{\run \to \infty}\curr[\energy]\one_{\event}} = 0$, we get that
\begin{equation}
\liminf_{\run \to \infty} \curr[\energy]\one_{\event}
	= 0
	\quad
	\text{\acl{wp1}}.
\end{equation}
Therefore, there exists a subsequence $\energy_{\run_{\runalt}}$ that converges to $0$ \acl{wp1} on the event $\event$, \ie $\policy_{\run_{\runalt}}$ converges to $\eq$. 
Hence, invoking \cref{prop:energy}, we further deduce that $\energy_{\run}$ converges to some $\energy_\infty$ \acl{wp1} on $\event$, and thus, we obtain that $\lim_{\run \to \infty}\energy_{\run} = 0$ on $\event$.
We thus get
\begin{align}
\probof{ \curr\to\eq \given \init \in \nhd}
	&\geq \probof{ \event \cap \{ \curr\to\eq \} \given \init \in \nhd}
	\notag\\
	&= \probof{ \curr\to\eq \given \init \in \nhd, \event}
		\times \probof{\event \given \init \in \nhd}
	\geq 1 - \conf,
\end{align}
as claimed.

For the last part of the theorem, note that
\begin{align}
\curr[\bD]
	= \ex[D_{\run} \one_{\event_{\run}}]
	\geq \exof*{D_{\run} \one_{\event}}
	&= \exof*{\exof{D_{\run} \given \sigma(\event)} \one_{\event}}
	\notag\\
	&= \exof*{\exof{D_{\run} \given \event} \one_{\event}} 
	\notag\\
	&= \exof{D_{\run} \given \event} \exof*{\one_{\event}} 
	\notag\\
	&= \exof{D_{\run} \given \event} \prob\parens{\event} 
\end{align}
where we used the fact that $\exof{D_{\run} \given \sigma(\event)} \one_{\event} = \exof{D_{\run} \given \event} \one_{\event}$.
We thus conclude that
\begin{align}
\exof*{\norm{\policy_{\run} - \np}^{2} \given \event}
	= 2\exof{D_{\run} \given \event} \leq \frac{2}{\prob(\event)}\curr[\bD]
	\leq \frac{2}{1-\conf}\curr[\bD]
	\notag\\
\end{align}
and hence
\begin{equation*}
\begin{aligned}
\exof*{\norm{\policy_{\run} - \np}^{2} \given \event} =
	\begin{cases}
	\bigoh(1/\run^{2\strong\step})
		&\quad
		\text{if $\pexp=1$ \, and \, $2\strong\step<\qexp$},
		\\
	\bigoh(1/\run^{\qexp})
		&\quad
		\text{otherwise}.
	\end{cases}
\end{aligned}
\qedhere
\end{equation*}
\end{proof}

\begin{proof}[Proof of \cref{cor:SOS}]
For \cref{mod:full,mod:stoch}, taking $\bexp = \infty, \sexp = 0$ we readily get that $\qexp = \pexp$ and $\pexp > 1/2$. Since we require that $\sum_{\run=1}^{\infty}\step_{\run} = \infty$, we obtain that $\pexp \in (1/2,1]$. Hence, for $\pexp = 1$ and $2\strong\step > 1$ we obtain $\bigoh(1/\run)$ rate of convergence.

For \cref{mod:value}, we have that $\curr[\bbound] = \bigoh(\curr[\expar])$ and $\curr[\sdev] = \bigoh(1/\sqrt{\curr[\expar]})$, \ie $\bexp = \pexp/2$ and $\sexp = \pexp/4$, and, hence, we readily get that $\qexp = \pexp/2$. Now, since $\pexp \leq 1$, $\pexp + \bexp >1$ and $\pexp - \sexp > 1/2$, we obtain that $\pexp \in (2/3,1]$. Hence, for $\pexp = 1$ and $\strong\step > 1$, we obtain $\bigoh(1/\sqrt{\run})$ rate of convergence.
\end{proof}

We conclude this appendix with a detailed statement and proof of the fact that, when run with perfect policy gradients (\ie as per \cref{mod:full}), the sequence of play generated by \eqref{eq:PG} achieves a geometric convergence rate to Nash policies satisfying \eqref{eq:SOS}.
The precise result is as follows:

\begin{proposition}
\label{prop:geometric}
Let $\eq$ be a \acl{SOS} policy,
let $\basin$ be a neighborhood of $\eq$ such that \eqref{eq:strong} holds on $\basin$,
and
let $\curr$ be the sequence of play generated by \eqref{eq:PG} with a sufficiently small constant step-size $\step > 0$ and perfect policy gradients as per \cref{mod:full}.
Then, there exists a neighborhood $\nhd$ of $\eq$ in $\policies$ and some $\rho>0$ such that
\begin{equation}
\label{eq:geometric}
\norm{\curr - \eq}
	= \bigoh(\exp(-\rho\run))
	\quad
	\text{whenever $\init\in\nhd$}.
\end{equation}
\end{proposition}

\begin{proof} 
The crucial part of the proof is the observation that, in the case of \cref{mod:full}, the energy inequality \eqref{eq:energy} of \cref{lem:energy} may be written in the sharper form:
\begin{align}
\label{eq:energy-strong}
\next[\energy]
	&\leq \curr[\energy]
		+ \step \braket{\gvalue(\curr) - \gvalue(\eq)}{\curr - \eq}
		+ \tfrac{1}{2} \step^{2} \norm{\gvalue(\curr) - \gvalue(\eq)}^{2}.
\end{align}
To see this, consider the development
\begin{align}
\label{eq:energy-dev1}
\norm{\next - \eq}^{2}
	&= \norm{\next - \curr + \curr - \eq}^{2}
	\notag\\
	&= \norm{\curr - \eq}^{2}
		+ 2 \braket{\next - \curr}{\curr - \eq}
		+ \norm{\next - \curr}^{2}
	\notag\\
	&= \norm{\curr - \eq}^{2}
		+ 2 \braket{\next - \curr}{\next - \eq}
		- \norm{\next - \curr}^{2}
	\notag\\
	&\leq \norm{\curr - \eq}^{2}
		+ 2\step\braket{\gvalue(\curr)}{\next - \eq}
		- \norm{\next - \curr}^{2},
	\notag\\
	&\leq \norm{\curr - \eq}^{2}
		+ 2\step\braket{\gvalue(\curr) - \gvalue(\eq)}{\next - \eq}
		- \norm{\next - \curr}^{2},
\end{align}
where the last line follows from \eqref{eq:FOS} and, in the penultimate step, we used the fact that $\eq\in\policies$ so $\braket{\next - (\curr + \step\gvalue(\curr))}{\next - \eq} \leq 0$.
In addition, by Young's inequality, we have
\begin{align}
\label{eq:energy-dev2}
\braket{\gvalue(\curr) - \gvalue(\eq)}{\next - \eq}
	&= \braket{\gvalue(\curr) - \gvalue(\eq)}{\curr - \eq}
		+ \braket{\gvalue(\curr) - \gvalue(\eq)}{\next - \curr}
	\notag\\
	&\leq \braket{\gvalue(\curr) - \gvalue(\eq)}{\curr - \eq}
		+ \frac{\step}{2} \norm{\gvalue(\curr) - \gvalue(\eq)}^{2}	
		+ \frac{1}{2\step} \norm{\next - \curr}^{2}
\end{align}
so \eqref{eq:energy-strong} follows by substituting \eqref{eq:energy-dev2} in \eqref{eq:energy-dev1} and simplifying.

Now, since $\gvalue$ is $G$-Lipschitz (\cf \cref{lem:smoothness} in \cref{app:marl}), we have $\norm{\gvalue(\curr) - \gvalue(\eq)} \leq G \norm{\curr - \eq}$, so the energy inequality \eqref{eq:energy-strong} becomes
\begin{equation}
\next[\energy]
	\leq \curr[\energy]
		+ \step \braket{\gvalue(\curr) - \gvalue(\eq)}{\curr - \eq}
		+ \step^{2} G^{2} \curr[\energy].
\end{equation}
However, if $\curr \in \basin$, \cref{prop:var} further yields
\begin{equation}
\braket{\gvalue(\curr) - \gvalue(\eq)}{\curr - \eq}
	\leq - \strong \norm{\curr - \eq}^{2}
\end{equation}
so
\begin{equation}
\next[\energy]
	\leq \curr[\energy]
		- \strong\step \norm{\curr - \eq}^{2}
		+ \step^{2} G^{2} \curr[\energy]
	= (1 - 2\strong\step + \step^{2}G^{2}) \curr[\energy].
\end{equation}
Thus, if $\step < 2\strong/G^{2}$ and $\init$ is initialized in a ball centered at $\eq$ and contained within $\basin$, our assertion follows from a straightforward induction argument.
\end{proof}

%% file: App-Deterministic.tex




\subsection{Structural preliminaries}

To prove \cref{thm:strict}, we will first require some notions describing the geometry of $\policies$ near $\eq$.
Referring to \cite{RW98} for a full treatment, we have:

\begin{definition}
Let $\cvx$ be a convex set and let $x\in \cvx$.
Then the tangent cone $\tcone_{\cvx}(x)$ is defined as the set of all rays emanating from $x$ and intersecting $\cvx$ to at least one other point different from $x$. The \textit{polar cone} $\pcone_{\cvx}(x)$ to $\cvx$ at $\point$ is then defined $\pcone_{\cvx}(x) = \braces{y: \braket{\dpoint}{\tvec} \leq 0 \text{ for all } \tvec\in \tcone_{\cvx}(x)}$, where $\dpoint$ belong in the dual space of the vector space in which $\cvx$ is defined.
\end{definition}

With these general definitions in hand, we proceed to characterize some further projections of Euclidean projections on $\policies$ that will play an important role in the sequel.
For notational simplicity, we suppress the player and state indices in the statement and proof of the next lemma.

\begin{lemma}
\label{lem:KKT}
$\point=\proj(\dpoint)$ if and only if there exist $\mu\in \R$ and $\nu_\pure\in\R_{+}$ such that, for all $\pure\in \pures$, we have
$\dpoint_\pure= \point_\pure+\mu-\nu_\pure$
with $\nu_\pure \geq 0$ and $\point_\pure \nu_\pure=0$.
\end{lemma}

\begin{proof}
Recall that 
\(
\proj(\dpoint)
    = \argmin_{\point\in\simplex(\pures)} \norm{\dpoint-\point}^{2}.
\)
Our result then follows by applying the KKT conditions to this optimization problem and noting that, since the constraints are affine, the KKT conditions are sufficient for optimality.
Our Langragian is
\begin{equation}
\mathcal{L}(\point,\mu,\nu)= \sum_{\pure\in\pures}\frac{1}{2}(\dpoint_\pure - \point_\pure )^2-{\mu(\sum_{\pure\in\pures} \point_\pure-1)} + {\sum_{\pure\in\pures} \nu_\pure \point_\pure}
\end{equation}
where the set of constraints (i) of the statement of the lemma  are the stationarity constraints, 
which in our case are $\nabla\mathcal{L}(\point,\mu,\nu)=0\Leftrightarrow
\nabla( \sum_{\pure\in\pures}\frac{1}{2}(\dpoint_\pure - \point_\pure )^2 )=\mu\nabla(\sum_{\pure\in\pures} \point_\pure-1) - \sum_{\pure\in\pures} \nu_\pure\nabla \point_\pure$
, while the set of constraints (ii) of the statement of the lemmas are the complementary slackness constraints.
Note that complementary slackness implies $\nu_\pure>0$ whenever $\pure\notin \supp(\point)$, so our proof is complete.
\end{proof}

Our next result is a concrete consequence of \cref{prop:var} which will be very useful in establishing the stability estimates required for the proof of \cref{thm:strict}.

\begin{lemma}\label{lem:basin of drift}
Let $\np = (\pure^*_{\play,\rstate})_{\play\in\players, \rstate\in\states}$ be a strict Nash policy.
Then there exists a neighborhood $\nhd$ of $\np$ and constants $c_{\play,\rstate}$ such that for each player $\play\in\players$ and state $\rstate\in\states$, we have:
\begin{equation}\label{eq:variational-inequality}
\text{$\payv_{\play\pure_{\play,\rstate}^*}\parens{\policy} - \payv_{\play\pure_{\play,\rstate}}\parens{\policy} \geq c_{\play,\rstate}\;$ for all $\policy\in\nhd$ and $\pure_\play\neq\pure_\play^*$, $\pure_\play\in\pures_\play$.}
\end{equation} 
\end{lemma}

\begin{proof}
Our claim  is a consequence of the definition of strict Nash policies. Specifically, from \cref{prop:var} we have
\begin{equation}
    \inner{\gvalue(\np)}{z} <0 \quad\text{ for all }\quad z\in \tcone(\np), z\neq 0
\end{equation}
Let $z = e_{\play,\pure_{\play,\rstate}} -e_{\play,\pure^*_{\play,\rstate}}$, then we get that 
\begin{equation}
    \payv_{\play\pure_{\play,\rstate}^*}\parens{\np} - \payv_{\play\pure_{\play,\rstate}}\parens{\np} >0
\end{equation}
where $e_{\play,\pure_{\play,\rstate}}$ is the vector that has one only in the index and zero anywhere else.
By continuity there exists a neighborhood $\nhd\subseteq\pspace$ and $c_{\play,\rstate} > 0$ for each player $\play\in\players$ such that 
\begin{equation*}
 \payv_{\play\pure_{\play,\rstate}^*}\parens{\policy} - \payv_{\play\pure_{\play,\rstate}}\parens{\policy} \geq c_{\play,\rstate} \quad\text{ for all }\quad\policy\in\nhd
 \qedhere
\end{equation*}
\end{proof}

Our final result is intimately tied to the lazy projection step in \eqref{eq:LPG}, and quantifies the relation between initializations in $\dspace$ and $\policies$. 

\begin{lemma}\label{lem:equivalence} Let $\np = (\pure^*_{\play,\rstate})_{\play\in\players, \rstate\in\states}$, be a  deterministic policy. For each agent $\play\in\players$ and each state $\rstate\in\states$, let $\dstate_{\play,\pure_{\play,\rstate}} -\dstate_{\play,\pure^*_{\play,\rstate}}$ be the difference of the aggregated gradients between the strategy of the equilibrium and any other strategy $\pure_\play^*\neq \pure_{\play}\in\pures_\play$. Then for any $\eps >0 $ such that $\nhd_{\eps} = \braces{\policy: \policy_{\play,\pure_{\play,\rstate}^*}\geq 1-\eps \text{ for all } \play\in\players \text{ and } \rstate\in\states}$, there exist $M_{\play,\eps,\rstate}$ such that if $\dbasin_{\play,\rstate} = \braces{\dstate\in\R^{\pures_\play}: \dstate_{\play,\pure_{\play,\rstate}} -\dstate_{\play,\pure^*_{\play,\rstate}} < - M_{\play,\eps,\rstate}}$  then $\prod_{\play\in\players,\rstate\in\states}\proj_{\policies_{\play}}\parens{\dbasin_{\play,\rstate}}\subseteq \nhd_{\eps}$.
\end{lemma}
\begin{proof}

Consider an arbitrary player $\play\in\players$, a state $\rstate\in\states$, and let $\dbasin_\play\parens{M_{\play,\eps,\rstate}}$ be an open set as defined in the statement of the lemma. For notational simplicity, we will drop the index $\rstate$.
We will show that any $M_{\play,\eps}> 1-\frac{\eps}{|\pures_\play|}>0 $ satisfies our claim.
By  using \cref{lem:KKT} for a $\dstate_\play\in \dbasin_\play\parens{M_{\play,\eps}}$ with $\policy_\play=\proj\parens{\dstate_\play}$ we have that 
\begin{equation}
\begin{aligned}
    &\dstate_{\play\pure_\play^*}-\dstate_{\play\pure_\play}>M_{\play,\eps}\\
    &  {\policy_{\play\pure_\play^*}}-{\policy_{\play\pure_\play}} - (\nu_{\pure^*_\play}-\nu_{\pure_\play}) > M_{\play,\eps}\label{eq:contradiction}
\end{aligned}
\end{equation}
with $\nu_{\pure_\play} \geq 0$ and $\policy_{\play\pure_\play} =0$ whenever $\nu_{\pure_\play} >0$.
Notice that since $M_{\play,\eps} >1-\frac{\eps}{\nPures_{\play}}$  we  have that ${\policy_{\play\pure_\play^*}}>{\policy_{\play\pure_\play}} + 1 - \frac{\eps}{\nPures_{\play}} 
+ (\nu_{\pure^*_\play}-\nu_{\pure_\play})$ or 
\begin{equation}
  {\policy_{\play\pure_\play}}<   {\policy_{\play\pure_\play^*}} - 1 + \frac{\eps}{\nPures_{\play}} 
- (\nu_{\pure^*_\play}-\nu_{\pure_\play}) <\frac{\eps}{\nPures_{\play}} 
\end{equation}
Hence, by summing over all strategies of player $\play$ we get the desired result.
\end{proof}

\subsection{Proof of the main theorem}

We are now in a position to prove our main result on the rate of convergence towards strict Nash policies.
For ease of reference, we restate \cref{thm:strict} below.

\strict*

\begin{proof}[Proof of \cref{thm:strict}]

We start by fixing a confidence level $\conf>0$ and all the parameters of the algorithm, such that all the assumptions stated in the theorem are satisfied and. We will prove that for each agent $\play\in\players$, $\rstate\in\states$ there exist $M_{1,\play,\rstate}>0$, $\dbasin_{1,\play,\rstate} = \braces{\dstate\in \R^{\pures_\play}: \dstate_{\play,\pure_\play} -\dstate_{\play,\pure^{*}_{\play}}<-M_{1,\play,\rstate} \text{ for all } \pure_\play\in\pures_\play, \pure_\play\neq\pure_{\play}^*}$,  such that if $\dstate_{\start}\in\dbasin_1 \defeq \prod_{\play\in\players,\rstate\in\states}\dbasin_{1,\play,\rstate}$  then the agents' sequence of play, converge to the deterministic Nash policy, in finite number of iterations. 

To simplify the notation, we will drop the indices $\rstate$ and $\play$ referring to the states and agents, accordingly, and we will focus on a specific agent and a specific state.
From \cref{lem:equivalence}, \cref{lem:basin of drift} we have that there exist constants $c,M$, neighborhood $\nhd_c =\braces{\policy\in\policies: \norm{\policy-\np}\leq \rad}$ and open set $\dbasin_{M}$ such that \begin{equation}
\begin{alignedat}{3}
      \payv_{\pure^*}\parens{\policy} -\payv_{\pure}\parens{\policy} &\geq c &\hspace{0.5em}&\text{ for all } \pure\neq\pure^*, \pure\in\pures \text{ and }\policy\in\nhd_c\\
      \dstate_{\pure^*}-\dstate_\pure &> M_c&\hspace{0.5em} &\text{ for all } \pure\neq\pure^*, \pure\in\pures  \text{ and } \policy = \proj\parens{\dstate}\in\nhd_c
\end{alignedat}
\end{equation}
The first step is to prove that for an appropriate initialization for $\dstate_\start$, we have $\dstate_\run\in \dbasin(M_c)$ for all $\run = \running$, with probability at least $1-\conf$. Assume that $\dstate_\runalt\in\dbasin(M_c)$ for all $\runalt=1,\hdots,n$; then 
for   the differences of the scores at a round $\run+1$ between any $\pure\in\pures$ and the equilibrium strategy $\pure^*$, we have
\begin{align}\label{eq: basic equality}
\dstate_{\pure,\run +1} -\dstate_{\pure^*,\run+1}
    &= \dstate_{\pure,\run} -\dstate_{\pure^*,\run} + (\signal_{\pure,\run} -\signal_{\pure^*,\run})
    \notag\\
    &=\dstate_{\pure,\start} -\dstate_{\pure^*,\start} + \sum_{\runalt = \start}^\run \step_\runalt [(\gvalue_{\pure,\runalt} -\gvalue_{\pure^*,\runalt})  + (\noise_{\pure,\runalt} -\noise_{\pure^*,\runalt}) +( \bias_{\pure,\runalt} -\bias_{\pure^*,\runalt})]
    \notag\\
    &\leq -M_1 + \sum_{\runalt = \start}^\run \step_\runalt [(\gvalue_{\pure,\runalt} -\gvalue_{\pure^*,\runalt})  + (\noise_{\pure,\runalt} -\noise_{\pure^*,\runalt}) +( \bias_{\pure,\runalt} -\bias_{\pure^*,\runalt})]
    \notag\\
    &\leq -M_1 -c\sum_{\runalt=1}^{\run}\step_{\runalt} + \sum_{\runalt=1}^{\run}\step_{\runalt}[ (\noise_{\pure,\runalt} -\noise_{\pure^*,\runalt}) +( \bias_{\pure,\runalt} -\bias_{\pure^*,\runalt})]
    \notag\\
    &\leq -M_1 -c\sum_{\runalt=1}^{\run}\step_{\runalt} + \sum_{\runalt=1}^{\run}\step_{\runalt}[ \snoise_{\runalt} + \sbias_{\runalt}]
\end{align}
where $\snoise_{\runalt} =  (\noise_{\pure,\runalt} -\noise_{\pure^*,\runalt})$ and $\sbias_{\runalt} =  2\norm{\bias_{\runalt}}$. 
Now, similarly to the proofs of \cref{thm:stable,thm:SOS} we will proceed to control the aggregate error terms 
\begin{equation}
\label{eq:err-agg}
\curr[R]
	= \sum_{\runalt=\start}^{\run} \iter[\step] \iter[\snoise]
	\quad
	\text{and}
	\quad
\curr[\submart]
	= \sum_{\runalt=\start}^{\run} {\iter[\step] \iter[\sbias]}.
\end{equation}
Since $\exof{\curr[\snoise] \given \curr[\filter]} = 0$, we have $\exof{\curr[R] \given \curr[\filter]} = \prev[R]$, so $\curr[R]$ is a martingale;
likewise, $\exof{\curr[\submart] \given \curr[\filter]} \geq \prev[\submart]$, so $\curr[\submart]$ is a sub-martingale. Furthermore from \eqref{eq:errorbounds}  we have:
\begin{enumerate}[I.]
     \item \label{eq:I}     $\exof{\snoise^{2}_\run} \leq  \exof{\norm{\noise_{\run}}^{2}} \leq \exof{\exof{\dnorm{\noise_{\run}}^{2} \given \filter_{\run}}} \leq \sdev^{2}_\run$\vspace{3pt}
     \item \label{eq:II}
     $\exof{\sbias_\run} = 2 \ex\bracks*{\dnorm{\bias_{\run}}}\leq \exof{\exof{\dnorm{\bias_{\run}} \given \filter_{\run}}} \leq \bbound_\run$
 \end{enumerate}

 Moreover, for any $\eta_{1} > 0$, we get by Doob's Maximal Inequality:

 \begin{align}
     \prob\parens*{\sup_{1 \leq \runalt \leq \run}R_{\runalt} \geq \eta_{1}} \leq \frac{\exof{R_{\run}^{2}}}{\eta_1^{2}} \stackrel{(a)}{=} \frac{\sum_{\runalt = 1}^{\run} \step_{\runalt}^{2}\exof{\snoise_{\runalt}^{2}}}{\eta_1^{2}} & \stackrel{(\ref{eq:I})}{\leq}
      \frac{\sum_{\runalt = 1}^{\run} \step_{\runalt}^{2}\sdev^{2}_{\runalt}}{\eta_1^{2}}   
 \end{align}
 where $(a)$ comes from the fact that $\exof{\snoise_i\snoise_j} = 0$ for $i \neq j$.  Since $\step_\run =\gamma/\run^{\pexp}$, $\sdev_\run = \bigoh(\run^{\sexp})$ and  $\pexp-\sexp >1/2$, 
 there exists $\gamma_1$ sufficiently small such that if $\gamma \leq \gamma_1$ then
 \begin{align}\label{eq: condition 1}
    \sum_{\runalt = 1}^{\infty}\step_{\runalt}^{2}\sdev^{2}_{\runalt} < \frac{\delta \eta_1^{2}}{2}
 \end{align}
 and so we automatically get that 
 \begin{equation}\label{eq:mart-prob}
     \prob\parens*{\sup_{1 \leq \runalt \leq \run}R_{\runalt} \geq \eta_{1}} \leq \frac{\conf}{2}
 \end{equation}

Furthermore, notice that  the term $\{S_{\run}\}_{\run \in \N}$ is a sub-martingale, since $\exof{\abs{S_{\run}} \given \curr[\filter]} < \infty$ and $\exof{S_{\run+1} \given \curr[\filter]} > S_{\run}$, for all $\run$. As before, using Doob's Maximal Inequality, we get for any $\eta_2 > 0$:
 \begin{align}
     \prob\parens*{\sup_{1 \leq \runalt \leq \run}S_{\runalt} \geq \eta_{2}} \leq \frac{\exof{S_{\run}}}{\eta_{2}} 
     & = \frac{\sum_{\runalt = 1}^{\run} \step_{\runalt}\exof{\sbias_{\runalt}} }{\eta_2}
     {\leq}\frac{2\sum_{\runalt = 1}^{\run} \step_{\runalt}\bbound_{\runalt} }{\eta_2}
 \end{align}
So, since $\pexp+\bexp > 1$  there exists $\gamma_2$ sufficiently small such that if $\gamma \leq \gamma_2$ then
\(
\sum_{\runalt = 1}^{\run} \step_{\runalt}\bbound_{\runalt} \leq \frac{\eta_2\conf}{4}
\)
which immediately implies that 
\begin{equation}\label{eq:prob-bias}
    \prob\parens*{\sup_{1 \leq \runalt \leq \run}S_{\runalt} \geq \eta_{2}} \leq\dfrac{\conf}{2}
\end{equation}
By choosing $\gamma \leq \min\braces{\gamma_1,\gamma_2}$ we get that 
 \begin{equation}
     \prob\parens*{\sup_{1 \leq \runalt \leq \run} R_\run + S_\run \leq M_c} \geq 1-\conf.
 \end{equation}
 Notice now that by choosing $M_1 > M_c + \eta_1 +\eta_2$, from \eqref{eq: basic equality} we have that with probability at least $1-\conf$, $\dstate_{\pure,\run+1} -\dstate_{\pure^*,\run+1} < -M_c$, which implies that $\policy_{\run+1}\in \nhd_c$.
 
 Defining the sequences of ``good'' events $\{\event_\run\}_{\run \in \N}$ and $\{\event'_\run\}_{\run \in \N}$ as $\event_\run \defeq \braces{\policy_\runalt \in \nhd_c \text{ for all }  \runalt=\running,\run}$ and $\event'_\run \defeq \braces*{\sup_{1 \leq \runalt \leq \run} R_{\runalt}+S_{\runalt} \leq \eta_1 + \eta_2}$, accordingly, we get that $\event'_\run \subseteq \event_\run$ for all $\run$. Because $\prob\parens*{\event'_\run} \geq 1-\conf$, we get that
\(
\prob\parens*{\event_{\run}} \geq 1-\conf
\)
and since $\{\event_\run\}_{\run \in \N}$ is a decreasing sequence converging to $\event \defeq \braces*{\policy_\run \in \nhd_c, \forall \run \in \N}$, we obtain
\(
\prob\parens*{\event} \geq 1-\conf.
\)
\ie
\begin{equation}
\prob\parens*{\curr \in \nhd_c, \; \forall \run \mid \dstate_{1} \in \dbasin_{1}} \geq 1-\conf
\end{equation}
Notice that the above conclusions immediately imply convergence in finite time. More specifically, constrained to the event $\event$ with probability at least $1-\conf$, from \cref{eq: basic equality} we have 
\begin{align}
    \dstate_{\pure,\run +1} -\dstate_{\pure^*,\run+1} \leq 
-M_c -c\sum_{\runalt=1}^{\run}\step_{\runalt} 
\end{align}
for all $\run=\running$.  Assume ad absurdum  that there exists at least one strategy $\pure\neq\pure^*,\pure\in\pures$ such that $\limsup_{\run\to\infty}\policy_{\pure,\run} \geq \eps >0$.
for all sufficiently large $\run$. Recall also that for $\policy\in \nhd_c$, it holds that $\policy_{\pure^*}>0$ by construction.
Using \cref{lem:KKT} we get
\begin{align}
   \dstate_{\pure,\run +1} -\dstate_{\pure^*,\run+1}=  \policy_{\pure,\run+1} - \policy_{\pure^*,\run+1}  &\leq -M_c -c\sum_{\runalt=1}^{\run}\step_{\runalt}
\end{align}
Notice that the \acs{LHS} of this inequality is bounded, while the \acs{RHS} goes to $-\infty$, which is a contradiction. Thus, with probability at least $1-\conf$, $\policy_{\run}\to\policy^*$ as $\run\to\infty$.

We can rewrite the previous inequality as
\begin{equation}
    \policy_{\pure,\run+1} \leq 1 -M_c -c\sum_{\runalt=1}^{\run}\step_{\runalt} \quad\text{ for all }\quad\pure^*\neq\pure\in\pures
\end{equation}
Now aggregating over all strategies, on the previous inequality,  we get that 
\begin{align}
    \onenorm{\policy_{\run+1} -\np} = 2(1-\policy_{\pure^*,\run+1}) \leq2 \sum_{\pure^*\neq\pure\in\pures}(1- M_c -c\sum_{\runalt=1}^{\run}\step_{\runalt} )
\end{align}
Thus, once $\sum_{\runalt=1}^{\run}\step_{\runalt}$  becomes at least $(1-M_c)/c$, which occurs in finite time, the convergence is implied.
\end{proof}

\begin{proof}[Proof of \cref{cor:strict}]
For \cref{mod:full,mod:stoch}, taking $\bexp = \infty, \sexp = 0$ we readily get that $\pexp > 1/2$. Since we require that $\sum_{\run=1}^{\infty}\step_\run = \infty$, we obtain that $\pexp \in (1/2,1]$. 

For \cref{mod:value}, we have that $\curr[\bbound] = \bigoh(\curr[\expar])$ and $\curr[\sdev] = \bigoh(1/\sqrt{\curr[\expar]})$, \ie $\bexp = r$ and $\sexp = r/2$. Now, since $\pexp \leq 1$, $\pexp + \bexp >1$ and $\pexp - \sexp > 1/2$, we obtain that $\pexp \in (2/3,1]$. 
\end{proof}

%% file: App-MARL.tex

In this appendix we will establish the required properties for the gradient of the players' value function.
More precisely, we prove the following intermediate results:
\begin{itemize}[topsep=0pt,leftmargin=2ex]
\item In \cref{lem:well-defined-visitation-distribution} we prove that in the random stopping episodic framework visitation the notion of discounted state visitation distribution is well-defined.
\item In \cref{lem:conversion}, we prove the conversion lemma, a standard lemma that connects a sample by visitation distribution and a random trajectory.
\item In \cref{lem:GradientViaQ}, we establish different versions of Policy Gradient theorem via $Q$-value function for the random stopping episodic framework.
\item In \cref{lem:boundness,lem:smoothness}, we establish the boundedness and the Lipschitz smoothness of policy gradient vector field, i.e., 
$\gvalue(\policy)=(\gvalue_{\play}(\policy))_{\play\in\players} $
where
$\gvalue_{\play}(\policy) = \nabla_{\policy_{\play}} \rvalue_{\play,\rstate}(\policy)$
\end{itemize}

For a policy profile $\policy \in \policies$ and an arbitrary initial state distribution $\rstate_\tstart\sim\stateDist$, let's recall the definition of  discounted state visitation measure/distribution as
\begin{equation*}
		\visitMeas{\stateDist}{\policy}\parens{\rstate} =\ex_{\traj \sim \MDP}\bracks*{\insum_{\time =\tstart}^{\stopTime\parens{\traj}}\one\braces{\rstate_{\time}=\rstate} \Big| \rstate_\tstart \sim \stateDist},\quad	\visitDistr{\stateDist}{\policy}\parens{\rstate} \defeq \visitMeas{\stateDist}{\policy}\parens{\rstate}/\normConst{\stateDist}{\policy}
\end{equation*}
To begin with, we prove formally that the above definition is well-posed for the random stopping episodic framework described above, \ie
$\visitMeas{\stateDist}{\policy}\parens{\rstate}<\infty$,
so 
$\normConst{\stateDist}{\policy}\defeq\sum_{\rstate\in\states }\visitMeas{\stateDist}{\policy}\parens{\rstate}$
is well-defined.

\begin{lemma}\label{lem:well-defined-visitation-distribution}
For any $\rstate \in \states$, $\visitMeas{\stateDist}{\policy}\parens{\rstate} < \infty$ and $\normConst{\stateDist}{\policy}\leq \tfrac{1}{\stopPr}$.
\end{lemma}

\begin{proof}
For the sake of the proof, we define a new state $\rstate_f$, indicating that the game has stopped. In other words, we have that $\probof{\rstate_f \given \rstate, \pure} = \stopPr_{\rstate,\pure} \geq \stopPr >0$ for all $\pure \in \pures, \rstate \in \states$. Hence, for $\rstate \in \states$ we obtain:
\begin{align*}
\visitMeas{\stateDist}{\policy}\parens{\rstate} &=\ex_{\traj \sim \MDP}\bracks*{\insum_{\time =\tstart}^{\stopTime\parens{\traj}}\one\braces{\rstate_{\time}=\rstate} \Big| \rstate_\tstart \sim \stateDist}
	\notag\\
	&=\ex_{\traj \sim \MDP}\left[\sum_{\time=\tstart}^{\infty}\one\{\rstate_\time = \rstate, \rstate_{i} \neq \rstate_f,  1 \leq i \leq t\} \vert \rstate_\tstart \sim \stateDist\right]
	\notag\\
	&\leq \sum_{\rstate\in\states} \visitMeas{\stateDist}{\policy}\parens{\rstate}
	=\ex_{\traj \sim \MDP}\left[\sum_{\time=\tstart}^{\infty}\one\{\rstate_{i} \neq \rstate_f,  1 \leq i \leq \time\} \vert \rstate_\tstart \sim \stateDist\right]
	\notag\\
	&= \sum_{\time=\tstart}^{\infty}\probof{\rstate_{i} \neq \rstate_f,  1 \leq i \leq t \given \rstate_\tstart \sim \stateDist}
	\notag\\
	&= \sum_{\time=\tstart}^{\infty} \prod_{i = 1}^{\time} \probof{\rstate_{i} \neq \rstate_f \given \rstate_\tstart \sim \stateDist, \rstate_{j} \neq \rstate_f, 1 \leq j \leq i-1}
	\notag\\
	&\leq \sum_{\time=\tstart}^{\infty} (1-\stopPr)^{\time}\leq \frac{1}{\stopPr}
	< \infty.
	\qedhere
\end{align*}
\end{proof}

\ConversionLemma*

\begin{proof}
\begin{align}
	 \ex_{\traj \sim \MDP}\bracks*{\insum_{\time =\tstart}^{\stopTime\parens{\traj}}f( \rstate_\time, \pure_\time)}
 	&= \sum_{\time = \tstart}^{\infty}\sum_{\rstate \in \states}\sum_{\pure \in \pures} \ex_{\traj\sim\MDP}\left[\one\{\time\leq T(\traj), \rstate_\time = \rstate, \pure_\time = \pure\} f(\rstate, \pure) \right] \notag\\
	&= \sum_{\rstate \in \states}\sum_{\time = \tstart}^{\infty}\sum_{\pure \in \pures} \prob^{\policy}(\rstate = \rstate_\time \vert \rstate_{\tstart}\sim \stateDist) \policy(\pure \vert \rstate)f(\rstate, \pure) \notag\\
	&= \sum_{\rstate \in \states}\sum_{\time = \tstart}^{\infty}\prob^{\policy}(\rstate = \rstate_\time \vert \rstate_{\tstart}\sim \stateDist)\sum_{\pure \in \pures}  \policy(\pure \vert \rstate)f(\rstate, \pure) \notag\\
	&= \sum_{\rstate \in \states}\dpol_{\stateDist}(\rstate) \ex_{\pure \sim \policy(\cdot \vert \rstate)}\left[f(\rstate, \pure)\right]\notag\\
 	&= \normConst{\stateDist}{\policy}
  \ex_{\rstate\sim\visitDistr{\stateDist}{\policy}}\ex_{\pure\sim\policy(\cdot \vert \rstate)}\bracks*{f(\rstate, \pure)}
\end{align}
where $\normConst{\stateDist}{\policy}\defeq\ex_{\rstate\sim \unif(\states) }\bracks[\big]{\visitMeas{\stateDist}{\policy}\parens{\rstate}}\cdot\abs*{\states}$ is well-defined by \cref{lem:well-defined-visitation-distribution}.
\end{proof}

A compact reformulation the aforementioned lemma is via the matrix representation of the discounted visitation distribution:

\begin{lemma}[Conversion Lemma (Matrix form)]
\label{lem:conversion-matrix}
For an arbitrary state-action function $f\from\states\times \pures \to \R$ and a policy profile $\policy$, we have
\begin{equation}
\ex_{\traj \sim \MDP}\bracks*{\insum_{\time =\tstart}^{\stopTime\parens{\traj}}f(\rstate_\time, \pure_\time )\vert \pure_\tstart=\pure,\rstate_\tstart=\rstate}
	=e_{\rstate,\pure}^\top\visitMatrix(\policy)f
\end{equation}
where $\visitMatrix$ is a  discounted visitation distribution (action-state)-matrix under poliy profile $\policy$ \ie $[\visitMatrix(\policy)]_{\underbrace{(\pure,\rstate)}_{\text{Row Index}}\to\underbrace{(\pure',\rstate')}_{\text{Column Index}}}= \sum_{\time = \tstart}^{\infty}\prob^{\policy}(\rstate_\time=\rstate' , \pure_\time=\pure'\vert \rstate_{\tstart}=\rstate, \pure_\tstart=\pure ) $
\end{lemma}
\begin{proof}
By definition we have
\begin{align*}
e_{\rstate,\pure}^\top\visitMatrix(\policy)f
	&= \braket{e_{\rstate,\pure}^\top\visitMatrix(\policy)}{f}
	\notag\\
	&=\sum_{\rstate' \in \states}\sum_{\pure' \in \pures}
		\left(e_{\rstate,\pure}^\top\visitMatrix(\policy)\right)_{(\rstate',\pure')} \cdot f(\rstate',\pure')
	\notag\\
	&=\sum_{\rstate' \in \states}\sum_{\pure' \in \pures} 
		e_{\rstate,\pure}^\top\visitMatrix(\policy)e_{\rstate',\pure'} \cdot f(\rstate',\pure')
	\notag\\
	&=\sum_{\rstate' \in \states}\sum_{\pure' \in \pures} 
		\sum_{\time = \tstart}^{\infty}\prob^{\policy}(\rstate_\time=\rstate' , \pure_\time=\pure'\vert \rstate_{\tstart}=\rstate, \pure_\tstart=\pure )\cdot f(\rstate',\pure')
	\notag\\
	&= \sum_{\time = \tstart}^{\infty}\sum_{\rstate' \in \states}\sum_{\pure' \in \pures} \ex_{\traj\sim\MDP}\left[\one\{\time\leq T(\traj), \rstate_\time' = \rstate, \pure_\time' = \pure,
  \} f(\rstate, \pure) \vert  \rstate_{\tstart}=\rstate, \pure_\tstart=\pure 
   \right]
   \notag\\
	&=	 \ex_{\traj \sim \MDP}\bracks*{\insum_{\time =\tstart}^{\stopTime\parens{\traj}}f(\rstate_\time, \pure_\time )\vert \pure_\tstart=\pure,\rstate_\tstart=\rstate}.
\qedhere
\end{align*}
\end{proof}
\begin{remark}\label{rem:matrix-definitions}
Notice that $\visitMatrix$ is a well-defined matrix.
Indeed, let's us define $\transMatrix(\policy)$ as the state-action one step transition matrix: \begin{equation}[\transMatrix(\policy)]_{\underbrace{(\pure,\rstate)}_{\text{Row Index}}\to\underbrace{(\pure',\rstate')}_{\text{Column Index}}}= \prob^{\policy}(\rstate_\tafterstart=\rstate' , \pure_\tafterstart=\pure'\vert \rstate_{\tstart}=\rstate, \pure_\tstart=\pure )=\policy(\pure'\vert \rstate')P(\rstate'|\rstate,\pure).\end{equation}
Notice that $\transMatrix(\policy)$ is a substochastic matrix and therefore $\mathrm{spectral}(\transMatrix(\policy))<1$ or equivalently $
(I-\transMatrix(\policy))^{-1}$ is invertible. Thus using Neumann series we have that $(I-\transMatrix(\policy))^{-1}=\sum_{\time=\tstart}^{\infty} \transMatrix(\policy)^{\time}$. By induction, a folklore probabilistic-graph theoretic fact, we can show that $\sum_{\time=\tstart}^{\infty} \transMatrix(\policy)^{\time}=\visitMatrix(\policy)$.
\end{remark}
In order to analyze the gradient of MARL policy gradient methods, we will introduce the notions $Q,A$ and their per-player averages that are useful in the MDP analysis. 
\begin{definition}\label{def:Q&A}
For a state $\rstate \in \states$, a policy $\policy$ and $\pure = (\pure_{1},\dots,\pure_{\nPlayers}) \in \pures$, we define:
\begin{enumerate}[(i)]
	\item The $Q$-value function of player $\play$ as:
	\begin{equation}
		Q_{\play}^{\policy}(\rstate,\pure):= \ex_{\traj \sim \MDP(\policy|\rstate)}\left[ \sum_{\time = \tstart}^{T(\traj)} \rewards_{\play}(\rstate_{\time}(\traj), \pure_{\time}(\traj)) \vert \rstate_{\tstart} = \rstate, \pure_{\tstart} = \pure \right]
	\end{equation}
	\item The \emph{advantage function} of player $\play$ as:
	\begin{equation}
		\adv_{\play}^{\policy}(\rstate,\pure):= Q_{\play}^{\policy}(\rstate,\pure) - \rvalue_{\play,\rstate}(\policy)
	\end{equation}
\end{enumerate}
We also define $\bar{Q}_{\play}^{\policy}, \avgadv_{\play}^{\policy}$ to be the averaged for $\play$-th player single MDP $Q$-value and advantage functions:
\begin{enumerate}[(i)]
	\item The averaged $\bar{Q}_{\play}^{\policy}$-value function of player $\play$ as:
	\begin{equation}
		\bar{Q}_{\play}^{\policy}(\rstate,\pure_{\play}) := 
		 \ex_{\pure_{-\play} \sim \policy_{-\play}(\cdot \vert \rstate)}\left[ Q_{\play}^{\policy}\left(\rstate,(\pure_{\play};\pure_{-\play})\right) \right]
	\end{equation}
	\item The \emph{averaged advantage function} $\avgadv_{\play}^{\policy}$ of player $\play$ as:
	\begin{equation}
		\avgadv_{\play}^{\policy}(\rstate,\pure_{\play}) := \ex_{\pure_{-\play} \sim \policy_{-\play}(\cdot \vert \rstate)}\left[ \adv_{\play}^{\policy}\left(\rstate,(\pure_{\play};\pure_{-\play})\right) \right],
	\end{equation}
\end{enumerate}
\end{definition}
By \cref{rem:matrix-definitions}, we can rewrite the above notations using $\visitMatrix,\transMatrix$.

\begin{lemma}\label{lem:matrix-rewritting}
For a policy profile $\policy$, we have:
\begin{enumerate}
	\item $
			Q_{\play}^{\policy}(\rstate,\pure)=e_{\rstate,\pure}^\top\visitMatrix(\policy)\reward_\play
	$
	\item $\visitMeas{\stateDist}{\policy}\parens{\rstate} =\bracks*{\sum_{\rstate'\in\states}\stateDist(\rstate')\sum_{\pure'\in\pures}\policy(\pure'\vert\rstate')
		e_{\rstate',\pure'}}^\top\visitMatrix(\policy)\sum_{\pure\in\pures} e_{\rstate,\pure}$
\end{enumerate}
\end{lemma}
\begin{proof}
We separately have using \cref{lem:matrix-rewritting} and \cref{rem:matrix-definitions}.
\begin{enumerate}
\item
For our first claim, a straightforward calculation gives:
\begin{equation}
Q_{\play}^{\policy}(\rstate,\pure)= \ex_{\traj \sim \MDP(\policy|\rstate)}\left[ \sum_{\time = \tstart}^{T(\traj)} \rewards_{\play}(\rstate_{\time}(\traj), \pure_{\time}(\traj)) \vert \rstate_{\tstart} = \rstate, \pure_{\tstart} = \pure \right]=e_{\rstate,\pure}^\top\visitMatrix(\policy)\rewards_\play
\end{equation}
\item
As for the second part of the lemma, we have:
\begin{align*}
\visitMeas{\stateDist}{\policy}\parens{\rstate} 
	&=\ex_{\traj \sim \MDP}\bracks*{\insum_{\time =\tstart}^{\stopTime\parens{\traj}}
		\oneof{\rstate_{\time}=\rstate} \Big| \rstate_\tstart \sim \stateDist}
	\notag\\
	&=\ex_{\rstate' \sim \stateDist}\ex_{\traj \sim \MDP}
		\bracks*{\sum_{\time =\tstart}^{\stopTime\parens{\traj}}\sum_{\pure\in\pures}\one\braces{\rstate_{\time}=\rstate,\pure_{\time}=\pure} \Big| \rstate_\tstart=\rstate'}
	\notag\\
	&=\ex_{\rstate' \sim \stateDist}\ex_{\pure' \sim \policy(\cdot\vert\rstate)}\ex_{\traj \sim \MDP}
		\bracks*{\sum_{\time =\tstart}^{\stopTime\parens{\traj}}\sum_{\pure\in\pures}\one\braces{\rstate_{\time}=\rstate,\pure_{\time}=\pure} \Big| \rstate_\tstart=\rstate',\pure_\tstart=\pure'}
	\notag\\
	&=\ex_{\rstate' \sim \stateDist}\ex_{\pure' \sim \policy(\cdot\vert\rstate)}\bracks*{
		e_{\rstate',\pure'}^\top\visitMatrix(\policy)\sum_{\pure\in\pures} e_{\rstate,\pure}}
	\notag\\
	&=\bracks*{\sum_{\rstate'\in\states}\stateDist(\rstate')\sum_{\pure'\in\pures}\policy(\pure'\vert\rstate')
		e_{\rstate',\pure'}}^\top\visitMatrix(\policy)\sum_{\pure\in\pures} e_{\rstate,\pure}.
	\qedhere
\end{align*}
\end{enumerate}
\end{proof}

Having defined the above notions, we are ready to provide equivalent forms of $\gvalue(\policy)$ that will permit us to prove its boundedness and  smoothness.
We start with the following versions of the policy gradient theorem for random stopping setting:
\begin{lemma}\label{lem:GradientViaQ}
For the  independent gradient operator $\gvalue(\policy)$ per player the following expressions are equal to $\gvalue_{\play}(\policy)$:
\begin{enumerate}
	\item $\gvalue_{\play}(\policy) = \ex_{\traj \sim \MDP}\left[\sum_{t = 0}^{T(\traj)}  \nabla_{\play}\left( \log\policy_\play(\pure_{\play,\time}(\traj) \vert \rstate_\time(\traj))\right)\bar{Q}_{\play}^{\policy}(\rstate_\time(\traj),\pure_{\play,\time}(\traj))\right]$
	\item $\gvalue_{\play}(\policy) =\normConst{\stateDist}{\policy}
  \ex_{\rstate\sim\visitDistr{\stateDist}{\policy}}\ex_{\pure_\play\sim\policy_\play(\cdot \vert \rstate)}\bracks*{
  \nabla_{\play}\left( \log\policy_\play(\pure_{\play} \vert \rstate)\right)\bar{Q}_{\play}^{\policy}(\rstate,\pure_{\play})}
  $
  \item  $(\gvalue_{\play}(\policy))_{\pure_\play^\circ,\rstate^\circ}=\frac{\pd \rvalue_{\play,\stateDist}(\policy)}{\pd \policy_\play(\pure_\play^\circ \vert \rstate^\circ)} =\dpol_{\stateDist}(\rstate^\circ)
 \bar{Q}_{\play}^{\policy}(\rstate^\circ,\pure_{\play}^\circ)=\normConst{\stateDist}{\policy}\visitDistr{\stateDist}{\policy}(\rstate^\circ)
 \bar{Q}_{\play}^{\policy}(\rstate^\circ,\pure_{\play}^\circ)$
\end{enumerate}
\end{lemma}

\begin{proof}
Recall first that the independent gradient operator $\gvalue(\policy)$ is given by
\(
\gvalue_{\play}(\policy) = \nabla_{\play} \rvalue_{\play,\stateDist}(\policy)
\)
Accordingly, we will begin by showing that:
\begin{equation}
\nabla_{\play} \left(\rvalue_{\play,\stateDist}(\policy) \right)
	= \ex_{\traj \sim \MDP}\left[\sum_{t = 0}^{T(\traj)}  \nabla_{\play}\left( \log\policy_\play(\pure_{\play,\time}(\traj) \vert \rstate_\time(\traj))\right)\bar{Q}_{\play}^{\policy}(\rstate_\time(\traj),\pure_{\play,\time}(\traj))\right]
\end{equation}
To that end, we will start with an arbitrary $\rstate_{\tstart}$, and by linearity of $\nabla_{\policy_{\play}} (\cdot)$ and $\ex_{\rstate_{\tstart} \sim \stateDist}[\cdot]$, we will obtain the result.
Indeed, we have:
\begin{align}
\nabla_{\play} \left(\rvalue_{\play,\rstate_{\tstart}}(\policy) \right) &= \nabla_{\play}\left(\ex_{\traj}\left[\rewards_{\play}(\traj)\right]\right)
	\notag\\
	&= \nabla_{\play}\left(\ex_{\pure_{\play}\sim\policy_{\play}(\cdot \vert \rstate_{\tstart})}\left[\bar{Q}_{\play}^{\policy}(\rstate_{\tstart},\pure_\play) \right]\right)
	\notag\\
	&= \nabla_{\play}\left(\sum_{\pure_\play \in \pures_\play} \policy_\play(\pure_\play \vert \rstate_{\tstart})\bar{Q}_{\play}^{\policy}(\rstate_{\tstart},\pure_\play)\right)
	\notag\\
	&= \sum_{\pure_\play \in \pures_\play}\nabla_{\play}\left( \policy_\play(\pure_\play \vert \rstate_{\tstart})\right)\bar{Q}_{\play}^{\policy}(\rstate_{\tstart},\pure_\play) + \policy_\play(\pure_\play \vert \rstate_{\tstart})\nabla_{\play}\left(\bar{Q}_{\play}^{\policy}(\rstate_{\tstart},\pure_\play)\right)
	\notag\\
	&= \sum_{\pure_\play \in \pures_\play}\nabla_{\play}\left( \log\policy_\play(\pure_\play \vert \rstate_{\tstart})\right)\policy_\play(\pure_\play \vert \rstate_{\tstart})\bar{Q}_{\play}^{\policy}(\rstate_{\tstart},\pure_\play) + \policy_\play(\pure_\play \vert \rstate_{\tstart})\nabla_{\play}\left(\bar{Q}_{\play}^{\policy}(\rstate_{\tstart},\pure_\play)\right)
	\notag\\
	&= \ex_{\pure_{\play}\sim\policy_{\play}(\cdot \vert \rstate_{\tstart})}\left[\nabla_{\play}\left( \log\policy_\play(\pure_\play \vert \rstate_{\tstart})\right)\bar{Q}_{\play}^{\policy}(\rstate_{\tstart},\pure_\play) \right]
	\notag\\
		&\qquad + \sum_{\pure_\play \in \pures_\play}\policy_\play(\pure_\play \vert \rstate_{\tstart})\nabla_{\play}\left(\ex_{\pure_{-\play}\sim\policy_{-\play}(\cdot \vert \rstate_{\tstart})}\left[\rewards_{\play}(\rstate_{\tstart},\pure) + \sum_{\rstate_1 \in \states}P(\rstate_1 \vert \rstate_{\tstart},\pure)\rvalue_{\play,\rstate_1}(\policy) \right] \right)
	\notag\\
	&= \ex_{\pure_{\play}\sim\policy_{\play}(\cdot \vert \rstate_{\tstart})}\left[\nabla_{\play}\left( \log\policy_\play(\pure_\play \vert \rstate_{\tstart})\right)\bar{Q}_{\play}^{\policy}(\rstate_{\tstart},\pure_\play) \right]
	\notag\\
		&\qquad +\sum_{\pure_\play \in \pures_\play}\policy_\play(\pure_\play \vert \rstate_{\tstart})\ex_{\pure_{-\play}\sim\policy_{-\play}(\cdot \vert \rstate_{\tstart})}\left[\sum_{\rstate_1 \in \states}P(\rstate_1 \vert \rstate_{\tstart},\pure)\nabla_{\play}\left(\rvalue_{\play,\rstate_1}(\policy)\right)\right] \notag\\
	&= \ex_{\pure_{\play}\sim\policy_{\play}(\cdot \vert \rstate_{\tstart})}\left[\nabla_{\play}\left( \log\policy_\play(\pure_\play \vert \rstate_{\tstart})\right)\bar{Q}_{\play}^{\policy}(\rstate_{\tstart},\pure_\play) \right]
	\notag\\
		&\qquad + \ex_{\pure \sim\policy(\cdot \vert \rstate_{\tstart})}\left[\sum_{\rstate_1 \in \states}P(\rstate_1 \vert \rstate_{\tstart},\pure)\nabla_{\play}\left(\rvalue_{\play,\rstate_1}(\policy)\right)\right] \end{align}
Thus, we can rewrite it as:
\begin{align}
\nabla_{\play} \left(\rvalue_{\play,\rstate_{\tstart}}(\policy) \right)
	&= \ex_{\pure_{\play}\sim\policy_{\play}(\cdot \vert \rstate_{\tstart})}\left[\nabla_{\play}\left( \log\policy_\play(\pure_\play \vert \rstate_{\tstart})\right)\bar{Q}_{\play}^{\policy}(\rstate_{\tstart},\pure_\play) \right] \notag\\
		&\qquad + \ex_{\pure \sim\policy(\cdot \vert \rstate_{\tstart})}\left[\sum_{\rstate_1 \in \states}P(\rstate_1 \vert \rstate_{\tstart},\pure)\nabla_{\play}\left(\rvalue_{\play,\rstate_1}(\policy)\right)\right]
	\notag\\
	&= \ex_{\traj \sim \MDP(\policy|\rstate_{\tstart})} \left[\nabla_{\play}\left( \log\policy_\play(\pure_{\play,0}(\traj) \vert \rstate_{\tstart})\right)\bar{Q}_{\play}^{\policy}(\rstate_{\tstart},\pure_{\play,0}(\traj)) \right]
	\notag\\
		&\qquad + \ex_{\traj \sim \MDP(\policy|\rstate_{\tstart})}\left[\one\left\{T(\traj) \geq 1 \right\}\nabla_{\play}\left(\rvalue_{\play,\rstate_1(\traj)}(\policy)\right)\right]
	\notag\\
	&= \sum_{t = 0}^{\infty} \ex_{\traj \sim \MDP(\policy|\rstate_{\tstart})}\left[\one\{\time \leq T(\traj)\} \nabla_{\play}\left( \log\policy_\play(\pure_{\play,\time}(\traj) \vert \rstate_\time(\traj))\right)\bar{Q}_{\play}^{\policy}(\rstate_\time(\traj),\pure_{\play,\time}(\traj))\right]
	\notag\\
		&\qquad + \ex_{\traj \sim \MDP(\policy|\rstate_{\tstart})}\left[\one\{T(\traj) = \infty\} A_{\infty} \right]
	\notag\\
	&\stackrel{(a)}{=} \ex_{\traj \sim \MDP(\policy|\rstate_{\tstart})}\left[\sum_{t = 0}^{T(\traj)}  \nabla_{\play}\left( \log\policy_\play(\pure_{\play,\time}(\traj) \vert \rstate_\time(\traj))\right)\bar{Q}_{\play}^{\policy}(\rstate_\time(\traj),\pure_{\play,\time}(\traj))\right]
\end{align}
where $(a)$ holds because $\probof{T(\traj) = \infty} = 0$, and $A_{\infty}$ is some limiting quantity. 

Hence,we readily obtain:
\begin{equation}
	\nabla_{\play} \left(\rvalue_{\play,\stateDist}(\policy) \right) = \ex_{\rstate_{\tstart} \sim \stateDist}\left[\nabla_{\play} \left(\rvalue_{\play,\rstate_{\tstart}}(\policy) \right) \right]
\end{equation}
Now by \cref{lem:conversion}, we further have
\begin{align}
\nabla_{\play} \left(\rvalue_{\play,\stateDist}(\policy) \right)
	&=  \normConst{\stateDist}{\policy}
		\ex_{\rstate\sim\visitDistr{\stateDist}{\policy}}\ex_{\pure\sim\policy(\cdot \vert \rstate)}
			\bracks*{\nabla_{\play}\left( \log\policy_\play(\pure_{\play} \vert \rstate)\right)\bar{Q}_{\play}^{\policy}(\rstate,\pure_{\play})}
	\notag\\
 	&=\normConst{\stateDist}{\policy}
	\ex_{\rstate\sim\visitDistr{\stateDist}{\policy}}\ex_{\pure_\play\sim\policy_\play(\cdot \vert \rstate)}
		\bracks*{\nabla_{\play}\left( \log\policy_\play(\pure_{\play} \vert \rstate)\right)\bar{Q}_{\play}^{\policy}(\rstate,\pure_{\play})}
\end{align}
Decoupling $\nabla_\play$ per a state $\rstate^\circ$ and action $\pure_\play^\circ$, we get
\begin{align*}
\frac{\pd \rvalue_{\play,\stateDist}(\policy)}{\pd \policy_\play(\pure_\play^\circ \vert \rstate^\circ)}
	&= \normConst{\stateDist}{\policy}
		\ex_{\rstate\sim\visitDistr{\stateDist}{\policy}}\ex_{\pure_\play\sim\policy_\play(\cdot \vert \rstate)}
		\bracks*{\frac{\pd  \left( \log\policy_\play(\pure_{\play} \vert \rstate)\right)}{\pd \policy_\play(\pure_\play^\circ \vert \rstate^\circ)} \bar{Q}_{\play}^{\policy}(\rstate,\pure_{\play})}
	\notag\\
	&= \normConst{\stateDist}{\policy}
	\ex_{\rstate\sim\visitDistr{\stateDist}{\policy}}\ex_{\pure_\play\sim\policy_\play(\cdot \vert \rstate)}
		\bracks*{\one\{\pure_\play^\circ=\pure_\play, \rstate^\circ=\rstate\} \frac{1}{\policy_\play(\pure_\play^\circ \vert \rstate^\circ)} \bar{Q}_{\play}^{\policy}(\rstate^\circ,\pure_{\play}^\circ)}
	\notag\\
	&= \sum_{\rstate \in \states}\dpol_{\stateDist}(\rstate)
		\sum_{\pure_\play \in \pures_\play} \policy_\play(\pure_{\play} \vert \rstate)
			\one\{\pure_\play^\circ=\pure_\play, \rstate^\circ=\rstate\}
			\frac{1}{\policy_\play(\pure_\play^\circ \vert \rstate^\circ)} \bar{Q}_{\play}^{\policy}(\rstate^\circ,\pure_{\play}^\circ)
	\notag\\
	&=\dpol_{\stateDist}(\rstate^\circ) \bar{Q}_{\play}^{\policy}(\rstate^\circ,\pure_{\play}^\circ)
		= \normConst{\stateDist}{\policy}\visitDistr{\stateDist}{\policy}(\rstate^\circ)
 \bar{Q}_{\play}^{\policy}(\rstate^\circ,\pure_{\play}^\circ).
\qedhere
\end{align*}
\end{proof}

We are ready to bound the amplitude of the independent player gradient operator:
\begin{lemma}\label{lem:boundness}
For a given initial state distribution $\stateDist$, the independent player policy gradient operator $\gvalue(\policy)$ is bounded. More precisely, \begin{equation}\norm{\gvalue_\play(\policy)}\leq \frac{\sqrt{\nPures_{\play}}}{\stopPr^2} \quad \&\quad \norm{\gvalue(\policy)}\leq \frac{\sum_{\play\in\players}\sqrt{\nPures_{\play}}}{\stopPr^2}\end{equation}
\end{lemma}
\begin{proof}
We start by analyzing $\norm{\gvalue_\play(\policy)}^2$ using \cref{lem:GradientViaQ}.
Specifically, we have:
\begin{align}
\norm{\gvalue_\play(\policy)}^2&=\sum_{\pure_\play^\circ,\rstate^\circ,\in \pures_\play,\states}
	(\gvalue_{\play}(\policy)_{\pure_\play^\circ,\rstate^\circ})^2
	=\sum_{\rstate^\circ\in\states}\sum_{\pure_\play^\circ\in\pures_\play}\parens{\frac{\pd \rvalue_{\play,\stateDist}(\policy)}{\pd \policy_\play(\pure_\play^\circ \vert \rstate^\circ)}}^2
	\notag\\
	&=\sum_{\rstate^\circ\in\states}\sum_{\pure_\play^\circ\in\pures_\play}
	\parens{\normConst{\stateDist}{\policy}\visitDistr{\stateDist}{\policy}(\rstate^\circ) \bar{Q}_{\play}^{\policy}(\rstate^\circ,\pure_{\play}^\circ)}^2
	\notag\\
	&\leq (\normConst{\stateDist}{\policy})^2 \max_{\pure_\play^\circ,\rstate^\circ,\in \pures_\play,\states} \parens{\bar{Q}_{\play}^{\policy}(\rstate^\circ,\pure_{\play}^\circ)}^2\sum_{\rstate^\circ\in\states}\sum_{\pure_\play^\circ\in\pures_\play}\visitDistr{\stateDist}{\policy}(\rstate^\circ)^2
	\notag\\
	&\leq \frac{1}{\stopPr^2}\max_{\pure_\play^\circ,\rstate^\circ,\in \pures_\play,\states}
		\parens{\ex_{\pure_{-\play} \sim \policy_{-\play}(\cdot \vert \rstate)}\left[ Q_{\play}^{\policy}\left(\rstate^\circ,(\pure_{\play}^\circ;\pure_{-\play})\right) \right]}^2\sum_{\rstate^\circ\in\states}\sum_{\pure_\play^\circ\in\pures_\play}\visitDistr{\stateDist}{\policy}(\rstate^\circ)
	\notag\\
	&\leq \frac{1}{\stopPr^2}\max_{\pure^\circ,\rstate^\circ,\in \pures,\states} \parens{Q_{\play}^{\policy}(\rstate^\circ,\pure^\circ)}^2\sum_{\pure_\play^\circ\in\pures_\play}\sum_{\rstate^\circ\in\states}\visitDistr{\stateDist}{\policy}(\rstate^\circ)
	\notag\\
	&\leq \frac{1}{\stopPr^2}\max_{\pure^\circ,\rstate^\circ,\in \pures,\states} \left(
		\ex_{\traj \sim \MDP(\policy|\rstate)}\left[ \sum_{\time = \tstart}^{T(\traj)} \rewards_{\play}(\rstate_{\time}(\traj), \pure_{\time}(\traj)) \vert \rstate_{\tstart} = \rstate^\circ, \pure_{\tstart} = \pure^\circ \right]
	\right)^2
	\abs{\pures_{\play}}
	\notag\\
	&\leq \frac{\nPures_{\play}}{\stopPr^2}
		\left(\ex_{\traj \sim \MDP(\policy|\rstate)}\left[ \sum_{\time = \tstart}^{T(\traj)} 1 \vert \rstate_{\tstart} = \rstate^\circ, \pure_{\tstart} = \pure^\circ \right]
	\right)^2
	\leq \frac{\nPures_{\play}}{\stopPr^4}
\end{align}
We thus conclude that
\begin{equation*}
\norm{\gvalue_\play(\policy)}
	\leq \frac{\sqrt{\nPures_{\play}}}{\stopPr^2}
	\quad
	\text{and}
	\quad
\norm{\gvalue(\policy)}
	\leq \frac{\sum_{\play\in\players}\sqrt{\nPures_{\play}}}{\stopPr^2}.
	\qedhere
\end{equation*}
\end{proof}

To prove the smoothness of the policy gradient operator, we have first to establish the performance lemma for our setting. Respectively, we get
\begin{lemma}[Performance lemma]\label{lem:performance-lemma}
For any pair of policy profiles $\policy=(\policy_{\play},\policy_{-\play}),\policy'=(\policy_{\play}',\policy_{-\play}')$, it holds
\begin{equation}
		\rvalue_{\play,\stateDist}(\policy_{\play},\policy_{-\play}) -  \rvalue_{\play,\stateDist}(\policy_{\play}',\policy_{-\play}') = \ex_{\traj \sim \MDP(\policy|\stateDist)}\left[\sum_{\time = \tstart}^{T(\traj)} \adv_{\play}^{\policy_\play',\policy_{-\play}'}(\rstate_\time,\pure_\time) \right]
\end{equation}
where $\MDP(\policy|\stateDist)$ signifies that players follow $\policy$ as policy profile with $\stateDist$ as the initial state distribution. 
\end{lemma}
\begin{proof}
We will initial prove the aforementioned result for an arbitrary deterministic initial state $\rstate_{\tstart}=\rstate$:
\begin{align}
	\rvalue_{\play,\rstate}(\policy) -  \rvalue_{\play,\rstate}(\policy') &= 
	\ex_{\traj \sim \MDP(\policy|\stateDist)}\left[\sum_{\time = \tstart}^{T(\traj)} \rewards_{\play}(\rstate_\time,\pure_\time)\right] - \rvalue_{\play,\rstate}(\policy')
	\notag\\
	&= \ex_{\traj \sim \MDP(\policy|\rstate)}\left[\sum_{\time = \tstart}^{T(\traj)}\left(\rewards_{\play}(\rstate_\time,\pure_\time)+\rvalue_{\play,\rstate_\time}(\policy') - \rvalue_{\play,\rstate_\time}(\policy')\right) \right] - \rvalue_{\play,\rstate}(\policy')
	\notag\\
	&= \ex_{\traj \sim \MDP(\policy|\rstate)}\left[\sum_{\time = \tstart}^{T(\traj)}\rewards_{\play}(\rstate_\time,\pure_\time)+\sum_{\time = \tstart}^{T(\traj)}\left(\rvalue_{\play,\rstate_\time}(\policy') - \rvalue_{\play,\rstate}(\policy') - \rvalue_{\play,\rstate_\time}(\policy')\right) \right]
	\notag\\
	&= \ex_{\traj \sim \MDP(\policy|\rstate)}\left[\sum_{\time = \tstart}^{T(\traj)}\left(\rewards_{\play}(\rstate_\time,\pure_\time)+ \one\{T(\traj) \geq \time+1\}\rvalue_{\play,\rstate_{\time+1}}(\policy')\right) - \rvalue_{\play,\rstate_\time}(\policy')\right]
	\notag\\
	&= \ex_{\traj \sim \MDP(\policy|\rstate)}\left[\sum_{\time = \tstart}^{T(\traj)}\left(Q_{\play}^{\policy'}(\rstate_\time,\pure_\time) - \rvalue_{\play,\rstate_\time}(\policy')\right) \right]
	\notag\\
	&= \ex_{\traj \sim \MDP(\policy|\rstate)}\left[\sum_{\time = \tstart}^{T(\traj)}\adv_{\play}^{\policy'}(\rstate_\time,\pure_\time) \right]
\end{align}

where in the last equation we recall the definition of the Advantage function and in the pre-last the equivalent definitions of $Q_{\play}^{\policy}(\rstate,\pure)$
\begin{align}
	Q_{\play}^{\policy}(\rstate,\pure) &= \ex_{\traj \sim \MDP(\policy|\rstate)}\left[ \sum_{\time = \tstart}^{T(\traj)} \rewards_{\play}(\rstate_{\time}(\traj), \pure_{\time}(\traj)) \vert \rstate_{\tstart} = \rstate, \pure_{\tstart} = \pure \right] \notag\\
	&= \rewards_{\play}(\rstate,\pure) + \ex_{\traj \sim \MDP(\policy|\rstate)}\left[\one\{T(\traj) \geq \tafterstart\}\rvalue_{\play,\rstate_1}(\policy) \vert \rstate_{\tstart} = \rstate, \pure_{0} = \pure \right]
\end{align}
Applying the linearity of $\ex_{\rstate\sim\stateDist}\bracks*{\cdot}$, we get the desired result:
\begin{equation}
		\rvalue_{\play,\stateDist}(\policy) -  \rvalue_{\play,\stateDist}(\policy') = \ex_{\traj \sim \MDP(\policy|\stateDist)}\left[\sum_{\time = \tstart}^{T(\traj)} \adv_{\play}^{\policy'}(\rstate_\time,\pure_\time) \right]=\normConst{\stateDist}{\policy}
  \ex_{\rstate\sim\visitDistr{\stateDist}{\policy}}\ex_{\pure\sim\policy(\cdot \vert \rstate)}\bracks*{\adv_{\play}^{\policy'}(\rstate,\pure) }
\end{equation}
where the last expression comes from \cref{lem:conversion}.
\end{proof}
%
Before closing this section by proving the Lipschitz-smoothness of our operator, we describe a useful observation that will be helpful in the smoothness bounds.

\begin{proposition}
\label{prop:useful-dtv}
For any pair of policy profiles $\policy=(\policy_{\play},\policy_{-\play}),\policy'=(\policy_{\play}',\policy_{-\play}')$ and an arbitrary initial state distribution $\stateDist$ and a subset $\mathcal{M}\subseteq\players$, it holds that:
\begin{equation}\sum_{\rstate}
\visitDistr{\stateDist}{\policy_{}}(\rstate)\sum_{\pure_{\mathcal{M}}}|(\policy_{\mathcal{M}}-\policy_{\mathcal{M}}^{'})(\pure_{\mathcal{M}} \vert \rstate)|\leq \sum_{\play\in\mathcal{M}}\sqrt{\nPures_{\play}}\norm{\policy_\play-\policy'_\play} \end{equation}
where $\policy_{\mathcal{M}}=(\policy_\play)_{\play\in\mathcal{M}} $ and $\pure_{\mathcal{M}}=(\pure_\play)_{\play\in\mathcal{M}} $, correspondingly.
\end{proposition}

\begin{proof}
A series of direct calculations gives:
\begin{align}
\sum_{\rstate}
\visitDistr{\stateDist}{\policy_{}}(\rstate)\sum_{\pure_{\mathcal{M}}}|(\policy_{\mathcal{M}}-\policy_{\mathcal{M}}^{'})(\pure_{\mathcal{M}} \vert \rstate)|
	&= 2\sum_{\rstate}
\visitDistr{\stateDist}{\policy_{}}(\rstate)\frac{1}{2}\norm{(\policy_{\mathcal{M}}-\policy_{\mathcal{M}}^{'})}_1
	\notag\\
	&= 2\sum_{\rstate}
		\visitDistr{\stateDist}{\policy_{}}(\rstate)\frac{1}{2}d_{\mathrm{TV}}{(\policy_{\mathcal{M}}(\cdot|\rstate),\policy_{\mathcal{M}}^{'}(\cdot|\rstate))}
	\notag\\
	&\leq 2\sum_{\rstate}
\visitDistr{\stateDist}{\policy_{}}(\rstate)\sum_{\play\in\mathcal{M}}\frac{1}{2}d_{\mathrm{TV}}{(\policy_{\play}(\cdot|\rstate),\policy_{\play}^{'}(\cdot|\rstate))}
	\notag\\
	&=\sum_{\rstate} \visitDistr{\stateDist}{\policy_{}}(\rstate)\sum_{\play\in\mathcal{M}}\norm{(\policy_{\play}(\cdot|\rstate)-\policy_{\play}^{'}(\cdot|\rstate))}_1
	\notag\\
	&=\sum_{\rstate}
\visitDistr{\stateDist}{\policy_{}}(\rstate)\sum_{\play\in\mathcal{M}}\sqrt{\nPures_{\play}}\norm{\policy_\play-\policy'_\play}_2
	\notag\\
	&=\sum_{\play\in\mathcal{M}}\sqrt{\nPures_{\play}}\norm{\policy_\play-\policy'_\play}_2(\sum_{\rstate}
\visitDistr{\stateDist}{\policy_{}}(\rstate))
	\notag\\
	&=\sum_{\play\in\mathcal{M}}\sqrt{\nPures_{\play}}\norm{\policy_\play-\policy'_\play}_2
\end{align}
where $d_{\mathrm{TV}}$ corresponds to the total variation distance,
and the first inequality is a consequence of the triangle inequality for $d_{\mathrm{TV}}$.
\end{proof}
\begin{lemma}
\label{lem:smoothness}
For a given initial state distribution $\stateDist$, the independent player policy gradient operator $\gvalue(\policy)$ is Lipschitz continuous. More precisely, for any pair of policy profiles $\policy=(\policy_{\play},\policy_{-\play}),\policy'=(\policy_{\play}',\policy_{-\play}')$, it holds
\begin{equation}\norm{\gvalue_\play(\policy)-\gvalue_{\play}(\policy')}=
\norm{\nabla_\play  (\rvalue_{\play,\stateDist}(\policy_{})- 
\nabla_\play  (\rvalue_{\play,\stateDist}(\policy_{}^{'})} \leq \frac{3\sqrt{\nPures_{\play}}}{\stopPr^3}\sum_{\playalt=1}^{\nPlayers}\sqrt{\nPures_{\playalt}}\norm{\policy_\playalt-\policy'_\playalt}\quad \forall \play\in\players
\end{equation}
and consequently,
\begin{equation}\norm{\gvalue(\policy)-\gvalue(\policy')} \leq \frac{3\nPures}{\stopPr^3}\norm{\policy-\policy'}\end{equation}
\end{lemma}

\begin{proof}
For the proof, we will follow the approach of \citet{zhang2021gradient} and \citet{AgarwalKLM20}.
Our first task is to bound the directional derivative of the $\play$-th player's value function. 
To that end, let $\policy,\policy'\in\policies$ and $\pertrubation \in \states\times\pures$ such that $\norm{\pertrubation}=1$, and consider $\lambda$-perturbed policies
\begin{equation}
\begin{aligned}
\policy_{\lambda}^{\this}(\pure\vert\rstate)
	&=(\policy_\play+\lambda\pertrubation,\policy_{-\play})
	\\
\policy_{\lambda}^{\that}(\pure\vert\rstate)
	&=(\policy_\play'+\lambda\pertrubation,\policy_{-\play}')
\end{aligned}
\end{equation}
We then have:
\begin{subequations}
\begin{align}
\abs*{\frac{\pd \rvalue_{\play,\stateDist}(\policy_{\lambda}^{\this})}{\pd \lambda}-
	\frac{\pd \rvalue_{\play,\stateDist}(\policy_{\lambda}^{\that})}{\pd \lambda}}
	&=\abs*{\frac{\pd \rvalue_{\play,\stateDist}(\policy_{\lambda}^{\this})- \rvalue_{\play,\stateDist}(\policy_{\lambda}^{\that})}{\pd \lambda}}
	=\abs*{\frac{\pd\left( \rvalue_{\play,\stateDist}(\policy_{\lambda}^{\this})- \rvalue_{\play,\stateDist}(\policy_{\lambda}^{\that})\right)}{\pd \lambda}}
	\notag\\
	&=\abs*{\frac{\pd\left(
	\normConst{\stateDist}{\policy_{\lambda}^{\this}}
  \ex_{\rstate\sim\visitDistr{\stateDist}{\policy_{\lambda}^{\this}}}\ex_{\pure\sim\policy_{\lambda}^{\this}(\cdot \vert \rstate)}\bracks*{\adv_{\play}^{\policy_{\lambda}^{\that}}(\rstate,\pure) }
	\right)}{\pd \lambda}}
	\label{eq:performance_application}\\
	&= \abs*{\frac{\pd\left(\normConst{\stateDist}{\policy_{\lambda}^{\this}}
  \sum_{\rstate,\pure}\visitDistr{\stateDist}{\policy_{\lambda}^{\this}}(\rstate)(\policy_{\lambda}^{\this}-\policy_{\lambda}^{\that})(\pure \vert \rstate)\adv_{\play}^{\policy_{\lambda}^{\that}}(\rstate,\pure)
	\right)}{\pd \lambda}}
	\label{eq:zero-of-advantage}\\
	&= \abs*{\frac{\pd\left(\normConst{\stateDist}{\policy_{\lambda}^{\this}}
  \sum_{\rstate,\pure}\visitDistr{\stateDist}{\policy_{\lambda}^{\this}}(\rstate)(\policy_{\lambda}^{\this}-\policy_{\lambda}^{\that})(\pure \vert \rstate)Q_{\play}^{\policy_{\lambda}^{\that}}(\rstate,\pure)
	\right)}{\pd \lambda}}
	\notag\\
	&= \abs*{\frac{\pd\left(
  \sum_{\rstate,\pure}\visitMeas{\stateDist}{\policy_{\lambda}^{\this}}(\rstate)(\policy_{\lambda}^{\this}-\policy_{\lambda}^{\that})(\pure \vert \rstate)Q_{\play}^{\policy_{\lambda}^{\that}}(\rstate,\pure)
	\right)}{\pd \lambda}}
\end{align}
\end{subequations}
where \eqref{eq:performance_application} follows from \cref{lem:performance-lemma} and \eqref{eq:zero-of-advantage} uses the fact that $\sum_{\pure\in\pures}\policy(\pure\vert\rstate)\adv_{\play}^{\policy}(\rstate,\pure)=$, for all $\rstate\in\states$ and the last one is derived by the definition $\visitDistr{\stateDist}{\policy}\parens{\rstate} \defeq \visitMeas{\stateDist}{\policy}\parens{\rstate}/\normConst{\stateDist}{\policy}$.

By triangular inequality, the linearity of $\pd$ operator and \cref{lem:well-defined-visitation-distribution}, we have:
\begin{align}
		\left|\frac{\pd (\rvalue_{\play,\stateDist}(\policy_{\lambda}^{\this})-\rvalue_{\play,\stateDist}(\policy_{\lambda}^{\that}))}{\pd \lambda}\Big|_{\lambda=0}\right|
	&\leq
	\abs*{\sum_{\rstate,\pure}\frac{
  \pd\visitMeas{\stateDist}{\policy_{\lambda}^{\this}}(\rstate)
  }{\pd \lambda}\Big|_{\lambda=0}
  (\policy_{}-\policy_{}^{'})(\pure \vert \rstate)Q_{\play}^{\policy_{}^{'}}(\rstate,\pure)}
	\notag\\
	&+
	\normConst{\stateDist}{\policy_{}^{\this}}\abs*{  \sum_{\rstate,\pure}
\visitDistr{\stateDist}{\policy_{}}(\rstate)\frac{\pd(\policy_{\lambda}^{\this}-\policy_{\lambda}^{\that})(\pure \vert \rstate)}{\pd \lambda}\Big|_{\lambda=0}Q_{\play}^{\policy_{}^{'}}(\rstate,\pure)}
	\notag\\
	&+ 
	\normConst{\stateDist}{\policy_{}^{\this}}\abs*{\sum_{\rstate,\pure}
  \visitDistr{\stateDist}{\policy_{}}(\rstate)(\policy_{}-\policy_{}^{'})(\pure \vert \rstate)\frac{\pd Q_{\play}^{\policy_{\lambda}^{\that}}(\rstate,\pure)
	}{\pd \lambda}\Big|_{\lambda=0}}
\end{align}
We will bound the following three terms separately:
\begin{equation}
\begin{aligned}
\mathrm{Term}_A
	&= \abs*{\sum_{\rstate,\pure}\frac{
  \pd\visitMeas{\stateDist}{\policy_{\lambda}^{\this}}(\rstate)
  }{\pd \lambda}\Big|_{\lambda=0}
  (\policy_{}-\policy_{}^{'})(\pure \vert \rstate)Q_{\play}^{\policy_{}^{'}}(\rstate,\pure)}
	\\
\mathrm{Term}_B
	&= \abs*{  \sum_{\rstate,\pure}
\visitDistr{\stateDist}{\policy_{}}(\rstate)\frac{\pd(\policy_{\lambda}^{\this}-\policy_{\lambda}^{\that})(\pure \vert \rstate)}{\pd \lambda}\Big|_{\lambda=0}Q_{\play}^{\policy_{}^{'}}(\rstate,\pure)}
	\\
\mathrm{Term}_C
	&= \abs*{\sum_{\rstate,\pure}
  \visitDistr{\stateDist}{\policy_{}}(\rstate)(\policy_{}-\policy_{}^{'})(\pure \vert \rstate)\frac{\pd Q_{\play}^{\policy_{\lambda}^{\that}}(\rstate,\pure)}{\pd \lambda}\Big|_{\lambda=0}}
\end{aligned}
\end{equation}
For $\mathrm{Term}_A$, we will use \cref{lem:matrix-rewritting} in order to compute compactly the derivative:
\begin{align}
\frac{\pd\visitMeas{\stateDist}{\policy_{\lambda}^{\this}}(\rstate)}{\pd \lambda}
	&= \frac{\pd\parens*{\bracks*{\sum_{\rstate'\in\states}\stateDist(\rstate')\sum_{\pure'\in\pures}\policy_{\lambda}^{\this}(\pure'\vert\rstate')e_{\rstate',\pure'}}^\top\visitMatrix(\policy_{\lambda}^{\this})\sum_{\pure\in\pures} e_{\rstate,\pure}}}{\pd \lambda}
	\notag\\
	&= \parens*{\;\bracks*{\sum_{\rstate'\in\states}\stateDist(\rstate')\sum_{\pure'\in\pures}\frac{\pd\policy_{\lambda}^{\this}(\pure'\vert\rstate')}{\pd \lambda}e_{\rstate',\pure'}}^\top\visitMatrix(\policy_{\lambda}^{\this})\sum_{\pure\in\pures} e_{\rstate,\pure}}
	\notag\\
	&\qquad
		+ \parens*{\;\bracks*{\sum_{\rstate'\in\states}\stateDist(\rstate')\sum_{\pure'\in\pures}\policy_{\lambda}^{\this}(\pure'\vert\rstate')e_{\rstate',\pure'}}^\top\frac{\pd\visitMatrix(\policy_{\lambda}^{\this})}{\pd \lambda}\sum_{\pure\in\pures} e_{\rstate,\pure}}
	\notag\\
	&= \parens*{\;\bracks*{\sum_{\rstate'\in\states}\stateDist(\rstate')\sum_{\pure'\in\pures}{\pertrubation(\pure'_\play\vert\rstate')\cdot\policy_{-\play}(\pure'_{-\play}\vert\rstate')}e_{\rstate',\pure'}}^\top\visitMatrix(\policy_{\lambda}^{\this})\sum_{\pure\in\pures} e_{\rstate,\pure}}
	\notag\\
	&\qquad
		+ \parens*{\;\bracks*{\sum_{\rstate'\in\states}\stateDist(\rstate')\sum_{\pure'\in\pures}\policy_{\lambda}^{\this}(\pure'\vert\rstate')e_{\rstate',\pure'}}^\top\frac{\pd
  (I-\transMatrix(\policy_{\lambda}^{\this}))^{-1}}{\pd \lambda}\sum_{\pure\in\pures} e_{\rstate,\pure}}
  \notag\\
	&=	\parens*{\;\bracks*{\sum_{\rstate'\in\states}\stateDist(\rstate')\sum_{\pure'\in\pures}{\pertrubation(\pure'_\play\vert\rstate')\cdot\policy_{-\play}(\pure'_{-\play}\vert\rstate')}e_{\rstate',\pure'}}^\top\visitMatrix(\policy_{\lambda}^{\this})\sum_{\pure\in\pures} e_{\rstate,\pure}}
	\notag\\
	&\qquad
		+ \parens*{\;\bracks*{\sum_{\rstate'\in\states}\stateDist(\rstate')\sum_{\pure'\in\pures}\policy_{\lambda}^{\this}(\pure'\vert\rstate')e_{\rstate',\pure'}}^\top
  \big( 
  \visitMatrix(\policy_{\lambda}^{\this})
  \frac{\pd
  \transMatrix(\policy_{\lambda}^{\this})}{\pd \lambda}
  \visitMatrix(\policy_{\lambda}^{\this})
  \big)
  \sum_{\pure\in\pures} e_{\rstate,\pure}}
\end{align}
Thus for $\lambda=0$, we get
\begin{align}
\frac{\pd\visitMeas{\stateDist}{\policy_{\lambda}^{\this}}(\rstate)}{\pd \lambda}\Big|_{\lambda=0}
	&=	\parens*{\;\bracks*{\sum_{\rstate'\in\states}\stateDist(\rstate')\sum_{\pure'\in\pures}{\pertrubation(\pure'_\play\vert\rstate')\cdot\policy_{-\play}(\pure'_{-\play}\vert\rstate')}e_{\rstate',\pure'}}^\top\visitMatrix(\policy_{})\sum_{\pure\in\pures} e_{\rstate,\pure}}
	\notag\\
	&+ \parens*{\;\bracks*{\sum_{\rstate'\in\states}\stateDist(\rstate')\sum_{\pure'\in\pures}\policy_{}(\pure'\vert\rstate')e_{\rstate',\pure'}}^\top
  \big( 
  \visitMatrix(\policy_{})
  \frac{\pd
  \transMatrix(\policy_{\lambda}^{\this})}{\pd \lambda}\Big|_{\lambda=0}
  \visitMatrix(\policy_{})
  \big)
  \sum_{\pure\in\pures} e_{\rstate,\pure}}
\end{align}
Notice that $ \bracks*{\frac{\pd
  \transMatrix(\policy_{\lambda}^{\this})}{\pd \lambda}\Big|_{\lambda=0}}_{(\rstate^\circ,\pure^\circ)\to(\rstate^\star,\pure^\star)}=\pertrubation(\pure^\star_\play\vert\rstate^\star)\cdot\policy_{-\play}(\pure^\star_{-\play}\vert\rstate^\star)P(\rstate^\star|\rstate^\circ,\pure^\circ)$.
\newmacro{\aux}{\mathrm{aux}}

To simplify notation let us call $\aux_A\defeq \bracks*{\sum_{\rstate'\in\states}\stateDist(\rstate')\sum_{\pure'\in\pures}{\pertrubation(\pure'_\play\vert\rstate')\cdot\policy_{-\play}(\pure'_{-\play}\vert\rstate')}e_{\rstate',\pure'}}$, $\aux_B\defeq
\bracks*{\sum_{\rstate'\in\states}\stateDist(\rstate')\sum_{\pure'\in\pures}\policy_{}(\pure'\vert\rstate')e_{\rstate',\pure'}}
$ and $\aux_C(\rstate)\defeq \sum_{\pure\in\pures} e_{\rstate,\pure}$.

We thus get:
\begin{align}
\mathrm{Term}_A
	&=\abs*{\sum_{\rstate,\pure}
		\frac{\pd\visitMeas{\stateDist}{\policy_{\lambda}^{\this}}(\rstate)}{\pd \lambda}\Big|_{\lambda=0} (\policy_{'}-\policy_{}^{'})(\pure \vert \rstate)Q_{\play}^{\policy_{\rstate'}^{'}}(\rstate,\pure)}
	\notag\\
	&=\abs*{\sum_{\rstate,\pure}\parens*{
		  \aux_A^\top\visitMatrix(\policy_{})\aux_C(\rstate)
  +\aux_B^\top
  \big( 
  \visitMatrix(\policy_{})
  \frac{\pd
  \transMatrix(\policy_{\lambda}^{\this})}{\pd \lambda}\Big|_{\lambda=0}
  \visitMatrix(\policy_{})
  \big)\aux_C(\rstate)}
	(\policy_{}-\policy_{}^{'})(\pure \vert \rstate)Q_{\play}^{\policy_{}^{'}}(\rstate,\pure)}
	\notag\\
	&=\abs*{\parens*{
  \aux_A^\top\visitMatrix(\policy_{})
  +\aux_B^\top
  \big( 
  \visitMatrix(\policy_{})
  \frac{\pd
  \transMatrix(\policy_{\lambda}^{\this})}{\pd \lambda}\Big|_{\lambda=0}
  \visitMatrix(\policy_{})
  \big)}
  \underbrace{\sum_{\rstate,\pure}
  (\policy_{}-\policy_{}^{'})(\pure \vert \rstate)Q_{\play}^{\policy_{}^{'}}(\rstate,\pure)\aux_C(\rstate)}_{\aux_D}}
	\notag\\
	&\leq\norm{\aux_A}_1 \norm{\visitMatrix(\policy_{})\aux_D}_\infty
		+ \norm{\aux_B}_1\norm{
  \big( 
  \visitMatrix(\policy_{})
  \frac{\pd
  \transMatrix(\policy_{\lambda}^{\this})}{\pd \lambda}\Big|_{\lambda=0}
  \visitMatrix(\policy_{})
  \big)\aux_D}_{\infty}
\label{eq:final-bound-term-a}
\end{align}
It is easy to see that $\norm{\aux_A}_1\leq \sqrt{\nPures_{\play}}$, $\norm{\aux_B}_1=1$. Indeed,
\begin{align}
\norm{\aux_A}_1
	&=\sum_{\rstate'\in\states}\stateDist(\rstate')\sum_{\pure'\in\pures}{|\pertrubation(\pure'_\play\vert\rstate')|\cdot\policy_{-\play}(\pure'_{-\play}\vert\rstate')}
	= \sum_{\rstate'\in\states}\stateDist(\rstate')\sum_{\pure_{\play}'\in\pures_\play}|\pertrubation(\pure'_\play\vert\rstate')|
	\notag\\
	&= \sum_{\rstate'\in\states}\stateDist(\rstate')\norm{\pertrubation_{\play\vert\rstate'}}_1
	\leq \sum_{\rstate'\in\states}\stateDist(\rstate')\sqrt{\nPures_{\play}}\norm{\pertrubation_{\play\vert\rstate'}}_2\leq \sqrt{\nPures_{\play}}
	\notag\\
	\norm{\aux_B}_1
	&= \sum_{\rstate'\in\states}\stateDist(\rstate')\sum_{\pure'\in\pures}\policy_{}(\pure'\vert\rstate')=1
\end{align}
Additionally by Conversion Lemma in Matrix form (See \cref{lem:conversion-matrix}), we have that:
\begin{equation}\label{eq:bound-on-T}
\norm{\visitMatrix(\policy_{})x}_\infty=\max_{\rstate,\pure}
|e_{\rstate,\pure}^\top\visitMatrix(\policy_{})x|=
\max_{\rstate,\pure}
|  \ex_{\traj \sim \MDP}\bracks*{\insum_{\time =\tstart}^{\stopTime\parens{\traj}}x(\rstate_\time, \pure_\time )\vert \pure_\tstart=\pure,\rstate_\tstart=\rstate}|\leq \frac{1}{\stopPr}\norm{x}_\infty
\end{equation}
Similarly, for the matrix $  \frac{\pd
  \transMatrix(\policy_{\lambda}^{\this})}{\pd \lambda}\Big|_{\lambda=0}$, we have that
\begin{align}
\label{eq:bound-on-P}
\norm{\frac{\pd \transMatrix(\policy_{\lambda}^{\this})}{\pd \lambda}\Big|_{\lambda=0}x}_\infty&=\max_{\rstate,\pure}
	\abs*{e_{\rstate,\pure}^\top\frac{\pd\transMatrix(\policy_{\lambda}^{\this})}{\pd \lambda}\Big|_{\lambda=0}x}
  \notag\\
	&= \max_{\rstate,\pure} | \sum_{\rstate',\pure'}  \pertrubation(\pure'_\play\vert\rstate')\cdot\policy_{-\play}(\pure'_{-\play}\vert\rstate')P(\rstate'|\rstate,\pure) x_{\rstate',\pure'} |
	\notag\\
	&\leq \sum_{\rstate',\pure'}
		\abs*{\pertrubation(\pure'_\play\vert\rstate')}
			\cdot\policy_{-\play}(\pure'_{-\play}\vert\rstate')P(\rstate'|\rstate,\pure)
	\notag\\
	&\leq \sqrt{\nPures_{\play}}\norm{\pertrubation_{\play\vert\rstate'}}_2\norm{x}_\infty
	\leq\sqrt{\nPures_{\play}}\norm{x}_\infty
\end{align}
since $\norm{\pertrubation}_2=1$. 
Then, using \eqref{eq:bound-on-P} and \eqref{eq:bound-on-T} in \eqref{eq:final-bound-term-a} we get that :
\begin{align}
\mathrm{Term}_A
	&\leq \frac{\sqrt{\nPures_{\play}}}{\stopPr}\norm{\aux_D}_\infty
		+ \frac{\sqrt{\nPures_{\play}}}{\stopPr^2}\norm{\aux_D}_\infty
	\notag\\
	&\leq \frac{\sqrt{\nPures_{\play}}}{\stopPr}(1+\frac{1}{\stopPr})
		\norm*{\sum_{\rstate,\pure}(\policy_{}-\policy_{}^{'})(\pure \vert \rstate)Q_{\play}^{\policy_{}^{'}}(\rstate,\pure)\aux_C(\rstate)}_\infty
	\notag\\
	&\leq \frac{\sqrt{\nPures_{\play}}}{\stopPr^2}(1+\frac{1}{\stopPr})\max_{\rstate}\abs*{\sum_{\pure}
  		(\policy_{}-\policy_{}^{'})(\pure \vert \rstate)}\norm{\aux_C(\rstate)}_\infty
	\notag\\
	&\leq \frac{\sqrt{\nPures_{\play}}}{\stopPr^2}(1+\frac{1}{\stopPr})
		\sum_{\playalt=1}^{\nPlayers}\sqrt{\nPures_{\play}}\norm{\policy_\playalt-\policy'_\playalt}
	\notag\\
	&\leq \frac{\sqrt{\nPures_{\play}}}{\stopPr^3}\sum_{\playalt=1}^{\nPlayers}
		\sqrt{\nPures_{\play}} \cdot \norm{\policy_\playalt-\policy'_\playalt}
\end{align} 
where we used above the fact that $Q$ function is bounded by $1/\stopPr$, $\norm{\pertrubation}=1$ and \cref{prop:useful-dtv} to bound the difference of the policy profiles.

Moving forward, for $\mathrm{Term}_B$, we have:
\begin{align}
\mathrm{Term}_B
	&=\abs*{  \sum_{\rstate,\pure}
\visitDistr{\stateDist}{\policy_{}}(\rstate)\frac{\pd(\policy_{\lambda}^{\this}-\policy_{\lambda}^{\that})(\pure \vert \rstate)}{\pd \lambda}\Big|_{\lambda=0}Q_{\play}^{\policy_{}^{'}}(\rstate,\pure)}
	\notag\\
	&=\abs*{  \sum_{\rstate,\pure} \visitDistr{\stateDist}{\policy_{}}(\rstate)\pertrubation(\pure_\play\vert\rstate)(\policy_{-\play}-\policy_{-\play}^{'})(\pure \vert \rstate)Q_{\play}^{\policy_{}^{'}}(\rstate,\pure)}
	\notag\\
	&\leq\frac{1}{\stopPr}\abs*{  \sum_{\rstate}
\visitDistr{\stateDist}{\policy_{}}(\rstate)\sum_{\pure_\play}\pertrubation(\pure_\play\vert\rstate)\sum_{\pure_{-\play}}(\policy_{-\play}-\policy_{-\play}^{'})(\pure \vert \rstate)}
	\notag\\
	&\leq\frac{1}{\stopPr} \sum_{\rstate}\abs*{
\visitDistr{\stateDist}{\policy_{}}(\rstate)\max_{\rstate}\sum_{\pure_\play}|\pertrubation(\pure_\play\vert\rstate)|\sum_{\pure_{-\play}}(\policy_{-\play}-\policy_{-\play}^{'})(\pure \vert \rstate)}
	\notag\\
	&\leq\frac{1}{\stopPr} \max_{\rstate}\norm{\pertrubation_{\play\vert\rstate}}_1 \sum_{\rstate}
\visitDistr{\stateDist}{\policy_{}}(\rstate)\sum_{\pure_{-\play}}|(\policy_{-\play}-\policy_{-\play}^{'})(\pure \vert \rstate)|
	\notag\\
	&\leq\frac{\sqrt{\nPures_{\play}}}{\stopPr} \max_{\rstate}\norm{\pertrubation_{\play\vert\rstate}}_2 \big(\sum_{\rstate}
\visitDistr{\stateDist}{\policy_{}}(\rstate)\sum_{\pure_{-\play}}|(\policy_{-\play}-\policy_{-\play}^{'})(\pure \vert \rstate)|\big)
	\notag\\
	&\leq \frac{\sqrt{\nPures_{\play}}}{\stopPr}\sum_{\playalt\in\players\setminus\{\play\}}\sqrt{\nPures_{\play}}\norm{\policy_\playalt-\policy'_\playalt}
\leq\frac{\sqrt{\nPures_{\play}}}{\stopPr}\sum_{\playalt=1}^{\nPlayers}\sqrt{\nPures_{\play}}\norm{\policy_\playalt-\policy'_\playalt}
\end{align}
where we used again the fact that $Q$ function is bounded by $1/\stopPr$ and \cref{prop:useful-dtv} to bound the difference of the policy profiles.

Finally, for $\mathrm{Term}_C$, we have:
\begin{align}
\mathrm{Term}_C
	&=\abs*{\sum_{\rstate,\pure}
  \visitDistr{\stateDist}{\policy_{}}(\rstate)(\policy_{}-\policy_{}^{'})(\pure \vert \rstate)\frac{\pd Q_{\play}^{\policy_{\lambda}^{\that}}(\rstate,\pure)
	}{\pd \lambda}\Big|_{\lambda=0}}
	\notag\\
	&\leq \max_{\rstate,\pure}\abs*{\frac{\pd Q_{\play}^{\policy_{\lambda}^{\that}}(\rstate,\pure)}{\pd \lambda}\Big|_{\lambda=0}} \sum_{\rstate,\pure}
	\visitDistr{\stateDist}{\policy_{}}(\rstate)\abs*{(\policy_{}-\policy_{}^{'})(\pure \vert \rstate)}
	\notag\\
	&\leq \max_{\rstate,\pure}\abs*{\frac{\pd Q_{\play}^{\policy_{\lambda}^{\that}}(\rstate,\pure)
	}{\pd \lambda}\Big|_{\lambda=0}} \sum_{\playalt=1}^{\nPlayers}\sqrt{\nPures_{\playalt}}\norm{\policy_\playalt-\policy'_\playalt}
	\notag\\
 &\leq \max_{\rstate,\pure}\abs*{e_{\rstate,\pure}^\top
 \frac{\pd \visitMatrix(\policy_{\lambda}^{\that})}{\pd \lambda}\Big|_{\lambda=0}\reward_\play}\sum_{\playalt=1}^{\nPlayers}\sqrt{\nPures_{\playalt}}\norm{\policy_\playalt-\policy'_\playalt}
	\notag\\
 &\leq \max_{\rstate,\pure}\abs*{e_{\rstate,\pure}^\top
  \frac{\pd
  (I-\transMatrix(\policy_{\lambda}^{\this}))^{-1}}{\pd \lambda}\Big|_{\lambda=0}\reward_\play}\sum_{\playalt=1}^{\nPlayers}\sqrt{\nPures_{\playalt}}\norm{\policy_\playalt-\policy'_\playalt}
	\notag\\
 &\leq \max_{\rstate,\pure}\abs*{e_{\rstate,\pure}^\top
  \big( 
  \visitMatrix(\policy_{})
  \frac{\pd
  \transMatrix(\policy_{\lambda}^{\this})}{\pd \lambda}\Big|_{\lambda=0}
  \visitMatrix(\policy_{})
  \big)
  \reward_\play}\sum_{\playalt=1}^{\nPlayers}\sqrt{\nPures_{\playalt}}\norm{\policy_\playalt-\policy'_\playalt}
	\notag\\
	&\leq \frac{\sqrt{\nPures_{\playalt}}}{\stopPr^2}\sum_{\playalt=1}^{\nPlayers}\sqrt{\nPures_{\playalt}}\norm{\policy_\playalt-\policy'_\playalt}
\end{align}
using again \eqref{eq:bound-on-P} and \eqref{eq:bound-on-T} and \cref{prop:useful-dtv}.
Thus, we are ready now to bound the gradient per player:
\begin{equation}
\left|\frac{\pd (\rvalue_{\play,\stateDist}(\policy_{\lambda}^{\this})-\rvalue_{\play,\stateDist}(\policy_{\lambda}^{\that}))}{\pd \lambda}\Big|_{\lambda=0}\right|
\leq \mathrm{Term}_A+\normConst{\stateDist}{\policy_{}^{\this}}
(\mathrm{Term}_B+ \mathrm{Term}_C)\leq \frac{3\sqrt{\nPures_{\play}}}{\stopPr^3}\sum_{\playalt=1}^{\nPlayers}\sqrt{\nPures_{\playalt}}\norm{\policy_\playalt-\policy'_\playalt}
\end{equation}
where we recall that $ \normConst{\stateDist}{\policy_{}^{\this}}\leq \frac{1}{\stopPr} $
Since we prove it for an arbitrary perturbation vector $\pertrubation$ for the directional derivative, for the independent player's policy gradient it holds also that:
\begin{equation}\norm{\gvalue_\play(\policy)-\gvalue_{\play}(\policy')}=
\norm{\nabla_\play  (\rvalue_{\play,\stateDist}(\policy_{})- 
\nabla_\play  (\rvalue_{\play,\stateDist}(\policy_{}^{'})} \leq \frac{3\sqrt{\nPures_{\play}}}{\stopPr^3}\sum_{\playalt=1}^{\nPlayers}\sqrt{\nPures_{\playalt}}\norm{\policy_\playalt-\policy'_\playalt}\quad \forall \play\in\players
\end{equation}
Finally for the concatenated gradient operator we get:
\begin{align}
\norm{\gvalue(\policy)-\gvalue(\policy')}
&=\sqrt{\sum_{\play\in\players}\norm{\gvalue_\play(\policy)-\gvalue_{\play}(\policy')}^2}=\sqrt{\sum_{\play\in\players}
\norm{\nabla_\play  (\rvalue_{\play,\stateDist}(\policy_{})- 
\nabla_\play  (\rvalue_{\play,\stateDist}(\policy_{}^{'})}^2}
	\notag\\
&\leq \sqrt{\sum_{\play\in\players} \frac{9\nPures_{\play}}{\stopPr^6}\big(\sum_{\playalt\in \players}\sqrt{\nPures_{\playalt}}\norm{\policy_\playalt-\policy'_\playalt}\big)^2}
\leq \sqrt{\sum_{\play\in\players} \frac{9\nPures_{\play}}{\stopPr^6}\sum_{\playalt\in \players}\nPures_{\playalt}\sum_{\playalt\in \players}\norm{\policy_\playalt-\policy'_\playalt}^2}
	\notag\\
&\leq \frac{3}{\stopPr^3}\sqrt{(\sum_{\play\in \players}\nPures_{\play})^2\norm{\policy-\policy'}^2}\leq \frac{3\nPures}{\stopPr^3}\norm{\policy-\policy'}
\end{align}
which completes our proof.
\end{proof}

\begin{remark}
In the proof above, we considered perturbations that may formally lie outside the game's policy space.
However, it is not difficult to see that for sufficiently small $\lambda$ both $V_i({\pi(\lambda)}),\nabla V_i({\pi(\lambda)}), Z_\rho^{\pi(\lambda)}$ are well-defined and bounded, for $ \policy_{\lambda}^{\this}(\pure\vert\rstate)=(\policy_\play+\lambda\pertrubation,\policy_{-\play},\policy_{-\play})$.
In view of this, we may harmlessly assume that all functions considered above are defined in an open neighborhood of the players' policy space.
\end{remark}

%% file: App-Statistics.tex

In this appendix, we prove the two fundamental properties of the \reinforce Policy Gradient estimator that we stated in \cref{lem:reinforce}, namely:
\begin{enumerate}
[\itshape a\upshape)]
\item
\reinforce is an unbiased estimator of $\gvalue(\policy)$.
\item
The variance of \reinforce is bounded from above by $\bigoh\parens{1/\min_{\rstate\in\states,\pure_{\play}\in\pures_{\play}} \policy_\play(\pure_\play \vert \rstate)}$ for each $\play\in\players$.
\end{enumerate}

We recall here that $\nabla_{\play}$ denotes the gradient of the quantity in question with respect to $\policy_{\play}$, \ie when $\policy_{-\play}$ is kept fixed and only $\policy_{\play}$ is varied. 
For concision, we will write $\gvalue_{\play}(\policy) = \nabla_{\play} \rvalue_{\play,\stateDist}(\policy)$ for the individual gradient of player $\play$'s value function, and $\gvalue(\policy) = (\gvalue_{\play}(\policy))_{\play\in\players}$ for the ensemble thereof.

With all this said and done, we begin by restating \cref{lem:reinforce} for convenience:

\ReinforceLemma*
\begin{proof}
Without loss of generality let's assume that $\MDP\equiv\MDP\parens{\policy\ |\ \stateDist}$ for some initial state distribution $\stateDist$.
Additionally, we denote $\prob^\policy\parens{\traj}$ the induced probability of a random trajectory $\traj = (\rstate_{\time}, \pure_{\time}, \realReward_{\time})_{\time \leq\stopTime(\traj)}$.
\begin{align}
\ex_{\traj \sim \MDP}\bracks{\logvalue_\play}	&= \ex_{\traj \sim \MDP}\bracks{\rewards_\play\parens{\traj}\cdot\logtrick_\play\parens{\traj}}=\sum_{\traj\in\mathcal{T}}\prob^\policy\parens{\traj}\rewards_\play\parens{\traj}\cdot\logtrick_\play\parens{\traj}
	\notag\\
	&= \sum_{\traj\in\mathcal{T}}\prob^\policy\parens{\traj}\rewards_\play\parens{\traj}\cdot\bracks{\sum_{\time =\tstart}^{\stopTime\parens{\traj}} \nabla_{\play}\parens{\log\policy_\play\parens{\action_{\play,\time}|\rstate_\time}}}
	\notag\\
	&=\sum_{\traj\in\mathcal{T}} \prob^\policy\parens{\traj} \rewards_\play\parens{\traj}
		\cdot \nabla_{\play}\bracks*{\sum_{\time =\tstart}^{\stopTime\parens{\traj}} {\log\policy_\play\parens{\action_{\play,\time}|\rstate_\time}}}
	\notag\\
	&=\sum_{\traj\in\mathcal{T}}
		\prob^\policy\parens{\traj} \rewards_\play\parens{\traj} \nabla_{\play}
		\sum_{\time=\tstart}^{\stopTime\parens{\traj}} {\log\policy_\play\parens{\action_{\play,\time}|\rstate_\time}}
	\notag\\
	&\qquad
		+ \sum_{\traj\in\mathcal{T}} \prob^\policy\parens{\traj}
  \rewards_\play\parens{\traj} \parens*{ \nabla_{\play} \sum_{\playalt\neq\play} \sum_{\time=\tstart}^{\stopTime\parens{\traj}}
  \log\policy_{\playalt}\parens{\pure_{\playalt,\time}|\rstate_\time}
    +   \nabla_{\play}\sum_{\time =\tstart}^{\stopTime\parens{\traj}} \log\probof{\rstate_\time\given\rstate_{\time -1 }, \action_{\time -1}}}
	\notag\\
	&\qquad
		+ \sum_{\traj\in\mathcal{T}} \prob^\policy\parens{\traj} \rewards_\play\parens{\traj} \nabla_{\play} \log \stateDist(\rstate_\tstart)
	\notag\\
	&= \sum_{\traj\in\mathcal{T}}\prob^{\policy}\parens{\traj} \rewards_\play\parens{\traj}\nabla_{\play}\parens{\log \prob^{\policy}\parens{\traj}}
	\notag\\
	&=  \sum_{\traj\in\mathcal{T}} \parens{\nabla_{\play} \prob^{\policy}\parens{\traj}}\rewards_\play\parens{\traj}= \nabla_{\play}\parens{\sum_{\traj\in\mathcal{T}}\prob^{\policy}\parens{\traj}\rewards_\play\parens{\traj}}
	=\nabla_{\play}\rvalue_{\play,\stateDist}\parens{\policy}
\end{align}
where in the penultimate inequality we used the definition for the derivative of the logarithm. We also note here that
\begin{equation}
  \ex_{\traj\sim \MDP }\bracks{ \logvalue_\play} = \ex_{\traj\sim \MDP}\bracks{\rewards_\play\parens{\traj} \nabla_{\play}\parens{\log\prob^{\policy}\parens{\traj}}}
\end{equation}

For the variance of \reinforce estimator we have that
\begin{equation}
\begin{aligned}
    \ex_{\traj \sim \MDP} \bracks*{\norm{\reinforce_\play(\policy)-\gvalue_\play(\policy)}^2}=&
\ex_{\traj \sim \MDP} \bracks*{\norm{\reinforce_\play(\policy)}^2}
	\notag\\
&-2\ex_{\traj \sim \MDP} \bracks*{\braket{\reinforce_\play(\policy)}{\gvalue_\play(\policy)}}
	\notag\\
&+\ex_{\traj \sim \MDP} \bracks*{\norm{\gvalue_\play(\policy)}^2}
\end{aligned}
\end{equation}
or equivalently
$    \ex_{\traj \sim \MDP} \bracks*{\norm{\reinforce_\play(\policy)-\gvalue_\play(\policy)}^2}=
\ex_{\traj \sim \MDP} \bracks*{\norm{\reinforce_\play(\policy)}^2}-\ex_{\traj \sim \MDP} \bracks*{\norm{\gvalue_\play(\policy)}^2}
$.
Therefore, we have that
\begin{align}
\ex_{\traj \sim \MDP} \bracks*{\norm{\reinforce_\play(\policy)-\gvalue_\play(\policy)}^2}\le\ex_{\traj \sim \MDP} \bracks*{\norm{\reinforce_\play(\policy)}^2}=\exof{\dnorm{\logvalue_\play}^2}
\end{align}
and, after a series of \textendash\ tedious but otherwise straightforward \textendash\ calculations, we get:
\begin{align}
    \exof{\dnorm{\logvalue_\play}^2} &= \ex_{\traj \sim\MDP }\bracks{\dnorm{\rewards_\play\parens{\traj}\logtrick_\play\parens{\traj}}^2}\leq \ex_{\traj \sim\MDP }\bracks{\dnorm{\rewards_\play\parens{\traj}}^2\dnorm{\logtrick_\play\parens{\traj}}^2}
	\notag\\
  	&\leq \ex_{\traj \sim\MDP}\bracks{\parens{\stopTime\parens{\traj}+1}^2\dnorm{\sum_{\time = \tstart}^{\stopTime\parens{\traj}}\nabla_{\play}\log \policy_\play\parens{\action_{\play,\time},\rstate_{\time}}}^2}
	\notag\\ 
  	&\leq \ex_{\traj \sim\MDP }\bracks{\parens{\stopTime\parens{\traj}+1}^3\sum_{\time =\tstart}^{\infty} \sum_{\rstate,\action\in\states\times\actions_\play}\mathds{1}\braces{\time\leq \stopTime}\mathds{1}\braces{\rstate_\time = \rstate,\action_{\play,\time} =\action}\dnorm{\nabla_{\play}\log \policy_\play\parens{\action,\rstate}}^2}
	\notag\\
  	&= \sum_{\time =\tstart}^{\infty} \sum_{\rstate,\action\in\states\times\actions_\play}\ex_{\traj \sim\MDP }\bracks{\parens{\stopTime\parens{\traj}+1}^3\mathds{1}\braces{\time\leq \stopTime}\mathds{1}\braces{\rstate_\time = \rstate,\action_{\play,\time} =\action}\dfrac{1}{\parens{\policy_\play\parens{\action,\rstate}}^2}}
	\notag\\
  	&\leq \sum_{\time =\tstart}^{\infty} \sum_{\rstate,\action\in\states\times\actions_\play} \dfrac{1}{\parens{\policy_\play\parens{\action,\rstate}}^2}\ex_{\traj \sim\MDP }\bracks{\parens{\stopTime\parens{\traj}+1}^3\mathds{1}\braces{\time\leq \stopTime}\mathds{1}\braces{\rstate_\time = \rstate,\action_{\play,\time} =\action}}
	\notag\\
  	&\leq \sum_{\time =\tstart}^{\infty} \sum_{\rstate,\action\in\states\times\actions_\play} \dfrac{1}{\policy_\play\parens{\action,\rstate}}\ex_{\traj \sim\MDP }\bracks{\parens{\stopTime\parens{\traj}+1}^3\mathds{1}\braces{\time\leq \stopTime}\mathds{1}\braces{\rstate_\time = \rstate}}
	\notag\\
  	&\leq \sum_{\time =\tstart}^{\infty} \sum_{\rstate,\action\in\states\times\actions_\play} \dfrac{1}{\kappa_\play}\braces{\parens{\stopTime\parens{\traj}+1}^3\mathds{1}\braces{\time\leq \stopTime}\mathds{1}\braces{\rstate_\time = \rstate}}
	\notag\\
  	&= \sum_{\time =\tstart}^{\infty} \sum_{\rstate\in\states} \dfrac{\abs{\nPures_\play}}{\kappa_\play}\ex_{\traj \sim\MDP }\bracks{\parens{\stopTime\parens{\traj}+1}^3\mathds{1}\braces{\time\leq \stopTime}\mathds{1}\braces{\rstate_\time = \rstate}}
	\notag\\
  	&= \dfrac{\abs{\nPures_\play}}{\kappa_\play}\ex_{\traj \sim\MDP }\bracks{\parens{\stopTime\parens{\traj}+1}^3\sum_{\time =\tstart}^{\stopTime}\mathds{1}\braces{\time\leq \stopTime}}
	\notag\\
  	&\leq \dfrac{\abs{\nPures_\play}}{\kappa_\play}\ex_{\traj \sim\MDP }\bracks{\parens{\stopTime\parens{\traj}+1}^4}
	\notag\\
  	&\leq \dfrac{\abs{\nPures_\play}}{\kappa_\play}\sum_{\time =\tstart}^\infty\parens{1-\stopPr}^\time\stopPr\parens{\time+1}^4\leq \dfrac{24}{\stopPr^4}\dfrac{\abs{\nPures_\play}}{\kappa_\play}
\end{align}
where, to go from the first to the second inequality we used the boundeness by one of the rewards, while from the second to the third, we used Jensen's inequality.
\end{proof}

%% file: App-Solutions.tex


In this last appendix, we proceed to establish three important facts regarding the gradient characterization of stationary Nash policies.
More precisely, we prove the following:
\begin{itemize}
\item
In \cref{lem:GDP}, we prove the crucial property of Gradient Dominance for the multi-agent random stopping setting.
\item
In \cref{lem:fos2ne}, we establish that any stationary point corresponds to Nash Equilibria.
\item
In \cref{prop:var}, we prove the ``drift'' inequalities for all the different types of stationary points. 
\end{itemize}

We begin with the gradient dominance property of the game, which we restate below for convenience:

\GradientDominance*

\begin{proof}
We start by rewriting the LHS of \eqref{eq:GDP} using \cref{lem:performance-lemma,lem:conversion} for $\policy^{\this}=(\policy_{\play}';\policy_{-\play})$ and 
$\policy^{\that}=(\policy_{\play};\policy_{-\play})$:
 \begin{align}
 \rvalue_{\play,\stateDist}\parens{\policy^{\this}} - \rvalue_{\play,\stateDist}\parens{\policy^{\that}}    
 &= \sum_{\rstate \in \states}\visitMeas{\stateDist}{\policy^{\this}}(\rstate) \ex_{\pure \sim \policy^{\this}(\cdot \mid \rstate)}\left[A_{\play}^{\policy_\play^{\that}}(\rstate,\pure)\right]
	\notag\\
 &= \sum_{\rstate \in \states}\visitMeas{\stateDist}{\policy^{\this}}(\rstate) \sum_{\action_\play\in\actions_\play} \policy'_\play(\action_\play|\rstate) \sum_{\action_{-\play}\in\actions_{-\play}} \policy_{-\play}(\action_{-\play}|\rstate) A_{\play}^{\policy^{\that}}(\rstate,\pure)
	\notag\\
 &= \sum_{\rstate \in \states}\visitMeas{\stateDist}{\policy^{\this}}(\rstate) \sum_{\action_\play\in\actions_\play} \policy'_\play(\action_\play|\rstate) \overline{A}_\play^{\policy^{\that}} (\rstate,\action_\play)
	\notag\\
&\leq \sum_{\rstate \in \states}\visitMeas{\stateDist}{\policy^{\this}}(\rstate) \sum_{\action_\play\in\actions_\play}\policy'_\play(\action_\play |\rstate ) \max_{\action_\play\in\actions_\play}  \overline{A}_\play^{\policy^{\that}} (\rstate,\action_\play)
\end{align}
Thus, by a series of direct calculations, we obtain:
\begin{align}
 \rvalue_{\play,\stateDist}\parens{\policy^{\this}} - \rvalue_{\play,\stateDist}\parens{\policy^{\that}}    
&\leq 
\max_{\tilde{\policy}_\play\in\simplex(\actions)^{\states}}\sum_{\rstate \in \states} {\visitMeas{\stateDist}{\policy^{\this}}\parens{\rstate}} \sum_{\action_\play\in\actions_\play}\tilde{\policy}_\play(\action_\play|\rstate)\overline{A}_\play^{\policy^{\that}} (\rstate,\action_\play)
	\notag\\
&\leq 
\max_{\tilde{\policy}_\play\in\simplex(\actions)^{\states}}\sum_{\rstate \in \states} {\visitMeas{\stateDist}{\policy^{\this}}\parens{\rstate}} \sum_{\action_\play\in\actions_\play}\parens{\tilde{\policy}_\play(\action_\play|\rstate)-{\policy}_\play(\action_\play|\rstate)}\overline{A}_\play^{\policy^{\that}} (\rstate,\action_\play)
	\notag\\
&\leq 
\max_{\tilde{\policy}_\play\in\simplex(\actions)^{\states}}\sum_{\rstate \in \states} {\frac{\visitMeas{\stateDist}{\policy^{\this}}\parens{\rstate}}{\visitMeas{\stateDist}{\policy^{\that}}\parens{\rstate}}\visitMeas{\stateDist}{\policy^{\that}}\parens{\rstate}} \sum_{\action_\play\in\actions_\play}\parens{\tilde{\policy}_\play(\action_\play|\rstate)-{\policy}_\play(\action_\play|\rstate)}\overline{A}_\play^{\policy^{\that}} (\rstate,\action_\play)
	\notag\\
&\leq \norm*{\dfrac{\visitMeas{\stateDist}{\policy^{\this}}(\rstate)}{\visitMeas{\stateDist}{\policy^{\that}}\parens{\rstate}}}_\infty \max_{\tilde{\policy}_\play\in\simplex(\actions)^{\states}}\sum_{\rstate \in \states} \sum_{\action_\play\in\actions_\play} {\visitMeas{\stateDist}{\policy^{\that}}\parens{\rstate}}  (\tilde{\policy}_\play(\action_\play|\rstate) - \policy_\play(\action_\play|\rstate))\overline{Q}_\play^{\policy^{\that}} (\rstate,\action_\play)
	\notag\\
&\leq \norm*{\dfrac{\visitMeas{\stateDist}{\policy^{\this}}(\rstate)}{\visitMeas{\stateDist}{\policy^{\that}}\parens{\rstate}}}_\infty \max_{\tilde{\policy}_\play\in\simplex(\actions)^{\states}}\sum_{\rstate \in \states,\action_\play\in\actions_\play} { (\tilde{\policy}_\play(\action_\play|\rstate) - \policy_\play(\action_\play|\rstate))\visitMeas{\stateDist}{\policy^{\that}}\parens{\rstate}} \overline{Q}_\play^{\policy^{\that}} (\rstate,\action_\play)
	\notag\\
&\leq \norm*{\dfrac{\visitMeas{\stateDist}{\policy^{\this}}(\rstate)}{\visitMeas{\stateDist}{\policy^{\that}}\parens{\rstate}}}_\infty \max_{\tilde{\policy}_\play\in\simplex(\actions)^{\states}}\sum_{\rstate \in \states,\action_\play\in\actions_\play}  (\tilde{\policy}_\play(\action_\play|\rstate) - \policy_\play(\action_\play|\rstate))\frac{\partial \rvalue_{\play,\stateDist}(\policy)}{\partial \policy_\play(\pure_\play \mid \rstate)}
\end{align}
so
\begin{align}
\rvalue_{\play,\stateDist}(\policyalt_{\play};\policy_{-\play}) -
	\rvalue_{\play,\stateDist}(\policy_{\play};\policy_{-\play})
	&\leq \mismatch \max_{\bar\policy_{\play}\in\policies_{\play}}
		\braket{\nabla_{\play}\rvalue_{\play,\stateDist}(\policy)}{\bar{\policy}_{\play}-{\policy_{\play}}}
\end{align}
and our proof is complete.
\end{proof}

\begin{remark*}
Notice that we have assumed that $\visitMeas{\stateDist}{\policy^{\that}}>0$. If this wasn't the case we could take a trivial bound of $\infty$.
\end{remark*}

We now proceed to establish the link between \eqref{eq:FOS} and \eqref{eq:NP}:

\FOSNE*

\begin{proof}
By the definition of first-order stationarity applied to the policies $\policy^{\ast}$ and $\policy$, it is straightforward to check that
$\braket{\gvalue(\policy^*)}{\policy^*-\policy} \geq 0$
if and only if
$\max_{\bar\policy_{\play}\in\policies_{\play}}\braket{\nabla_{\play}\rvalue_{\play,\stateDist}(\policy^*)}{\policy_\play-\bar\policy_\play^*} \leq 0$.
However, by the gradient dominance property established in \cref{lem:GDP}, we readily get
\begin{align}
\rvalue_{\play,\stateDist}(\policy_{\play};\policy_{-\play}^*) -
	\rvalue_{\play,\stateDist}(\policy_{\play}^*;\policy_{-\play}^*)
	&\leq  \mismatch \max_{\bar\policy_{\play}\in\policies_{\play}}
		\braket{\nabla_{\play}\rvalue_{\play,\stateDist}(\policy^*)}{\bar{\policy}_{\play}-{\policy_{\play}^*}}\leq 0
	\\
\shortintertext{and hence}
\rvalue_{\play,\stateDist}(\policy_{\play};\policy_{-\play}^*)
	&\leq \rvalue_{\play,\stateDist}(\policy_{\play}^*;\policy_{-\play}^*)
		\quad
		\text{for all $\policy_{\play}\in \policies_\play$}
\end{align}
and our claim follows.
\end{proof}

With all this in place, we are finally in a position to prove the characterization of \acl{SOS} and strict Nash policies that of \cref{prop:var}.
For ease of reference, we restate the relevant claims below.

\StationaryProperties*

\begin{proof}
We begin with the characterization of \acl{SOS} policies.
To that end, let $\vdim = \abs{\states} \sum_{\play} \abs{\pures_{\play}}$ denote the ambient dimension of $\dspace$ and consider the mapping $\varphi\from\R^{\vdim\times\vdim} \to \R$ mapping $\hmat \mapsto \max\setdef{\tvec^{\top}\hmat\tvec}{\tvec \in \tcone(\eq), \norm{\tvec} = 1}$.
Clearly, $\varphi$ is convex as the pointwise maximum of a set of linear \textendash\ and hence convex \textendash\ functions.
This in turn implies the continuity of $\varphi$ as every convex function is continuous on the interior of its effective domain.
Since $\eq$ satisfies \eqref{eq:SOS} by assumption, we have $\varphi(\Jac_{\gvalue}(\eq)) < 0$, so, by continuity and the convexity of $\policies$, there exists some $\strong > 0$ and a convex neighborhood $\nhd$ of $\eq$ in $\policies$ such that $\varphi(\Jac_{\gvalue}(\policy)) \leq -\strong$ for all $\policy \in \nhd$.

With this in mind, letting $\tvec = \policy-\eq \in \tcone(\eq)$ for some $\policy\in\nhd$, a straightforward Taylor expansion with integral remainder yields
\begin{equation}
\gvalue(\policy) - \gvalue(\eq)
	= \int_{0}^{1} \Jac_{\gvalue}(\eq + \tau\tvec) \tvec \dd\tau
\end{equation}
and hence, setting $\policy_{\tau} = \eq + \tau\tvec$, we get
\begin{align}
\braket{\gvalue(\policy) - \gvalue(\eq)}{\policy - \eq}
	&= \int_{0}^{1} \tvec^{\top} \Jac_{\gvalue}(\policy_{\tau}) \tvec \dd\tau
	\notag\\
	&\leq \norm{\tvec}^{2}
		\int_{0}^{1} \varphi(\Jac_{\gvalue}(\policy_{\tau})) \dd\tau
	\leq -\strong \norm{\tvec}^{2}
	= -\strong \norm{\policy - \eq}^{2}
\end{align}
However, by \eqref{eq:FOS}, we have $\braket{\gvalue(\eq)}{\policy - \eq} \leq 0$ which, combined with the above, yields $\braket{\gvalue(\policy)}{\policy - \eq} \leq -\strong \norm{\policy - \eq}^{2}$, as claimed.

For the second part of our lemma, pick some $\policy\neq\eq$ and let $\tvec = (\policy-\eq) / \norm{\policy-\eq}$, so $\tvec \in \tcone(\eq)$ and $\norm{\tvec} = 1$.
Then, given that \eqref{eq:FOS} is satisfied as a strict inequality for all $\policy\neq\eq$, we readily get $\braket{\gvalue(\eq)}{\tvec} < 0$ for all $\tvec\in\tcone(\eq)$ with $\norm{\tvec} = 1$.
Thus, by the joint continuity of the function $\braket{\gvalue(\policy)}{\tvec}$ in $\policy$ and $\tvec$, there exists a compact convex neighborhood $\cpt$ of $\eq$ in $\policies$ such that $\strong \defeq \min\setdef{\braket{\gvalue(\policy)}{\tvec}}{\policy\in\cpt,\tvec\in\tcone(\eq),\norm{\tvec}=1} < 0$.
Thus, letting $\tvec = (\policy-\eq) / \norm{\policy-\eq}$ as above, we conclude that $\braket{\gvalue(\policy)}{\policy - \eq} \leq -\strong \norm{\policy - \eq}$, as claimed.
\end{proof}

%% file: Thanks.tex
%
%
Part of this work was done while the authors were visiting the Simons Institute for the Theory of Computing.
P.~Mertikopoulos gratefully acknowledges financial support by
the French National Research Agency (ANR) in the framework of
the ``Investissements d'avenir'' program (ANR-15-IDEX-02),
the LabEx PERSYVAL (ANR-11-LABX-0025-01),
MIAI@Grenoble Alpes (ANR-19-P3IA-0003),
and the bilateral ANR-NRF grant ALIAS (ANR-19-CE48-0018-01). K. Lotidis and E. Vlatakis are grateful for financial support by the Onassis Foundation (F ZR 033-1/2021-2022, 010-1/2018-2019).
A. Giannou is grateful for financial support by  ONR: “A Theoretically Principled Framework for Learning by Pruning”.
E. V. Vlatakis-Gkaragkounis is grateful for financial support by the Google-Simons Fellowship, Pancretan Association of America and Simons Collaboration on Algorithms and Geometry. This project was completed while he was a visiting research fellow at the Simons Institute for the Theory of Computing. Additionally, he would like to acknowledge the following series of NSF-CCF grants under the numbers 1763970/2107187/1563155/1814873.